\documentclass[smallcondensed]{svjour3}
\usepackage[utf8]{inputenc}
\usepackage[T1]{fontenc}
\usepackage[english]{babel}

\usepackage{lmodern}
\usepackage{mathtools}
\usepackage{xcolor}
\usepackage{amsmath}

\usepackage{amsfonts}
\usepackage{amssymb}
\usepackage{tikz}
\usepackage{stmaryrd}
\usepackage{hyperref}
\usepackage{comment}
\usepackage{ esint }
\usepackage{tikz}
\usepackage[margin=1in]{geometry}

\newcommand{\R}{\mathbb{R}}
\newcommand{\N}{\mathbb{N}}
\newcommand{\B}{\mathcal{B}}

\newcommand{\C}{\mathcal{C}}
\newcommand{\E}{\mathcal{E}}

\newcommand{\A}{\mathcal{A}}

\newcommand{\QED}{\hfill $\square$}

\newtheorem{assumption}[definition]{Assumption}
\newtheorem{defassumption}[definition]{Definition and assumption}

\numberwithin{equation}{section}
\numberwithin{definition}{section}
\numberwithin{proposition}{section}
\numberwithin{corollary}{section}
\numberwithin{lemma}{section}
\numberwithin{theorem}{section}
\numberwithin{remark}{section}

\title{A general framework for the kinetic modelling of polyatomic gases}

\author{ Thomas Borsoni   \and
          Marzia Bisi\and
        Maria Groppi
}

\institute{Thomas Borsoni \at
               ENS Paris-Saclay, Université Paris-Saclay, Gif-sur-Yvette, France \\
               Università di Parma, Parma, Italia \\
              \email{tborsoni@ens-paris-saclay.fr}           
           \and
         Marzia Bisi \at
             Università di Parma, Parma, Italia \\
             \email{marzia.bisi@unipr.it}
             \and
             Maria Groppi \at
              Università di Parma, Parma, Italia \\
              \email{maria.groppi@unipr.it}
             }

\date{Received: date / Accepted: date}

\begin{document}

\maketitle

\begin{abstract}
A general framework for the kinetic modelling of non-relativistic polyatomic gases is proposed, where each particle is characterized both by its velocity and by its internal state, and the Boltzmann collision operator involves suitably weighted integrals over the space of internal energies. The description of the internal structure of a molecule is kept highly general, and this allows classical and semi-classical models, such as the monoatomic gas description, the continuous internal energy structure, and the description with discrete internal energy levels, to fit our framework. We prove the H-Theorem for the proposed kinetic equation of Boltzmann type in this general setting, and characterize the equilibrium Maxwellian distribution and the thermodynamic number of degrees of freedom. Euler equations are derived, as zero-order approximation in a suitable asymptotic expansion. In addition, within this general framework it is possible to build up new models, highly desirable for physical applications, where rotation and vibration are precisely described. Examples of models for the Hydrogen Fluoride gas are presented.
\end{abstract}

\section*{Acknowledgements}
This study was initiated while T.B. was visiting the Department
of Mathematical, Physical and Computer Sciences,
University of Parma, with the financial support
of the Ecole Normale Superieure Paris-Saclay.
The authors thank the support by the
University of Parma, by the Italian National Group of Mathematical
Physics (GNFM-INdAM), and by the Italian National
Research Project ''{\it Multiscale phenomena in Continuum Mechanics:
singular limits, off-equilibrium and transitions}'' (Prin
2017YBKNCE).

\tableofcontents

\section{Introduction}

The purpose of this paper is to propose a consistent general Boltzmann--type model for a polyatomic gas, able to include the kinetic models already exploited in the pertinent literature, but also to give rise to more accurate descriptions of internal energies of polyatomic molecules. The construction of a reliable mathematical model for polyatomic gaseous particles is highly desirable in view of physical applications, since the main constituents of the atmosphere are polyatomic. A review of possible problems involving monoatomic and polyatomic species, possibly undergoing also chemical reactions, may be found in \cite{Nagnibeda-Kustova,Zhdanov}. For this reason, in last decades the investigation of polyatomic particles has gained interest not only in the frame of kinetic theory, but also in Extended Thermodynamics \cite{Ruggeri-Sugiyama-new} and in the derivation of accurate schemes for fluid-dynamics equations~\cite{Xu-book2013,Toro-book2014}.

In kinetic theory, non--translational degrees of freedom of polyatomic particles are usually described by means of an internal energy variable, that may be assumed discrete or continuous. The basic features of the model with a finite set of discrete energy levels may be found in \cite{groppi1999kinetic,Giovangigli}; the gas is considered as a sort of mixture of monatomic components, each one characterized by its energy level, and particles interact by binary collisions preserving total energy, but with possible exchange of energy between its kinetic and internal (excitation) forms, allowing thus particle transitions from one energy component to another.
The kinetic model based on a continuous internal energy parameter $I$ has been proposed by Borgnakke and Larsen in~\cite{borgnakke1975statistical} and then extensively investigated in \cite{desvillettes1997modele,desvillettes2005kinetic}; the gas distribution function turns out to depend also on this energy variable, and macroscopic fields and Boltzmann operators involve integrations in $\text{d}I$ with an associated integration weight $\varphi(I)$. This measure $\varphi(I)\,\text{d}I$ is a parameter of the model, and different options for it allow to reproduce any desired number of internal degrees of freedom. With the commonly adopted choice $\varphi(I) = I^\alpha$, this number of internal degrees of freedom is independent of the temperature of the gas (describing thus polytropic gases). On the other hand, the set of discrete internal energy levels proposed in \cite{wang1951transport,groppi1999kinetic} allows to obtain a temperature-dependent number of degrees of freedom \cite{Bisi-Spiga-RicMat2017}, with a specific heat at constant volume which resembles the physical laws of statistical mechanics \cite{Landau}.
Both formulations with discrete and continuous energy turn out to be well suited also in presence of simple chemical reactions (see for instance \cite{groppi1999kinetic,desvillettes2005kinetic})  and discrete energy levels are also well adapted to problems involving complex chemistry since real molecules actually have discrete energy levels \cite{Giovangigli}.

Since Boltzmann equations are quite complicated to deal with, simpler kinetic models have been proposed, mainly of BGK or ES-BGK type, both for gases or mixtures with discrete energies \cite{GroSpiPoF2004,BGS-PhysRevE,Bisi-Travaglini} and for the continuous energy description \cite{Struchtrup,Andries-etal-2000,Brull-Schneider2009,Bisi-Monaco-Soares}.
Hydrodynamic limits of kinetic equations have been performed, at Euler or Navier--Stokes accuracy, in order to build up consistent fluid--dynamic equations for the main macroscopic fields, focusing the attention on the transport coefficients, especially on the bulk viscosity or, equivalently, on the so-called dynamical pressure, a correction to the classical scalar pressure due to the polyatomic structure of molecules \cite{Milana-2014,Ruggeri-PhysRevE-2017,Bisi-Ruggeri-Spiga}.
Such kinetic or hydrodynamic equations have been applied to simple physical problems, as the existence of shock-wave profiles \cite{Kosuge-Aoki-2018,Kosuge-Kuo-Aoki,Taniguchi-etal-2016,Pavic-Madjarevic-Simic}, the Poiseuille and thermal creep flows \cite{Funagane-etal,LeiWu-etal}, and suitable boundary conditions have been derived in order to allow the investigation of flow problems with solid boundaries \cite{Hattori-etal}.

The main drawback of these kinetic models is the fact that they wish to describe all internal features of polyatomic molecules by means of a single internal energy parameter (discrete or continuous). In order to improve the applicability of Boltzmann equations to real problems, more accurate models would be needed, able to separate the vibrational degrees of freedom from the rotational ones. More precisely, since the gap between two subsequent discrete
levels is much lower for rotational energy than for vibrational energy \cite{Herzberg}, a reasonable way of modeling could approximate the rotational part by means of a continuous variable, keeping the vibrational part discrete.
This is the main motivation for the present work; indeed, we would like to merge the existing discrete and continuous energy models in a general (abstract) framework, in order to combine their strength and obtain a full description of both rotation and vibration.
Models with discrete vibrational degrees of freedom and continuous rotational states have long been used in the literature in order to analyse vibrational state-to-state non-equilibrium phenomena \cite{Nagnibeda-Kustova,Kal20}. Some attempts to distinguish vibrational from rotational energy have been  performed also at kinetic level: in the discrete energy formulation, this aim could be achieved by choosing discrete energies depending on several discrete indices \cite{Giovangigli}, but also other models have been proposed, of BGK or Fokker-Planck--type \cite{Mathiaud-Mieussens,Dauvois-etal} or in the frame of Extended Thermodynamics~\cite{Ruggeri-vibrot} (involving two continuous energy variables). The idea of introducing two different temperatures (translational and internal) in the kinetic model for a polyatomic gas, investigating then their influence on volume viscosity and other transport coefficients, has been developed in \cite{Nagnibeda-Kustova,BG11,Aoki-Bisi-etal-PRE}. Moreover, a proper rovibrational collision model has been built up in order to study the internal energy excitation and dissociation processes behind a strong shockwave in a nitrogen flow \cite{PMMJ14}.

We propose here a general kinetic framework for the description of a single polyatomic gas, in which each particle is characterized, in addition to its position and velocity, also by a suitable internal state $\zeta$. The space of all admissible internal states $\mathcal{E}$ is endowed with a measure $\mu$.
In other words, instead of considering a continuous or a discrete variable to describe the internal state of a given molecule, we use a single parameter $\zeta$, and integration over the continuous variable or summation over the discrete one are replaced by a single integration over $\zeta$ against the given proper measure $\mu$. Moreover, we clearly distinguish the state of a molecule~$\zeta$ and the energy function $\varepsilon(\zeta)$, which associates to each state a suitable energy value.

The continuous and the discrete energy models may be included into this framework by choosing $\mathcal{E} = \mathbb{R}_+$ and $\mathcal{E} = \mathbb{N}$, respectively. We prove that, even keeping the space $\mathcal{E}$, the measure $\mu$ and the energy function $\varepsilon(\zeta)$ generic (not explicit), the collision rule determining post-collision velocities may be defined, and the corresponding Boltzmann equation fulfills the expected consistency properties. More precisely, the Boltzmann H theorem holds in the generic framework, and collision equilibria turn out to be provided by Maxwellian distributions depending on the internal energy function $\varepsilon(\zeta)$. The technical steps of the proofs rely essentially on the conservation of momentum and of total energy, therefore there is no need of fixing from the beginning the microscopic energy structure of particles.
We remark that our framework describes isotropic gases, where polarization effects due to external fields acting on (non ionized) polyatomic gases may be neglected \cite{MBKK90}.

In more detail, the content of the paper is organized as follows. In section \ref{section:presentationmodel}, we introduce the framework and its assumptions, and we make the collision rules explicit. Our method is inspired by \cite{desvillettes1997modele,desvillettes2005kinetic,bisi2005kinetic}, with some necessary modifications, since the procedure in such references is specifically designed for a continuous or discrete internal variable.  Section \ref{section:boltzmann} focuses then on the Boltzmann collision operator and on some of its main properties, as the weak formulation and the characterisation of collision invariants. In section \ref{section:Maxwellian} we prove the validity of the Boltzmann H-theorem, and we investigate the equilibrium Maxwellian distribution and the thermodynamic number of internal degrees of freedom. In section \ref{section:euler}, the hydrodynamic Euler equations corresponding to our general model are derived. Sections \ref{section:combine} and \ref{section:application} are devoted to modelling: we explain how to recover the monoatomic, continuous and discrete formulations, and we give insights on how to build consistent new models within our framework, including a description able to separate rotation and vibrations of polyatomic particles; examples of possible models are provided and commented on, with reference to the pertinent gas--dynamics literature. In Section 8, we show how is possible to reduce  models fitting the general setting to a one-real-variable description  with a suitable measure, similar to the continuous model with integration weight. We address the problem of computing this measure from the general setting, and show that the combination of a continuous model with weight and a discrete model reduces to a continuous model with a different weight that can be made explicit, allowing an easier treatment for numerical applications. Finally, section \ref{concl} contains some concluding remarks and perspectives.

\section{General setting and collision rules} \label{section:presentationmodel}

\noindent We consider a gas composed of one species of molecules, which mass we denote by $m \in \R_+^*$. Each molecule is microscopically described by its velocity, denoted by $v$, and its internal state, denoted by $\zeta$. The latter is usually related to rotation or vibration that occurs inside the molecule. We assume that there exists a set of all possible internal states $\E$, and that this set is a measure  space. All molecules, being of the same species, are related to the same space $\E$. Secondly, we assume that to each internal state $\zeta \in \E$ corresponds an energy, that is, there exists an internal energy function $\varepsilon : \E \to \R$. The energy of a molecule with velocity $v$ and internal state $\zeta$ is then equal to $\displaystyle \frac{m}{2} |v|^2 + \varepsilon(\zeta)$, the sum of its kinetic and internal parts.

\smallskip

\begin{definition}
We call  {\it space of internal states} a non-empty measure   space $(\E,\A,\mu)$, with $\mu(\E) \neq 0$.
\end{definition}

\noindent To clarify here, $\E$ is the space of internal states, $\A$ is a $\sigma$-algebra on $\E$, and $\mu$ is a (non-negative) measure on $(\E,\A)$. In the following definition, $\text{Bor}(\R)$ is the set of Borelians of $\R$.

\begin{definition}
We call {\it internal energy function} a $(\A,\text{Bor}(\R))$-measurable function
$
\varepsilon : \E \to \R.
$
\end{definition}

\noindent In our framework, we only require two assumptions on $(\E,\A,\mu)$ and $\varepsilon$, which both physically make sense. The first is to assume that $\varepsilon$ admits an infimum.

\begin{defassumption} \label{assumption:minimum}
\emph{We denote by $\varepsilon^{0}$ the essential infimum of $\varepsilon$ in $\Bar{\R}$ on $\E$ under the measure $\mu$, that is}
\begin{equation} \label{eqdef:vareps0}
\varepsilon^{0} := \emph{inf ess}_{\mu} \{ \varepsilon \} = \emph{sup} \, \Big\{ R \in \R, \; \emph{s.t. } \mu(\varepsilon < R) = 0 \Big\},
\end{equation}
\emph{and assume that $\varepsilon^{0} \in \R$}.
\end{defassumption}

\begin{definition}
We define the grounded internal energy function $\Bar{\varepsilon}$ by
\begin{equation} \label{eqdef:groundedenergyfunc}
\Bar{\varepsilon} := \varepsilon - \varepsilon^{0}.
\end{equation}
\end{definition}
Note that by definition of $\varepsilon^0$,  $\text{inf ess}_{\mu} \{ \Bar{\varepsilon}\} = 0$. \smallskip

\begin{remark}
The reader shall remark that in this paper, the value of $\varepsilon^0$ has no importance. Indeed, $\varepsilon^0$ represents the fundamental energy of configuration of the molecule, and it plays a role only when chemical reactions are involved.
\end{remark}

\noindent The second assumption is the well-definedness of the partition function.

\begin{defassumption} \label{assumption:integrability}
\emph{We define the partition function $Z$, for all $\beta > 0$, as}
\begin{equation} \label{eq:partitionfunc}
Z(\beta) := \int_{\E} e^{-\beta \Bar{\varepsilon}(\zeta)} \, \emph{d}\mu(\zeta),
\end{equation}
\emph{and assume that for all $\beta > 0$, $Z(\beta) < \infty$.}
\end{defassumption}

\noindent The choice of $\bar{\varepsilon}$ instead of simply $\varepsilon$ in the definition of $Z$ is explained in Remark \ref{remark:whyisZdefinedwithepsilonbar}. Assumptions \ref{assumption:minimum} and \ref{assumption:integrability} imply that $\forall R \in \R$, $\mu(\varepsilon < R) < \infty$, which itself implies that $\mu$ is $\sigma$-finite.

\begin{proposition} \label{prop:ZisCinfty}
$Z$ is $\mathcal{C}^{\infty}$ on $\R_+^*$ and for all $k \in \N$ and $\beta > 0$ and, denoting its derivatives by $\displaystyle Z^{(k)}(\beta) = \frac{{\rm d}^k Z}{{\rm d} \beta^k}$, we have
\begin{equation}
Z^{(k)}(\beta) =  \int_{\E} (-\Bar{\varepsilon}(\zeta))^k \, e^{-\beta \Bar{\varepsilon}(\zeta)} \, \emph{d}\mu(\zeta).
\end{equation}
\end{proposition}

\begin{proof}
Remarking that for all $x \in \R_+$ and $k \in \N$, $e^{\frac{x}{2}} \geq \frac{\left( \frac{x}{2} \right)^k}{k!}$, we get the inequality $x^k \, e^{-x} \leq 2^k \, k! \, e^{-\frac{x}{2}}$. Since $\Bar{\varepsilon} \geq 0$ $\mu$-a.e., we get for all $k \in \N$  and $\beta > 0$, $\int_{\E} \Bar{\varepsilon}(\zeta)^k \, e^{-\beta \Bar{\varepsilon}(\zeta)} \, \text{d}\mu(\zeta) < +\infty$ and the result comes by dominated convergence.
\QED
\end{proof}


\noindent We consider a binary collision and define the collision rule. The states (velocity and internal state) of the pre-collision molecules are denoted by $(v,\zeta)$ and $(v_*,\zeta_*)$, and the states of the post-collision molecules are denoted by $(v',\zeta')$ and $(v'_*,\zeta'_*)$. In each collision momentum and total energy are conserved

\begin{equation} \label{eq:conservation}
    \begin{cases}
   \hspace{83.5pt}  mv +  mv_* =   mv' +  mv'_* \\
    \displaystyle \frac{m}{2} |v|^2 + \varepsilon(\zeta) + \frac{m}{2} |v_*|^2 + \varepsilon(\zeta_*) = \frac{m}{2} |v'|^2 + \varepsilon(\zeta')  + \frac{m}{2} |v'_*|^2 + \varepsilon(\zeta'_*).
    \end{cases}
\end{equation}

\noindent The parallelogram identity yields

\begin{align*}
    |v-v_*|^2 &= 2 \left( |v|^2 + |v_*|^2 \right) - |v + v_*|^2 \\
    &= |v' - v'_*|^2 + \frac{4}{m} (\varepsilon(\zeta') + \varepsilon(\zeta'_*) - \varepsilon(\zeta) - \varepsilon(\zeta_*)).
\end{align*}

\noindent We set
\begin{equation} \label{eq:def:delta}
\Delta \left(v,v_*,\zeta,\zeta_*,\zeta',\zeta'_* \right) := \frac14 |v-v_*|^2 + \frac{1}{m}( \varepsilon(\zeta) + \varepsilon(\zeta_*) - \varepsilon(\zeta') - \varepsilon(\zeta'_*)),
\end{equation}
we then have
\begin{equation}  \label{eq:epsilon}
\frac14 |v'-v'_*|^2 = \Delta \left(v,v_*,\zeta,\zeta_*,\zeta',\zeta'_* \right).
\end{equation}

\noindent We deduce that the collision can occur only \emph{when $\Delta$ is non-negative}, and in this case, by a classical argument, there exists $\omega \in \mathbb{S}^2$ (two-dimensional unit sphere) such that

\begin{equation} \label{eq:outvelocities}
\begin{cases}
v' &= \displaystyle \frac{v + v_*}{2} + \sqrt{ \Delta} \, T_{\omega}\left[ \frac{v - v_*}{|v-v_*|} \right] \smallskip\\
v'_* &= \displaystyle \frac{v + v_*}{2} - \sqrt{ \Delta} \,T_{\omega}\left[ \frac{v - v_*}{|v-v_*|} \right]
\end{cases}
\end{equation}
where $T_{\omega}$ is the symmetry with respect to $(\R \omega)^{\perp}$,
$$
\forall \, V \in \mathbb{S}^2, \quad T_{\omega} \left[ V \right] := V - 2 (\omega \cdot V) \, \omega \; \in \mathbb{S}^2.
$$
It is not required to extend formula \eqref{eq:outvelocities} to the set $\big\{(v,v)  \; \text{s.t. } v \in \R^3 \big\}$, since it is negligible.

\smallskip

\noindent From the previous computation, it is clear that not any collision is possible. If, before collision, two molecules are in low energy internal states with close velocities, they cannot reach too much high energy internal states  after collision. For this reason, we  fix at first the pre- and post-collision internal states, and then we consider the space of pre-collision velocity pairs that make the collision possible. In the following definition, $E[\zeta,\zeta_*,\zeta',\zeta'_*]$ is the set of pre-collision velocity pairs $(v,v_*)$ so that a collision $(v,\zeta),(v_*,\zeta_*) \rightarrow (\cdot,\zeta'),(\cdot,\zeta'_*)$ is allowed.

\begin{definition}
For all $\zeta,\zeta_*,\zeta',\zeta'_* \in \E$, we define
\begin{equation} \label{eq:setofcoll}
E[\zeta,\zeta_*,\zeta',\zeta'_*] := \Big\{ (v,v_*) \in \R^3 \times \R^3 \text{ s.t. } \Delta \left(v,v_*,\zeta,\zeta_*,\zeta',\zeta'_* \right) \geq 0 \Big\},
\end{equation}
where $\Delta$ is defined by \eqref{eq:def:delta}.
\end{definition}
It can be proved that for all  $\zeta,\zeta_*,\zeta',\zeta'_* \in \E$, the interior of $E[\zeta,\zeta_*,\zeta',\zeta'_*]$ is non-empty, unbounded and arcwise connected.

\smallskip

\noindent Now we can define the transformation linking the pre and post-collision velocities.

\begin{definition}
We define for all $\zeta,\zeta_*,\zeta',\zeta'_* \in \E$ and $\omega \in \mathbb{S}^2$,
\begin{equation} \label{eq:transformoutvelocities}
 S_{\omega}[\zeta,\zeta_*,\zeta',\zeta'_*] :
\begin{cases}
E[\zeta,\zeta_*,\zeta',\zeta'_*] &\to \hspace{100pt} E[\zeta',\zeta'_*,\zeta,\zeta_*] \\
\displaystyle \quad (v,v_* ) &\mapsto \quad \displaystyle  \left(\frac{v + v_*}{2} + \sqrt{ \Delta} \, T_{\omega}\left[ \frac{v - v_*}{|v-v_*|} \right],  \frac{v + v_*}{2} - \sqrt{ \Delta} \,T_{\omega}\left[ \frac{v - v_*}{|v-v_*|} \right]\right).
\end{cases}
\end{equation}

\end{definition}

\begin{lemma} \label{lemmabijection} $S_{\omega}[\zeta,\zeta_*,\zeta',\zeta'_*]$ is well-defined and a bijection, with
$$
\left(S_{\omega}[\zeta,\zeta_*,\zeta',\zeta'_*]\right)^{-1} = S_{\omega}[\zeta',\zeta'_*,\zeta,\zeta_*].
$$
\end{lemma}

\begin{proof} Let $\zeta,\zeta_*,\zeta',\zeta'_* \in \E$, $\omega \in \mathbb{S}^2$. First, we show that the transformation is well-defined. We fix $(v,v_*) \in E[\zeta,\zeta_*,\zeta',\zeta'_*]$. We have that $\Delta(v,v_*,\zeta,\zeta_*,\zeta',\zeta'_*) \geq 0$. We define $(v',v'_*) = S_{\omega}[\zeta,\zeta_*,\zeta',\zeta'_*](v,v_*)$. Then

\begin{align*}
\Delta(v',v'_*,\zeta',\zeta'_*,\zeta,\zeta_*) &= \frac14 |v'-v'_*|^2 + \frac{1}{m}(\varepsilon(\zeta') + \varepsilon(\zeta'_*) - \varepsilon(\zeta) - \varepsilon(\zeta_*)) \\
&= \frac14 \left(2 \sqrt{\Delta(v,v_*,\zeta,\zeta_*,\zeta',\zeta'_*)} \right)^2 + \frac{1}{m}(\varepsilon(\zeta) + \varepsilon(\zeta_*) - \varepsilon(\zeta) - \varepsilon(\zeta_*))\\
&= \Delta(v,v_*,\zeta,\zeta_*,\zeta',\zeta'_*) + \frac{1}{m}(\varepsilon(\zeta') + \varepsilon(\zeta'_*) - \varepsilon(\zeta) - \varepsilon(\zeta_*)) \\
&= \frac14 |v - v_*|^2 \geq 0.
\end{align*}
Thus $(v',v'_*) \in E[\zeta',\zeta'_*,\zeta,\zeta_*]$. The application is well-defined.

\bigskip

\noindent Let us now prove that the transformation $S_{\omega}[\zeta,\zeta_*,\zeta',\zeta'_*]$ is a bijection. We fix $(v',v'_*) \in E[\zeta',\zeta'_*,\zeta,\zeta_*]$ and define  $(v,v_*) = S_{\omega}[\zeta',\zeta'_*,\zeta,\zeta_*](v',v'_*)$. First, we just proved (with inverted notations) that $(v,v_*) \in E[\zeta,\zeta_*,\zeta',\zeta'_*]$ and $\displaystyle \Delta(v,v_*,\zeta,\zeta_*,\zeta',\zeta'_*) = \frac14|v'-v'_*|^2$.  Now we also remark that
$$
v+v_* = v'+v'_*,
$$
and
$$
T_{\omega} \left[ \frac{v-v_*}{|v-v_*|} \right] = T_{\omega} \left[  T_{\omega} \left[ \frac{v'-v'_*}{|v'-v'_*|} \right] \right] = \frac{v'-v'_*}{|v'-v'_*|},
$$
because $T_{\omega}$ is a symmetry. Since $\displaystyle \Delta(v,v_*,\zeta,\zeta_*,\zeta',\zeta'_*) = \frac14|v'-v'_*|^2$, we deduce that
\begin{align*}
\frac{v+v_*}{2} + \sqrt{\Delta(v,v_*,\zeta,\zeta_*,\zeta',\zeta'_*)} \, T_{\omega} \left[ \frac{v-v_*}{|v-v_*|} \right] = \frac{v'+v'_*}{2} +  \frac12 |v'-v'_*| \frac{v'-v'_*}{|v'-v'_*|} = v'.
\end{align*}
Analogously,
$$
\frac{v+v_*}{2} - \sqrt{\Delta(v,v_*,\zeta,\zeta_*,\zeta',\zeta'_*)} \, T_{\omega} \left[ \frac{v-v_*}{|v-v_*|} \right] = \frac{v'+v'_*}{2} -  \frac12 |v'-v'_*| \frac{v'-v'_*}{|v'-v'_*|} = v'_*,
$$
so that $S_{\omega}[\zeta,\zeta_*,\zeta',\zeta'_*](v,v_*) = (v',v'_*)$, which ends the proof.

\QED
\end{proof}

\begin{lemma} \label{lemmajacobian} The Jacobian of the transformation $S_{\omega}[\zeta,\zeta_*,\zeta',\zeta'_*]$ is given by the formula

\begin{equation}\label{eq:jacobian}
J\left[ S_{\omega}[\zeta,\zeta_*,\zeta',\zeta'_*] \right](v,v_*) = \frac{|v'-v'_*|}{|v-v_*|},
\end{equation}
where $(v',v'_*) = S_{\omega}[\zeta,\zeta_*,\zeta',\zeta'_*](v,v_*)$.
\end{lemma}

\begin{proof}
This proof is based on that one of Lemma 1 in \cite{desvillettes2005kinetic}. We decompose the transformation in a chain of elementary changes of variable. For the sake of clarity, we fix $\zeta,\zeta_*,\zeta',\zeta'_* \in \E$ and $\omega \in \mathbb{S}^2$, and denote the transformation $ S_{\omega}[\zeta,\zeta_*,\zeta',\zeta'_*]$ by simply writing
$$
(v,v_*) \mapsto (v',v'_*).
$$

\smallskip

\noindent First, we pass to the center-of-mass reference frame. The transformations
$A_1 : (v,v_*) \mapsto (g,G)$
where $g = v-v_*$, $G = (v + v_*)/2$, and
$B_1 : (g',G') \mapsto (v',v'_*)$
have both Jacobian equal to 1. Since $G = G'$, we are led to study the transformation
$C_1 : g \mapsto g'.$
Now, we pass to spherical coordinates for $g$ and $g'$. We perform the transformations
$A_2 : g  \mapsto (|g|, g/|g|)$
and
$B_2 : (|g'|,g'/|g'|) \mapsto g'$,
and finally we study the transformation
$C_2 : (|g|,g/|g|) \mapsto (|g'|,g'/|g'|)$
taking into account the determinant $|g'|^2 / |g|^2$ of the  Jacobian coming from $A_2$ and $B_2$.
Since $T_{\omega}$ is a symmetry, we have

\begin{align*}
&\hspace{13.5pt} J \left[  \left( |g|,\frac{g}{|g|} \right) \mapsto \left( 2 \sqrt{\frac14 |g|^2 + \frac{1}{m}( \varepsilon(\zeta) + \varepsilon(\zeta_*) - \varepsilon(\zeta') - \varepsilon(\zeta'_*) ) }, T_{\omega}\left[ \frac{g}{|g|} \right] \right) \right]\\
&= J \left[ |g| \mapsto  2 \sqrt{\frac14 |g|^2 + \frac{1}{m}(\varepsilon(\zeta) + \varepsilon(\zeta_*) - \varepsilon(\zeta') - \varepsilon(\zeta'_*))}  \right] \\ &= \frac{\frac12 |g|}{\sqrt{\frac14 |g|^2 + \frac{1}{m}(\varepsilon(\zeta) + \varepsilon(\zeta_*) - \varepsilon(\zeta') - \varepsilon(\zeta'_*))}} = \frac{|g|}{|g'|}.
\end{align*}
Multiplying by $|g'|^2/|g|^2$, we finally obtain

\begin{equation*}
J\left[ S_{\omega}[\zeta,\zeta_*,\zeta',\zeta'_*] \right](v,v_*) = \frac{|g'|}{|g|} = \frac{|v' - v'_*|}{|v-v_*|},
\end{equation*}
where $(v',v'_*) = S_{\omega}[\zeta,\zeta_*,\zeta',\zeta'_*](v,v_*)$.
\QED
\end{proof}

\section{Boltzmann model} \label{section:boltzmann}

We consider a distribution function $f$ that represents the density of molecules of the studied gas. This density depends on 4 parameters:

\begin{itemize}
    \item 2 macroscopic parameters: the time $t \in \R_+$ and the position in space $x$, typically $x \in \R^3$ or $\mathbb{T}^3$
    \item 2 microscopic parameters: the velocity $v \in \R^3$ and the internal state $\zeta \in \E$
\end{itemize}
As a distribution function, $f$ is non-negative, so that we study $f(t,x,v,\zeta) \in \R_+$. Moreover, we assume that for all $(t,x)$, $f(t,x,\cdot,\cdot) \in L^1(\R^3 \times \E, \, \text{d}v \, \text{d}\mu(\zeta))$. The evolution of this distribution function is governed by the Boltzmann equation

\begin{equation} \label{eq:boltzmann}
\partial_t f(t,x,v,\zeta) + v \cdot \nabla_x f(t,x,v,\zeta) = \B(f,f)(t,x,v,\zeta),
\end{equation}
where $\B$ is the Boltzmann collision operator we define in a following subsection.

\medskip

\noindent If $\phi$ defined on $\R^3 \times \E$ is a microscopic extensive quantity (test function, often called molecular property), then the associated macroscopic quantity is
$$
 \int_{\E} \int_{\R^3}\phi(v,\zeta) \, f(t,x,v,\zeta) \, \text{d}v \, \text{d}\mu(\zeta).
$$
Typical choice for the molecular properties are $\phi= m$, $\phi =  mv$ and $\displaystyle \phi = \frac{m}{2} |v|^2 + \varepsilon (\zeta)$, giving  the mass density $\rho$, velocity $u$ and total energy density $e$, respectively, namely
\begin{align*}
&\rho(t,x) :=  \int_{\E} \int_{\R^3}  m \, f(t,x,v,\zeta) \, \text{d}v \, \text{d}\mu(\zeta) \\
& u(t,x) := \frac{1}{\rho(t,x)} \,  \int_{\E} \int_{\R^3}  m \, v \, f(t,x,v,\zeta) \, \text{d}v \, \text{d}\mu(\zeta) \\
&e(t,x) :=   \frac{1}{\rho(t,x)} \,  \int_{\E}  \int_{\R^3}\left( \frac{m}{2} \,  |v|^2 + \varepsilon (\zeta) \right) \, f(t,x,v,\zeta) \, \text{d}v \, \text{d}\mu(\zeta).
\end{align*}
We can also define the temperature by
\begin{equation} \label{eqdef:temperature}
T(t,x) :=  \Theta^{-1} \left( \frac{1}{\rho(t,x)} \,  \int_{\E}  \int_{\R^3}\left( \frac{m}{2} \, |v- u|^2 + \Bar{\varepsilon}(\zeta) \right) \, f(t,x,v,\zeta) \, \text{d}v \, \text{d}\mu(\zeta) \right),
\end{equation}
where we recall that by definition (equation \eqref{eqdef:groundedenergyfunc}) $\Bar{\varepsilon} = \varepsilon - \varepsilon^0$, while $\Theta : \R_+ \rightarrow \R_+$ is a suitable function related to the integral of the specific heat at constant volume (the precise definition will be given later, in equation~\eqref{eqdef:theta} in Section \ref{section:Maxwellian}).

\subsection{Boltzmann collision operator} \label{subsection:boltzmann}

The Boltzmann operator, representing the effect of collisions on the change of velocities and internal states of each molecule,
is composed by a loss  and a gain term
$$
\B = \B_+ - \B_-.
$$
%
%
We focus on the microscopic aspect and put temporarily aside the variables $t$ and $x$. When two molecules collide, not all possible post-collision states are equiprobable. This idea is translated in the concept of the collision (or scattering) kernel $b$,  which encapsulates the information on the interaction potential (for example a hard sphere or Lennard-Jones one). As well known, the collision kernel is related to the differential cross section $\sigma$ through the formula $b=|v-v_*|\sigma$.

\noindent In this paper, for a purpose of generality, following the line of \cite{groppi2004two} we make as few assumptions as possible on the collision kernel $b$.

\begin{definition}
The collision kernel $b$ is a measurable function
\begin{equation} \label{eq:defcollisionkernel}
b : \begin{cases}
\R^3 \times \R^3 \times \E \times \E \times \E \times \E \times \mathbb{S}^2 \to \R_+ \\
(v,v_*,\zeta,\zeta_*,\zeta',\zeta'_*,\omega) \mapsto b(v,v_*,\zeta,\zeta_*,\zeta',\zeta'_*,\omega).
\end{cases}
\end{equation}
\end{definition}

\smallskip

\begin{assumption}
\emph{\textbf{Symmetry}. For a.e. $v,v_* \in \R^3$, for $\mu$-a.e. $\zeta,\zeta_*,\zeta',\zeta'_* \in \E$ and a.e. $\omega \in \mathbb{S}^2$,}
\begin{equation} \label{eq:symmetry}
b(v,v_*,\zeta,\zeta_*,\zeta',\zeta'_*,\omega) = b(v,v_*,\zeta,\zeta_*,\zeta'_*,\zeta',\omega) = b(v_*,v,\zeta_*,\zeta,\zeta'_*,\zeta',\omega).
\end{equation}
\end{assumption}

\begin{assumption} \label{assumption:microreversibility}
\emph{\textbf{Micro-reversibility}. For a.e. $v,v_* \in \R^3$, for $\mu$-a.e. $\zeta,\zeta_*,\zeta',\zeta'_* \in \E$ and a.e. $\omega \in \mathbb{S}^2$,}
\begin{equation} \label{eq:micro-reversibility}
 |v-v_*| \, b(v,v_*,\zeta,\zeta_*,\zeta',\zeta'_*,\omega) = |v'-v'_*| \, b(v',v'_*,\zeta',\zeta'_*,\zeta,\zeta_*,\omega).
\end{equation}
\emph{where} $(v',v'_*) = S_{\omega}[\zeta,\zeta_*,\zeta',\zeta'_*](v,v_*)$.
\end{assumption}

\smallskip

\begin{defassumption} \label{assumption:positivity}
\emph{We denote by $\mathbb{S}^2[v,v_*,\zeta,\zeta_*,\zeta',\zeta'_*]$ the subset of $\mathbb{S}^2$ such that for all $\omega \in \mathbb{S}^2$,
$$
\omega \in \mathbb{S}^2[v,v_*,\zeta,\zeta_*,\zeta',\zeta'_*] \iff b(v,v_*,\zeta,\zeta_*,\zeta',\zeta'_*,\omega) > 0.
$$
We assume that for a.e. $v,v_* \in \R^3$ and $\mu$-a.e. $\zeta,\zeta_*,\zeta',\zeta'_* \in \E$,}
\begin{align} \label{eq:positivity}
(v,v_*) \in \, &E[\zeta,\zeta_*,\zeta',\zeta'_*] \iff \quad \mathbb{S}^2[v,v_*,\zeta,\zeta_*,\zeta',\zeta'_*] \emph{ has positive measure}, \nonumber \\
&\mathbb{S}^2 \setminus \mathbb{S}^2[v,v_*,\zeta,\zeta,\zeta,\zeta] \emph{ has zero measure}.
\end{align}
\end{defassumption}
This assumption \ref{assumption:positivity} has the following meaning: if a collision is possible, then there exists a non-negligible set of angles $\omega$ for which the kernel is positive; if, on the other hand, a collision is not allowed, then the corresponding kernel is zero. Finally, if the collision is elastic, then the set of angles $\omega$ for which the kernel is positive is the whole sphere $\mathbb{S}^2$ (minus a negligible set). We thus allow, for inelastic collisions, the kernel to be zero for a non-negligible set of angles. This is an important feature, already present for example in \cite{groppi2004two} for the  case of reacting spheres.

\noindent Note that the symmetry and micro-reversibility assumptions are not in contradiction with the positivity assumption, due to the fact the the function $\mathbf{1}_{\Delta \geq 0}$ also has symmetry and reversibility properties. Moreover, setting $(v',v'_*) = S_{\omega}[\zeta,\zeta_*,\zeta',\zeta'_*](v,v_*)$, the symmetry and micro-reversibility assumptions imply
$$
\mathbb{S}^2[v_*,v,\zeta_*,\zeta,\zeta'_*,\zeta'] = \mathbb{S}^2[v,v_*,\zeta,\zeta_*,\zeta',\zeta'_*] = \mathbb{S}^2[v',v'_*,\zeta',\zeta'_*,\zeta,\zeta_*].
$$

\noindent We are now ready to define the Boltzmann collision operator. The loss term of the collision operator expresses the fact that a collision involving a molecule with state $(v,\zeta)$ causes a change of the state. Thus for all states $(v,\zeta) \in \R^3 \times \E$, and recalling that $b$ is null for not admissible collisions, the loss term reads as

\begin{equation} \label{eq:lossterm}
\B_-(f,f)(v,\zeta) := \iiint_{\E^3} \int_{\R^3}  \int_{\mathbb{S}^2}  f(v,\zeta)f(v_*,\zeta_*) \, b(v,v_*,\zeta,\zeta_*,\zeta',\zeta'_*,\omega) \, \text{d}\omega  \,  \text{d}v_*  \, \text{d}\mu^{\otimes 3}(\zeta_*,\zeta',\zeta'_*).
\end{equation}


\noindent The gain term of the collision operator expresses the production of the state $(v,\zeta)$ by collisions. Thus the gain term is obtained by considering \emph{all collisions} for which $(v,\zeta)$ is a post-collision state.
We have already proved in Lemma \ref{lemmabijection} that if $(v',v'_*) = S_{\omega}[\zeta,\zeta_*,\zeta',\zeta'_*](v,v_*)$ then $(v,v_*) = S_{\omega}[\zeta',\zeta'_*,\zeta,\zeta_*](v',v'_*)$, therefore the gain term will be symmetric to the loss one as expected. In more detail,
denoting by $\chi$ the Dirac mass on $\R^3$, the gain term is defined by
\begin{align*}
\B_+(f,f)(v,\zeta) := \iiint_{\E^3} \iint_{E[\zeta',\zeta'_*,\zeta,\zeta_*]}  \int_{\mathbb{S}^2}  \, \mathcal{\chi}(\Tilde{v} = v)  \, f(v',\zeta')f(v'_*,\zeta'_*) \, b(v',v'_*,\zeta',\zeta'_*,\zeta,\zeta_*,\omega) \\ \text{d}\omega  \, \text{d}v' \, \text{d}v'_*  \, \text{d}\mu^{\otimes 3}(\zeta',\zeta'_*,\zeta_*),
\end{align*}
where just for convenience we define $(\tilde{v},\Tilde{v}_*) = S_{\omega}[\zeta',\zeta'_*,\zeta,\zeta_*](v',v'_*)$, in order to be able to perform the change of variables $(\Tilde{v},\Tilde{v}_*) = S_{\omega}[\zeta',\zeta'_*,\zeta,\zeta_*](v',v'_*)$ (recalling that $S_{\omega}[\zeta',\zeta'_*,\zeta,\zeta_*]$ is a $\mathcal{C}^1$-diffeomorphism on $E[\zeta',\zeta'_*,\zeta,\zeta_*]$, except on a negligible set) leading to
\begin{align*}
\B_+(f,f)(v,\zeta) = \iiint_{\E^3} \iint_{E[\zeta,\zeta_*,\zeta',\zeta'_*]}  \int_{\mathbb{S}^2}  \, \mathcal{\chi}(\Tilde{v} = v)  \, f(v',\zeta')f(v'_*,\zeta'_*) \, J \left[(S_{\omega}[\zeta',\zeta'_*,\zeta,\zeta_*])^{-1} \right] \, b(v',v'_*,\zeta',\zeta'_*,\zeta,\zeta_*,\omega) \\ \text{d}\omega  \, \text{d}\Tilde{v} \, \text{d}\Tilde{v}_*  \, \text{d}\mu^{\otimes 3}(\zeta',\zeta'_*,\zeta_*)\\
 = \iiint_{\E^3} \int_{\R^3}  \int_{\mathbb{S}^2} f(v',\zeta')f(v'_*,\zeta'_*) \, J \left[(S_{\omega}[\zeta',\zeta'_*,\zeta,\zeta_*])^{-1} \right] \, b(v',v'_*,\zeta',\zeta'_*,\zeta,\zeta_*,\omega)\, \text{d}\omega  \, \text{d}v_*  \, \text{d}\mu^{\otimes 3}(\zeta',\zeta'_*,\zeta_*),
\end{align*}
where in last line we went back to the original notation, changing $\Tilde{v}_*$ to $v_*$.
The Jacobian of $(S_{\omega}[\zeta',\zeta'_*,\zeta,\zeta_*])^{-1} = S_{\omega}[\zeta,\zeta_*,\zeta',\zeta'_*]$ is given by equation \eqref{eq:jacobian}, and it provides
\begin{align*}
\B_+(f,f)(v,\zeta) = \iiint_{\E^3} \int_{\R^3}  \int_{\mathbb{S}^2} f(v',\zeta')f(v'_*,\zeta'_*) \, \frac{|v'-v'_*|}{|v - v_*|} \, b(v',v'_*,\zeta',\zeta'_*,\zeta,\zeta_*,\omega) \text{d}\omega  \, \text{d}v_* \, \text{d}\mu^{\otimes 3}(\zeta',\zeta'_*,\zeta_*).
\end{align*}
Using now the micro-reversibility condition on $b$, Assumption \ref{assumption:microreversibility}, we finally get
\begin{equation} \label{eq:gainterm}
\B_+(f,f)(v,\zeta) = \iiint_{\E^3}  \int_{\R^3}  \int_{\mathbb{S}^2} \, f(v',\zeta')f(v'_*,\zeta'_*) \, b(v,v_*,\zeta,\zeta_*,\zeta',\zeta'_*,\omega) \, \text{d}\omega \, \text{d}v_*   \, \text{d}\mu^{\otimes 3}(\zeta_*,\zeta',\zeta'_*),
\end{equation}
where we recall that $(v',v'_*) = S_{\omega}[\zeta,\zeta_*,\zeta',\zeta'_*](v,v_*)$. Since $f$ and $b$ are non-negative measurable functions, the gain and loss terms are always defined. In order to define the full collision term, we impose the following assumption on $f$.

\begin{assumption}
\emph{For a.e. $v \in \R^3$ and for $\mu$-a.e. $\zeta \in \E$,}
$$
\B_-(f,f)(v,\zeta)  + \B_+(f,f)(v,\zeta) < \infty.
$$
\end{assumption}

\smallskip

\noindent Writing $b(\cdot) \equiv b(v,v_*,\zeta,\zeta_*,\zeta',\zeta'_*,\omega)$ to lighten the notations, the full Boltzmann collision operator writes, for all $v \in \R^3$ and $\zeta \in \E$,

\begin{equation}\label{eq:fullterm}
\B(f,f)(v,\zeta) = \iiint_{\E^3}  \int_{\R^3}  \int_{\mathbb{S}^2} \, \Big(f(v',\zeta')f(v'_*,\zeta'_*) - f(v,\zeta)f(v_*,\zeta_*) \Big) \; b(\cdot)  \;  \text{d}\omega \, \text{d}v_*  \, \text{d}\mu^{\otimes 3}(\zeta_*,\zeta',\zeta'_*),
\end{equation}
where $(v',v'_*) = S_{\omega}[\zeta,\zeta_*,\zeta',\zeta'_*](v,v_*)$.

\begin{remark}
Alternatively, one can define the Boltzmann collision operator using transition probabilities instead of $b$, see \cite{Giovangigli}.
\end{remark}

\subsection{Note about degeneracy} \label{subsection:degeneracy}
The concept of degeneracy appears when the internal description is done through the direct consideration of energy levels. For instance, in a quantum description, different proper states can sometimes have the same energy. The number of proper states sharing the same energy level is called the \emph{degeneracy} of this energy level. When degeneracy is considered in kinetic theory, it usually appears in the micro-reversibility condition of the collision kernel (see for example Section 4.2.4 in Giovangigli \cite{giovangigli2012multicomponent}, equations (4.2.7) and (4.2.8)), and ultimately in the formulation of the Boltzmann integral. We show in this subsection that, without loss of generality, it is possible to assume no degeneracy in the micro-reversibility condition and take this into account through the measure $\mu$. Note also that if the internal description is done through the proper states instead of directly through the energy levels, there is no degeneracy to consider.

\smallskip

\noindent The degeneracy is a measurable positive function on the space of internal states
$$
a : \E \to \R_+^*.
$$
The micro-reversibility condition taking degeneracy into account writes (see Section 4.2.4 in Giovangigli \cite{giovangigli2012multicomponent}, equation (4.2.8))
\begin{equation} \label{deg-symmetry}
|v-v_*| \, b(v,v_*,\zeta,\zeta_*,\zeta',\zeta'_*,\omega) \, a(\zeta) a(\zeta_*)= |v'-v'_*| \,  b(v',v'_*,\zeta',\zeta'_*,\zeta,\zeta_*,\omega) \, a(\zeta') a(\zeta'_*).
\end{equation}
Classical symmetry condition for cross-sections should be valid only if the molecular states are non-degenerate \cite{FK72}. However, Waldmann \cite{Wa58} has shown that semi-classical Boltzmann modelling is indeed admissible for all molecules, even with degeneracy properties, provided the quantum mechanical cross-sections are replaced by suitably degeneracy averaged cross-sections \cite{Wa58,EG94}. Symmetry property (\ref{deg-symmetry}) may be obtained owing to the invariance of the Hamiltonian under the combined operation of space inversion and time reversal \cite{Wa58,MS64}. We remark that in the classical model with a continuous internal energy variable $I$ \cite{borgnakke1975statistical,desvillettes1997modele,desvillettes2005kinetic}, the weight function $\varphi(I)$ in the Boltzmann kernel plays the role of the degeneracy function $a(\cdot)$.

\noindent With the same reasoning as in the previous subsection \ref{subsection:boltzmann}, the collision term writes in this case
$$
\B(f,f)(v,\zeta) = \iiint_{\E^3} \int_{\R^3}  \int_{\mathbb{S}^2}  \, \left(\frac{a(\zeta) a(\zeta_*)}{a(\zeta') a(\zeta'_*)}  f(v',\zeta')f(v'_*,\zeta'_*) - f(v,\zeta)f(v_*,\zeta_*) \right) \; b(\cdot)  \;  \text{d}\omega \, \text{d}v_*  \, \text{d}\mu^{\otimes 3}(\zeta_*,\zeta',\zeta'_*).
$$
Let us now define
$$
\Tilde{b}(v,v_*,\zeta,\zeta_*,\zeta',\zeta'_*,\omega) = \frac{b(v,v_*,\zeta,\zeta_*,\zeta',\zeta'_*,\omega)}{a(\zeta') a(\zeta'_*)}, \quad \quad  \text{d}\Tilde{\mu}(\zeta) = a(\zeta) \, \text{d}\mu(\zeta), \quad \quad \Tilde{f}(v,\zeta) =  \frac{f(v,\zeta)}{a(\zeta)}.
$$
Then $\Tilde{b}$ verifies the assumptions of our framework. Now let us set
$$
\Tilde{\B}(\Tilde{f},\Tilde{f})(v,\zeta) = \iiint_{\E^3}  \int_{\R^3}  \int_{\mathbb{S}^2} \, \Big( \Tilde{f}(v',\zeta')\Tilde{f}(v'_*,\zeta'_*) - \Tilde{f}(v,\zeta)\Tilde{f}(v_*,\zeta_*) \Big) \; \Tilde{b}(\cdot)  \; \text{d}\omega \, \text{d}v_*   \, \text{d}\Tilde{\mu}^{\otimes 3}(\zeta_*,\zeta',\zeta'_*).
$$
Then
\begin{align*}
\Tilde{\B}(\Tilde{f},\Tilde{f})(v,\zeta) &= \frac{1}{a(\zeta)} \iiint_{\E^3}  \int_{\R^3}  \int_{\mathbb{S}^2} \, \left( \frac{f(v',\zeta')f(v'_*,\zeta'_*)}{a(\zeta') a(\zeta'_*)}  - \frac{f(v,\zeta)f(v_*,\zeta_*)}{a(\zeta) a(\zeta_*)} \right) \; b(\cdot) \, a(\zeta) a(\zeta_*) \; \text{d}\omega  \, \text{d} v_* \,  \text{d}\mu^{\otimes 3}(\zeta_*,\zeta',\zeta'_*) \\
&= \frac{1}{a(\zeta)} \iiint_{\E^3} \int_{\R^3}  \int_{\mathbb{S}^2}  \, \left(\frac{a(\zeta) a(\zeta_*)}{a(\zeta') a(\zeta'_*)}  f(v',\zeta')f(v'_*,\zeta'_*) - f(v,\zeta)f(v_*,\zeta_*) \right) \; b(\cdot)  \; \text{d}\omega  \, \text{d}v_* \, \text{d}\mu^{\otimes 3}(\zeta_*,\zeta',\zeta'_*)\\
&= \frac{1}{a(\zeta)} \B(f,f)(v,\zeta).
\end{align*}
It follows that
$$
\Tilde{\B}(\Tilde{f},\Tilde{f})(v,\zeta) \text{d}\Tilde{\mu}(\zeta) =  \frac{1}{a(\zeta)} \B(f,f)(v,\zeta) a(\zeta) \text{d}\mu(\zeta) = \B(f,f)(v,\zeta)  \text{d}\mu(\zeta)
$$
and
$$
\Tilde{f}(v,\zeta)  \text{d}\Tilde{\mu}(\zeta) =  \frac{1}{a(\zeta)} f(v,\zeta) a(\zeta) \text{d}\mu(\zeta) = f(v,\zeta)  \text{d}\mu(\zeta).
$$
Thus it is equivalent to study the model $(\E,\A,\mu)$ with the kernel $b$ and degeneracy $a$, and the model $(\E,\A,\Tilde{\mu})$ with the kernel $\Tilde{b}$ and no degeneracy. In our framework, degeneracy can be simply included in the measure on the space of internal states, and does not appear in the micro-reversibility condition.

\subsection{Weak formulation}
In this section, we study the weak formulation of the collision integral that will later enable us to deduce the conservation laws and collision invariants.

\smallskip

\begin{proposition} \label{prop:weak}
Let $\psi : \R^3 \times \E \to \R$ be a measurable function such that
$$
\int_{\E} \int_{\R^3} \B_-(f,f)(v,\zeta) \left|\psi(v,\zeta) \right| \, \emph{d}v \, \emph{d}\mu(\zeta) + \int_{\E} \int_{\R^3} \B_+(f,f)(v,\zeta) \left|\psi(v,\zeta) \right| \, \emph{d}v \, \emph{d}\mu(\zeta) < \infty.
$$
Then

\begin{align*}
\int_{\E} \int_{\R^3} \B(f,f)(v,\zeta) \psi(v,\zeta) \, \emph{d}v \, \emph{d}\mu(\zeta) = - \frac14 \iiiint_{\E^4} \iint_{(\R^3)^2}  \int_{\mathbb{S}^2} \Big( f(v',\zeta')f(v'_*,\zeta'_*) - f(v,\zeta)f(v_*,\zeta_*) \Big)  \\
\times \Big[\psi(v',\zeta') + \psi(v'_*,\zeta'_*) - \psi(v,\zeta) - \psi(v_*,\zeta_*)  \Big]  b(\cdot)  \, \emph{d}\omega \, \emph{d}v \,  \emph{d}v_* \, \emph{d}\mu^{\otimes 4}(\zeta,\zeta_*,\zeta',\zeta'_*)
\end{align*}
where $(v',v'_*) = S_{\omega}[\zeta,\zeta_*,\zeta',\zeta'_*](v,v_*)$.
\end{proposition}

\begin{proof}

\noindent First we have, thanks to the symmetry conditions on  $b$ and Fubini's Theorem,
\begin{align*}
&\int_{\E} \int_{\R^3} \B_-(f,f)(v,\zeta) \psi(v,\zeta) \, \text{d}v \, \text{d}\mu(\zeta) \\
&= \iiiint_{\E^4} \iint_{(\R^3)^2}  \int_{\mathbb{S}^2}  \, f(v,\zeta)f(v_*,\zeta_*) \, \psi(v,\zeta) b(v,v_*,\zeta,\zeta_*,\zeta',\zeta'_*,\omega) \, \text{d}\omega \,  \text{d}v  \,  \text{d}v_*  \, \text{d}\mu^{\otimes 4}(\zeta,\zeta_*,\zeta',\zeta'_*)\\
&=  \iiiint_{\E^4}  \iint_{(\R^3)^2}  \int_{\mathbb{S}^2}  \, f(v_*,\zeta_*)f(v,\zeta) \, \psi(v_*,\zeta_*) b(v_*,v,\zeta_*,\zeta,\zeta'_*,\zeta',\omega) \, \text{d}\omega \,  \text{d}v_* \, \text{d}v \,  \text{d}\mu^{\otimes 4}(\zeta_*,\zeta,\zeta'_*,\zeta')  \\
&= \iiiint_{\E^4} \iint_{(\R^3)^2}  \int_{\mathbb{S}^2}  \, f(v,\zeta)f(v_*,\zeta_*) \, \psi(v_*,\zeta_*) b(v,v_*,\zeta,\zeta_*,\zeta',\zeta'_*,\omega) \, \text{d}\omega \,  \text{d}v  \,  \text{d}v_*  \, \text{d}\mu^{\otimes 4}(\zeta,\zeta_*,\zeta',\zeta'_*),
\end{align*}
thus
\begin{align*}
&\int_{\E} \int_{\R^3} \B_-(f,f)(v,\zeta) \psi(v,\zeta) \, \text{d}v \, \text{d}\mu(\zeta)\\ & = \frac12 \iiiint_{\E^4}  \iint_{(\R^3)^2} \int_{\mathbb{S}^2}  \, f(v,\zeta)f(v_*,\zeta_*) \, [\psi(v,\zeta) + \psi(v_*,\zeta_*)] \, b(\cdot) \, \text{d}\omega \,  \text{d}v  \,  \text{d}v_*  \, \text{d}\mu^{\otimes 4}(\zeta,\zeta_*,\zeta',\zeta'_*).
\end{align*}
Now by renaming all the variables, we get
\begin{align*}
&\int_{\E} \int_{\R^3} \B_-(f,f)(v,\zeta) \psi(v,\zeta) \, \text{d}v \, \text{d}\mu(\zeta)\\ & = \frac12 \iiiint_{\E^4}  \iint_{(\R^3)^2}  \int_{\mathbb{S}^2} f(v',\zeta')f(v'_*,\zeta'_*) \, [\psi(v',\zeta') + \psi(v'_*,\zeta'_*)]  b(v',v'_*,\zeta',\zeta'_*,\zeta,\zeta_*,\omega) \, \text{d}\omega \, \text{d}v'   \text{d}v'_* \,  \text{d}\mu^{\otimes 4}(\zeta',\zeta'_*,\zeta,\zeta_*).
\end{align*}
We now perform the change of variables $(v,v_*) = S_{\omega}[\zeta',\zeta'_*,\zeta,\zeta_*](v',v'_*)$, use the formula for the Jacobian \eqref{eq:jacobian} and the micro-reversibility condition and obtain
\begin{align*}
&\int_{\E} \int_{\R^3} \B_-(f,f)(v,\zeta) \psi(v,\zeta) \, \text{d}v \, \text{d}\mu(\zeta)\\ & = \frac12 \iiiint_{\E^4}  \iint_{(\R^3)^2}  \int_{\mathbb{S}^2} \, f(v',\zeta')f(v'_*,\zeta'_*) \, [\psi(v',\zeta') + \psi(v'_*,\zeta'_*)] \, b(v,v_*,\zeta,\zeta_*,\zeta',\zeta'_*,\omega) \, \text{d}\omega \,  \text{d}v  \,  \text{d}v_*  \, \text{d}\mu^{\otimes 4}(\zeta,\zeta_*,\zeta',\zeta'_*).
\end{align*}
where $(v',v'_*) = S_{\omega}[\zeta,\zeta_*,\zeta',\zeta'_*](v,v_*)$. We deduce that
\begin{align*}
\int_{\E} \int_{\R^3} \B_-(f,f)(v,\zeta) \psi(v,\zeta) \, \text{d}v \, \text{d}\mu(\zeta)  \hspace{260pt} \\  = \frac14 \iiiint_{\E^4}  \iint_{(\R^3)^2}  \int_{\mathbb{S}^2} \, \Big(f(v',\zeta')f(v'_*,\zeta'_*) \, [\psi(v',\zeta') + \psi(v'_*,\zeta'_*)] + f(v,\zeta)f(v_*,\zeta_*) \, [\psi(v,\zeta) + \psi(v_*,\zeta_*)] \Big) \\ b(\cdot) \, \text{d}\omega \,  \text{d}v  \,  \text{d}v_*  \, \text{d}\mu^{\otimes 4}(\zeta,\zeta_*,\zeta',\zeta'_*).
\end{align*}
where $(v',v'_*) = S_{\omega}[\zeta,\zeta_*,\zeta',\zeta'_*](v,v_*)$.

\bigskip
\noindent On the other hand, since
$$
\B_+(f,f)(v,\zeta) = \iiint_{\E^3}  \int_{\R^3}  \int_{\mathbb{S}^2} \, f(v',\zeta')f(v'_*,\zeta'_*) \, b(v,v_*,\zeta,\zeta_*,\zeta',\zeta'_*,\omega) \, \text{d}\omega \, \text{d}v_*   \, \text{d}\mu^{\otimes 3}(\zeta_*,\zeta',\zeta'_*),
$$
where $(v',v'_*) = S_{\omega}[\zeta,\zeta_*,\zeta',\zeta'_*](v,v_*)$, we have
\begin{align*}
&\int_{\E} \int_{\R^3} \B_+(f,f)(v,\zeta) \psi(v,\zeta) \, \text{d}v \, \text{d}\mu(\zeta)\\ & = \iiiint_{\E^4}   \iint_{(\R^3)^2} \int_{\mathbb{S}^2} \,  f(v',\zeta')f(v'_*,\zeta'_*) \, \psi(v,\zeta) \, b(v,v_*,\zeta,\zeta_*,\zeta',\zeta'_*,\omega) \, \text{d}\omega \, \text{d}v \, \text{d}v_* \, \text{d}\mu^{\otimes 4}(\zeta,\zeta_*,\zeta',\zeta'_*).
\end{align*}
Applying the change of variables $(v',v'_*) = S_{\omega}[\zeta,\zeta_*,\zeta',\zeta'_*](v,v_*)$ using the formula for the Jacobian, equation \eqref{eq:jacobian}, together with the micro-reversibility condition on $b$, Assumption \ref{assumption:microreversibility}, we obtain
\begin{align*}
&\int_{\E} \int_{\R^3} \B_+(f,f)(v,\zeta) \psi(v,\zeta) \, \text{d}v \, \text{d}\mu(\zeta)\\ & = \iiiint_{\E^4}   \iint_{(\R^3)^2} \int_{\mathbb{S}^2} \,  f(v',\zeta')f(v'_*,\zeta'_*) \, \psi(v,\zeta) \, b(v',v'_*,\zeta',\zeta'_*,\zeta,\zeta_*,\omega) \, \text{d}\omega \, \text{d}v' \, \text{d}v'_* \, \text{d}\mu^{\otimes 4}(\zeta',\zeta'_*,\zeta,\zeta_*)
\end{align*}
where $v$ in the last integral is defined by $(v,v_*) = S_{\omega}[\zeta',\zeta'_*,\zeta,\zeta_*] (v',v'_*)$. By simply renaming the variables (exchange of prime and non-prime), we get
\begin{align*}
&\int_{\E} \int_{\R^3} \B_+(f,f)(v,\zeta) \psi(v,\zeta) \, \text{d}v \, \text{d}\mu(\zeta)\\ & = \iiiint_{\E^4}  \iint_{(\R^3)^2} \int_{\mathbb{S}^2}  \,  f(v,\zeta)f(v_*,\zeta_*) \, \psi(v',\zeta') \, b(v,v_*,\zeta,\zeta_*,\zeta',\zeta'_*,\omega) \, \text{d}\omega \, \text{d}v \, \text{d}v_* \, \text{d}\mu^{\otimes 4}(\zeta,\zeta_*,\zeta',\zeta'_*)
\end{align*}
where $v'$ is defined by $(v',v'_*) = S_{\omega}[\zeta,\zeta_*,\zeta',\zeta'_*] (v,v_*)$. By using the same arguments as previously, we get
\begin{align*}
\int_{\E} \int_{\R^3} \B_+(f,f)(v,\zeta) \psi(v,\zeta) \, \text{d}v \, \text{d}\mu(\zeta) \hspace{260pt} \\  = \frac14 \iiiint_{\E^4}  \iint_{(\R^3)^2}  \int_{\mathbb{S}^2} \, \Big(f(v',\zeta')f(v'_*,\zeta'_*) \, [\psi(v,\zeta) + \psi(v_*,\zeta_*)] + f(v,\zeta)f(v_*,\zeta_*) \, [\psi(v',\zeta') + \psi(v'_*,\zeta'_*)] \Big) \\ b(\cdot) \, \text{d}\omega \,  \text{d}v  \,  \text{d}v_*  \, \text{d}\mu^{\otimes 4}(\zeta,\zeta_*,\zeta',\zeta'_*).
\end{align*}
where $(v',v'_*) = S_{\omega}[\zeta,\zeta_*,\zeta',\zeta'_*](v,v_*)$. Putting the previous results together yields
\begin{align*}
\int_{\E} \int_{\R^3} \B(f,f)(v,\zeta) \psi(v,\zeta) \, \text{d}v \, \text{d}\mu(\zeta) = - \frac14 \iiiint_{\E^4} \iint_{(\R^3)^2}  \int_{\mathbb{S}^2}  \Big( f(v',\zeta')f(v'_*,\zeta'_*) - f(v,\zeta)f(v_*,\zeta_*) \Big)  \\
\times \Big[\psi(v',\zeta') + \psi(v'_*,\zeta'_*) - \psi(v,\zeta) - \psi(v_*,\zeta_*)  \Big]  b(\cdot)  \, \text{d}\omega \, \text{d}v \,  \text{d}v_* \, \text{d}\mu^{\otimes 4}(\zeta,\zeta_*,\zeta',\zeta'_*).
\end{align*}
\QED
\end{proof}

\subsection{Collision invariants}

\begin{definition}
$\psi$ is a \emph{collision invariant} if and only if $\text{for $\mu$-a.e. } \zeta,\zeta_*,\zeta',\zeta'_* \in \E, \text{ a.e. } (v,v_*) \in E[\zeta,\zeta_*,\zeta',\zeta'_*]$ and a.e. $\omega \in \mathbb{S}^2[v,v_*,\zeta,\zeta_*,\zeta',\zeta'_*]$,
\begin{equation} \label{eq:collisioninvariant}
\psi(v',\zeta') + \psi(v'_*,\zeta'_*) = \psi(v,\zeta) + \psi(v_*,\zeta_*),
\end{equation}
where $(v',v'_*) = S_{\omega}[\zeta,\zeta_*,\zeta',\zeta'_*] (v,v_*)$.
\end{definition}

\noindent First, we immediately deduce from the conservation laws \eqref{eq:conservation} that if there exists $\alpha \in \R, \; \beta \in \R^3$ and $\gamma \in \R$ such that for a.e. $v \in \R^3$ and $\mu$-a.e. $\zeta \in \E$,
$$
\psi(v,\zeta) = \alpha +  \beta \cdot v + \gamma \left( \frac{m}{2} \,|v|^2 + \varepsilon(\zeta) \right),
$$
then $\psi$ is a collision invariant. In this subsection, we prove that all collision invariants verifying a certain integrability assumption have this form.

\bigskip

\begin{lemma} \label{lemma:meanS}
For all $\zeta \in \E$ and $\omega \in \mathbb{S}^2$,
$$
S_{\omega}[\zeta,\zeta,\zeta,\zeta] \equiv \Bar{S}_{\omega} :
\begin{cases}
\R^3 \times \R^3 \hspace{20pt} \longrightarrow \hspace{50pt} \R^3 \times \R^3 \\
\displaystyle (v,v_*) \mapsto \left( \frac{v + v_*}{2} + \frac{|v-v_*|}{2} T_{\omega} \left[ \frac{v-v_*}{|v-v_*|} \right], \frac{v + v_*}{2} - \frac{|v-v_*|}{2} T_{\omega} \left[ \frac{v-v_*}{|v-v_*|} \right] \right)
\end{cases}
$$
\end{lemma}

\begin{proof}
For any $v,v_* \in \R^3$,
$$
\Delta(v,v_*,\zeta,\zeta,\zeta,\zeta) = \frac14 |v-v_*|^2 + \frac{1}{m} (\varepsilon(\zeta) + \varepsilon(\zeta) - \varepsilon(\zeta) - \varepsilon(\zeta)) =  \frac14 |v-v_*|^2 \geq 0
$$
so that
$$
E[\zeta,\zeta,\zeta,\zeta] = \Big\{ (v,v_*) \in \R^3 \times \R^3 \text{ s.t. } \Delta(v,v_*,\zeta,\zeta,\zeta,\zeta) \geq 0 \Big\} = \R^3 \times \R^3.
$$
Now let $(v,v_*) \in \R^3 \times \R^3$. We set $(v',v'_*) = S_{\omega}[\zeta,\zeta,\zeta,\zeta](v,v_*)$. Then
\begin{align*}
v' &= \frac{v + v_*}{2} + \sqrt{ \Delta(\cdot) } \, T_{\omega} \left[ \frac{v-v_*}{|v-v_*|} \right]\\
&= \frac{v + v_*}{2} + \frac{|v-v_*|}{2} \, T_{\omega} \left[ \frac{v-v_*}{|v-v_*|} \right],
\end{align*}
and analogously
$$
v'_* = \frac{v + v_*}{2} - \frac{|v-v_*|}{2} \, T_{\omega} \left[ \frac{v-v_*}{|v-v_*|} \right].
$$
Thus we indeed have
$$
S_{\omega}[\zeta,\zeta,\zeta,\zeta] :
\begin{cases}
\displaystyle \R^3 \times \R^3 \hspace{20pt} \longrightarrow \hspace{50pt} \R^3 \times \R^3 \\
\displaystyle (v,v_*) \mapsto \left( \frac{v + v_*}{2} + \frac{|v-v_*|}{2} \, T_{\omega} \left[ \frac{v-v_*}{|v-v_*|} \right], \frac{v + v_*}{2} - \frac{|v-v_*|}{2} \, T_{\omega} \left[ \frac{v-v_*}{|v-v_*|} \right] \right).
\end{cases}
$$
\QED

\end{proof}

\smallskip
\noindent The following Theorem \ref{theoremcollisioninvariants} is the well-known characterization of collision invariants in the elastic case.

\begin{theorem} \label{theoremcollisioninvariants}
Let $r_0 > 0$ and $\phi \in L^1 \left(\R^3, \, e^{-|v|^2/r_0} \, \emph{d}v \right)$ such that for a.e. $(v,v_*) \in \R^3 \times \R^3 \text{ and } \omega \in \mathbb{S}^2$,
$$
\phi(v'(v,v_*,\omega)) + \phi(v'_*(v,v_*,\omega)) = \phi(v) + \phi(v_*).
$$
with $(v'(v,v_*,\omega),v'_*(v,v_*,\omega)) = \Bar{S}_{\omega}(v,v_*)$. Then there exists $(\alpha, \beta, \gamma) \in \R \times \R^3 \times \R$ such that for a.e. $v \in \R^3$,
$$
\phi(v) = \alpha + \beta \cdot v + \frac{\gamma}{2} |v|^2.
$$
\end{theorem}

\begin{proof}
For the proof, we refer the reader to Bouchut and Golse \cite{bouchut2000kinetic}, Chapter 2.
\QED
\end{proof}
\smallskip

\noindent The following Theorem \ref{theoreminvariantgeneral} characterises collision invariants in the general setting.
\begin{theorem} \label{theoreminvariantgeneral}
Let $\psi$ be a collision invariant, such that for $\mu$-a.e. $\zeta \in \E$ there exists $r_{\zeta}>0$ such that $\psi(\cdot,\zeta) \in L^1 \left(\R^3, \, e^{-|v|^2/r_{\zeta}} \, \emph{d}v \right)$. Then there exists $(\alpha, \beta, \gamma) \in \R \times \R^3 \times \R$ such that for $\mu$-a.e. $\zeta \in \E$ and a.e. $v \in \R^3$,

$$
\psi(v,\zeta) = \alpha  + \beta \cdot  v + \gamma \left( \frac{m}{2} \, |v|^2 + \varepsilon(\zeta) \right).
$$
\end{theorem}

\begin{proof}
\noindent Let us denote by $\Tilde{\E}$ the set such that $\mu(\E \setminus \Tilde{\E}) = 0$, and for all $\zeta \in \tilde{\E}$, there exists $r_{\zeta}>0$ such that $\psi(\cdot,\zeta) \in L^1 \left(\R^3, \, e^{-|v|^2/r_{\zeta}} \text{d}v \right)$ and for all $\zeta_*,\zeta',\zeta'_* \in \tilde{\E}$, a.e. $(v,v_*) \in E[\zeta,\zeta_*,\zeta',\zeta'_*]$ and a.e. $\omega \in \mathbb{S}^2[v,v_*,\zeta,\zeta_*,\zeta',\zeta'_*]$,
$$
\psi(v',\zeta') + \psi(v'_*,\zeta'_*) = \psi(v,\zeta) + \psi(v_*,\zeta_*),
$$
where $(v',v'_*)$ are defined by $(v',v'_*) = S_{\omega}[\zeta,\zeta_*,\zeta',\zeta'_*] (v,v_*)$.

\smallskip

\noindent First, we consider the equality for $\zeta=\zeta_*=\zeta'=\zeta'_* \in \tilde{\E}$. From Lemma \ref{lemma:meanS}, Assumption \ref{assumption:positivity} (in our framework, all angles are allowed in elastic collisions) and Theorem \ref{theoremcollisioninvariants}, there exists $(\alpha_{\zeta}, \beta_{\zeta}, \gamma_{\zeta}) \in \R \times \R^3 \times \R$ such that for a.e. $v \in \R^3$,

\begin{equation} \label{eq:psizeta}
\psi(v,\zeta) = \alpha_{\zeta}  + \beta_{\zeta} \cdot v + \gamma_{\zeta} \left( \frac12  |v|^2 \right).
\end{equation}

\noindent Let us now consider again the equation \eqref{eq:collisioninvariant}, using formula \eqref{eq:psizeta}, $\zeta = \zeta_* \in \tilde{\E}$ and $\zeta' = \zeta'_* \in \Tilde{\E}$. It writes, for a.e. $(v,v_*) \in E[\zeta,\zeta,\zeta',\zeta']$ and a.e. $\omega \in \mathbb{S}^2[v,v_*,\zeta,\zeta,\zeta',\zeta']$,
$$
2 \alpha_{\zeta} + \beta_{\zeta} \cdot (v+v_*) + \gamma_{\zeta} \left( \frac12 |v|^2 + \frac12 |v_*|^2 \right) = 2 \alpha_{\zeta'} + \beta_{\zeta'} \cdot (v'+v'_*) + \gamma_{\zeta'} \left( \frac12 |v'|^2 + \frac12 |v'_*|^2 \right)
$$
where $(v',v'_*) = S_{\omega}[\zeta,\zeta,\zeta',\zeta'](v,v_*)$. By assumption \ref{assumption:positivity}, at fixed $v,v_*,\zeta,\zeta'$, the set of angles $\omega$ for which the above equation holds, $\mathbb{S}^2[v,v_*,\zeta,\zeta,\zeta',\zeta']$, is not negligible. Using the conservation equations \eqref{eq:conservation}, we thus get for a.e. $v,v_*$ such that $\displaystyle \frac14|v-v_*|^2 \geq \frac{2}{m} \,(\varepsilon(\zeta') - \varepsilon(\zeta))$,
$$
2 \left(\alpha_{\zeta} - \alpha_{\zeta'} - \frac{\gamma_{\zeta'}}{m}  \left(\varepsilon(\zeta') - \varepsilon(\zeta) \right) \right) + (\beta_{\zeta} - \beta_{\zeta'}) \cdot (v+v_*) + (\gamma_{\zeta} - \gamma_{\zeta'}) \left( \frac12 |v|^2 + \frac12 |v_*|^2 \right) = 0.
$$
A polynomial that vanishes on a set of infinite cardinal must have all its coefficients equal to zero, hence
\begin{align*}
    \gamma_{\zeta} &= \gamma_{\zeta'}\\
    \beta_{\zeta} &= \beta_{\zeta'}\\
    \alpha_{\zeta} -  \gamma_{\zeta'} \frac{\varepsilon(\zeta)}{m}  &= \alpha_{\zeta'} - \gamma_{\zeta'} \frac{\varepsilon(\zeta')}{m}  .
\end{align*}
We conclude that there exists $(\alpha, \beta, \gamma) \in \R \times \R^3 \times \R$ such that for a.e. $v \in \R^3$ and for all $\zeta \in \Tilde{\E}$ (thus $\mu$-a.e. $\zeta \in \E$),
$$
\psi(v,\zeta) = \alpha + \beta \cdot v + \gamma \left( \frac{m}{2} |v|^2 + \varepsilon(\zeta)\right).
$$
\QED
\end{proof}

\section{H Theorem and collision equilibria} \label{section:Maxwellian}

\begin{theorem} \textbf{H Theorem.} \label{htheorem} Let $f : \R^3 \times \E \to \R_+$ be a measurable function such that for a.e. $v \in \R^3$ and $\mu$-a.e. $\zeta \in \E$, $f(v,\zeta) > 0$. We assume that
$$
\int_{\E} \int_{\R^3} \B_-(f,f)(v,\zeta) \left|\emph{ln}(f)(v,\zeta) \right| \, \emph{d}v \, \emph{d}\mu(\zeta) + \int_{\E} \int_{\R^3} \B_+(f,f)(v,\zeta) \left|\emph{ln}(f)(v,\zeta) \right| \, \emph{d}v \, \emph{d}\mu(\zeta) < \infty.
$$
Then
\begin{equation}\label{eq:htheorem}
\int_{\E} \int_{\R^3}  \B(f,f)(v,\zeta) \, \emph{ln}(f)(v,\zeta) \, \emph{d}v \, \emph{d}\mu(\zeta) \leq 0,
\end{equation}
and

\begin{align*}
\int_{\E} \int_{\R^3}  \B(f,f)(v,\zeta) \, \emph{ln}&(f)(v,\zeta) \, \emph{d}v \, \emph{d}\mu(\zeta) = 0 \\
&\Updownarrow \\ \B(f,f)(v,\zeta) = 0 \text{ for a.e. } &v \in \R^3 \text{ and $\mu$-a.e. } \zeta \in \E \\
&\Updownarrow \\ \emph{ln}(f) \text{ is a } & \text{collision invariant.}
\end{align*}
\end{theorem}

\begin{proof}
This proof is highly inspired by a proof which can be found in \cite{bouchut2000kinetic}, Chapter 2. First we define, for $\mu$-a.e. $\zeta,\zeta_*,\zeta',\zeta'_* \in \E$, a.e. $(v,v_*) \in E[\zeta,\zeta_*,\zeta',\zeta'_*]$ and a.e. $\omega \in \mathbb{S}^2$,
$$
\mathcal{P}[f](v,v_*,\zeta,\zeta_*,\zeta',\zeta'_*,\omega) := \frac14 (f(v',\zeta') f(v'_*,\zeta'_*) - f(v,\zeta) f(v_*,\zeta_*)) \text{ln} \left( \frac{f(v',\zeta') f(v'_*,\zeta'_*)}{f(v,\zeta) f(v_*,\zeta_*)} \right),
$$
where $(v',v'_*) = S_{\omega}[\zeta,\zeta_*,\zeta',\zeta'_*] (v,v_*)$. Now by remarking that
$$
\forall x > 0, \quad (x - 1) \text{ln}(x) \geq 0,
$$
we have, since $f > 0$ and considering $\displaystyle x = \frac{f(v',\zeta') f(v'_*,\zeta'_*)}{f(v,\zeta) f(v_*,\zeta_*)} > 0$,
$$
\mathcal{P}[f](v,v_*,\zeta,\zeta_*,\zeta',\zeta'_*,\omega) = \frac14 \, f(v,\zeta) f(v_*,\zeta_*) \, (x - 1)\text{ln}(x) \geq 0.
$$
Using Proposition \ref{prop:weak} with $\psi = \text{ln}(f)$,
\begin{align*}
&\int_{\E} \int_{\R^3}  \B(f,f)(v,\zeta) \, \text{ln}(f)(v,\zeta) \, \text{d}v \, \text{d}\mu(\zeta) \\ &= -  \iiiint_{\E^4} \iint_{(\R^3)^2}  \int_{\mathbb{S}^2}   \mathcal{P}[f](v,v_*,\zeta,\zeta_*,\zeta',\zeta'_*,\omega) \, b(\cdot) \, \text{d}\omega \, \text{d}v  \,  \text{d}v_* \,  \text{d}\mu^{\otimes 4}(\zeta,\zeta_*,\zeta',\zeta'_*).
\end{align*}
Now from $b \geq 0$ and $\mathcal{P}[f] \geq 0$,
$$
\int_{\E} \int_{\R^3}  \B(f,f)(v,\zeta) \, \text{ln}(f)(v,\zeta) \, \text{d}v \, \text{d}\mu(\zeta) \leq 0.
$$

\bigskip

\noindent Let us now focus on the case of equality. Since the integrand $\mathcal{P}[f] \, b$ is non-negative, it must be equal to zero.

\begin{align*}
&\hspace{30pt} \int_{\E} \int_{\R^3}  \B(f,f)(v,\zeta) \, \text{ln}(f)(v,\zeta) \, \text{d}v \, \text{d}\mu(\zeta) = 0\\
&\implies \iiiint_{\E^4}   \iint_{E[\zeta,\zeta_*,\zeta',\zeta'_*]} \int_{\mathbb{S}^2} \mathcal{P}[f](v,v_*,\zeta,\zeta_*,\zeta',\zeta'_*,\omega) \, b(\cdot) \, \text{d}\omega \, \text{d}v  \,  \text{d}v_*   \, \text{d}\mu^{\otimes 4}(\zeta,\zeta_*,\zeta',\zeta'_*) = 0 \\
&\implies  \mathcal{P}[f](v,v_*,\zeta,\zeta_*,\zeta',\zeta'_*,\omega) \, b(v,v_*,\zeta,\zeta_*,\zeta',\zeta'_*,\omega) = 0 \quad \text{for a.e.  }v,v_*,\zeta,\zeta_*,\zeta',\zeta'_*,\omega \\
&\implies  \mathcal{P}[f](v,v_*,\zeta,\zeta_*,\zeta',\zeta'_*,\omega) = 0 \quad \text{for $\mu$-a.e. } \zeta,\zeta_*,\zeta',\zeta'_* \in \E, \text{  a.e. } (v,v_*) \in E[\zeta,\zeta_*,\zeta',\zeta'_*] \\
&\hspace{230pt}\text{ and a.e. }\omega \in \mathbb{S}^2[v,v_*,\zeta,\zeta_*,\zeta',\zeta'_*],
\end{align*}
where we recall that the set $\mathbb{S}^2[v,v_*,\zeta,\zeta_*,\zeta',\zeta'_*]$ is defined by Definition \ref{assumption:positivity}. Now since $f > 0$ and
$$
\forall x > 0, \quad (x - 1) \text{ln}(x) = 0 \implies x = 1,
$$
we finally obtain, writing $\mathbb{S}^2[\cdot] \equiv \mathbb{S}^2[v,v_*,\zeta,\zeta_*,\zeta',\zeta'_*]$ to lighten the notations,
\begin{align*}
&\hspace{30pt} \mathcal{P}[f](v,v_*,\zeta,\zeta_*,\zeta',\zeta'_*,\omega) = 0 \quad \text{for $\mu$-a.e. } \zeta,\zeta_*,\zeta',\zeta'_* \in \E, \text{   a.e. } (v,v_*) \in E[\zeta,\zeta_*,\zeta',\zeta'_*]  \text{ and a.e. }\omega \in \mathbb{S}^2[\cdot] \\
&\implies f(v',\zeta') f(v'_*,\zeta'_*) = f(v,\zeta) f(v_*,\zeta_*) \; \, \text{for $\mu$-a.e. } \zeta,\zeta_*,\zeta',\zeta'_* \in \E, \text{  a.e. } (v,v_*) \in E[\zeta,\zeta_*,\zeta',\zeta'_*],  \text{ a.e. }\omega \in \mathbb{S}^2[\cdot]\\
&\implies \text{ln}(f) \text{ is a collision invariant}.
\end{align*}
Note also that
\begin{align*}
&\hspace{30pt} \text{ln}(f) \text{ is a collision invariant} \\
&\implies f(v',\zeta') f(v'_*,\zeta'_*) = f(v,\zeta) f(v_*,\zeta_*) \; \, \text{for $\mu$-a.e. } \zeta,\zeta_*,\zeta',\zeta'_* \in \E, \text{  a.e. } (v,v_*) \in E[\zeta,\zeta_*,\zeta',\zeta'_*],  \text{ a.e. }\omega \in \mathbb{S}^2[\cdot]\\
&\implies \B(f,f)(v,\zeta) = 0 \quad \text{for a.e. } v \in \R^3 \text{ and $\mu$-a.e. } \zeta \in \E\\
&\implies \int_{\E} \int_{\R^3}  \B(f,f)(v,\zeta) \, \text{ln}(f)(v,\zeta) \, \text{d}v \, \text{d}\mu(\zeta) = 0,
\end{align*}
which ends the proof.
\QED
\end{proof}


\noindent Let us assume that there exists $f \in L^1(\R^3 \times \E, \, \text{d}v \, \text{d}\mu(\zeta))$ such that $f$ verifies assumptions of Theorem \ref{htheorem} and $\text{ln}(f)$ verifies the hypotheses of Theorem \ref{theoreminvariantgeneral}. From Theorem \ref{theoreminvariantgeneral}, $\exists \, (\alpha, \beta, \gamma) \in \R \times \R^3 \times \R$ such that for a.e. $v \in \R^3$ and $\mu$-a.e. $\zeta \in \E$,
$$
\text{ln}(f)(v,\zeta) = \alpha + \beta \cdot v + \gamma \left( \frac{m}{2} |v|^2 + \varepsilon(\zeta)  \right).
$$
Since $f \in L^1(\R^3 \times \E, \, \text{d}v \, \text{d}\mu(\zeta))$, we have $\gamma < 0$. Let us now set
$$
T = - \frac{1}{k_B \gamma}, \qquad u = - \frac{\beta}{ m \gamma}, \qquad A = \text{exp} \left(\alpha  - \frac{|\beta|^2}{2  m \gamma} \right),
$$
where $k_B$ is the Boltzmann constant. Then
$$
f(v,\zeta) = A \, \text{exp} \left( -\frac{ m|v- u|^2}{2 k_B T}  - \frac{\varepsilon(\zeta)}{k_B T}  \right).
$$
and it can be easily checked that $u$ is the average velocity of $f$, while $T$ will provide its temperature (see below).
Note that indeed $f \in L^1(\R^3 \times \E, \, \text{d}v \, \text{d}\mu(\zeta))$ thanks to Assumption \ref{assumption:integrability}. We define the density $\rho$
$$
\rho := \int_{\E} \int_{\R^3}  m \, f(v,\zeta) \, \text{d}v \, \text{d}\mu(\zeta).
$$
Then, almost everywhere,
$$
f = \mathcal{M}[\rho, u,T],
$$
where, for all $(v,\zeta) \in \R^3 \times  \E$,
\begin{equation*}
\mathcal{M}[\rho,u,T](v,\zeta) := \frac{\rho}{m} \; \frac{\displaystyle \text{exp} \left( -\frac{ m|v-u|^2}{2 k_B T}  - \frac{\varepsilon(\zeta)}{k_B T}  \right)}{\displaystyle \left( \frac{2 \pi k_B T}{ m} \right)^{3/2} \, \int_{\E} \text{exp} \left(-\frac{\varepsilon(\zeta_*)}{k_B T} \right) \text{d}\mu(\zeta_*)} ,
\end{equation*}
which we can rewrite
\begin{equation} \label{eq:generalmaxwell}
\mathcal{M}[\rho,u,T](v,\zeta) = \rho \; m^{1/2} \, (2 \pi k_B T)^{-3/2}  \, Z \left( \frac{1}{k_B T} \right)^{-1} \; \text{exp} \left( -\frac{ m|v-u|^2}{2 k_B T}  - \frac{\Bar{\varepsilon}(\zeta)}{k_B T}  \right),
\end{equation}
where $\Bar{\varepsilon}$ and $Z$ are defined by equations \eqref{eqdef:groundedenergyfunc} and \eqref{eq:partitionfunc}.


First, note that by definition of $\rho$,
$$
\int_{\E} \int_{\R^3}  m \, \mathcal{M}[\rho,u,T](v,\zeta) \, \text{d}v \, \text{d}\mu(\zeta) = \rho > 0.
$$
Also,
\begin{equation*}
\frac{1}{\rho} \int_{\E} \int_{\R^3}  m \, v \, \mathcal{M}[\rho,u,T](v,\zeta)  \, \text{d}v \, \text{d}\mu(\zeta) = m^{3/2} \, (2 \pi k_B T)^{-3/2} \; \int_{\R^3} v \, \text{exp} \left( -\frac{ m |v-u|^2}{2 k_B T} \right) \, \text{d}v = u.
\end{equation*}

\bigskip

\noindent The total energy density at collision equilibrium writes
\begin{align*}
e^{eq}[u,T] &:=  \frac{ m}{\rho} \,  \int_{\E} \int_{\R^3} \left( \frac{ m}{2} |v|^2 + \varepsilon(\zeta) \right) \, \mathcal{M}[\rho,u,T](v,\zeta) \, \text{d}v \, \text{d}\mu(\zeta) \\
&= m^{3/2} (2 \pi k_B T)^{-3/2} \; \int_{\R^3} \frac{m}{2} |v|^2 \, \text{exp} \left( -\frac{ m|v-u|^2}{2 k_B T} \right) \, \text{d}v + \frac{\displaystyle\int_{\E} \varepsilon(\zeta) \, \text{exp} \left(- \frac{\Bar{\varepsilon}(\zeta)}{k_B T} \right) \text{d}\mu(\zeta)}{\displaystyle\int_{\E} \text{exp} \left(- \frac{\Bar{\varepsilon}(\zeta)}{k_B T} \right) \text{d}\mu(\zeta)} ,
\end{align*}
and is finite thanks to Assumption \ref{assumption:integrability} and Proposition \ref{prop:ZisCinfty}. The term  associated with the velocity is well-known, and a classical computation gives
$$
m^{3/2} (2 \pi k_B T)^{-3/2} \; \int_{\R^3} \frac{m}{2} |v|^2 \, \text{exp} \left( -\frac{ m |v-u|^2}{2 k_B T} \right) \, \text{d}v = \frac{ m}{2} |u|^2 + \frac32 k_B T.
$$

\noindent Remark that
\begin{equation*}
\frac{\displaystyle \int_{\E} \varepsilon(\zeta) \, \text{exp} \left(- \frac{\Bar{\varepsilon}(\zeta)}{k_B T} \right) \text{d}\mu(\zeta)}{ \displaystyle \int_{\E} \text{exp} \left(- \frac{\Bar{\varepsilon}(\zeta)}{k_B T} \right) \text{d}\mu(\zeta)}  = \varepsilon^{0} + \frac{\displaystyle \int_{\E} \Bar{\varepsilon}(\zeta) \, \text{exp} \left(- \frac{\Bar{\varepsilon}(\zeta)}{k_B T} \right) \text{d}\mu(\zeta)}{\displaystyle \int_{\E} \text{exp} \left(- \frac{\Bar{\varepsilon}(\zeta)}{k_B T} \right) \text{d}\mu(\zeta)},
\end{equation*}
where we recall that $\varepsilon^0 = \text{inf ess}_{\mu} \{ \varepsilon \}$ and $\Bar{\varepsilon} = \varepsilon - \varepsilon^0$, see equation \eqref{eqdef:groundedenergyfunc}.

\smallskip

\begin{definition}
We define the (thermodynamic) number of internal degrees of freedom $\delta$ by
\begin{equation} \label{eq:averagefreedom}
\delta(T) := \frac{2}{k_B T} \left( \frac{\displaystyle\int_{\E} \Bar{\varepsilon}(\zeta) \, \text{exp} \left(- \frac{\Bar{\varepsilon}(\zeta)}{k_B T} \right) \text{d}\mu(\zeta)}{\displaystyle \int_{\E} \text{exp} \left(- \frac{\Bar{\varepsilon}(\zeta)}{k_B T} \right) \text{d}\mu(\zeta)} \right).
\end{equation}
\end{definition}
We notice that, since  $\Bar{\varepsilon} \geq 0$ $\mu$-a.e., $\delta \geq 0$.
\smallskip

\noindent The total energy density at collision equilibrium finally writes
$$
e^{eq}[u,T] =  \varepsilon^{0} +  \frac{m}{2} |u|^2 + \frac{3 + \delta(T)}{2} \, k_B T.
$$

\noindent $3 + \delta(T)$ is the \emph{total} (thermodynamic) number of degrees of freedom at temperature $T$, provided by the sum of the 3 translational degrees of freedom and of the number of internal ones $\delta(T)$.

\smallskip

\begin{definition}
We define the specific heat at constant volume $c_V$, for all $T > 0$,
\begin{equation}
c_V(T) := \frac{3 + D(T)}{2},
\end{equation}
where
$$
D(T) := \frac{\text{d} \big( T \delta(T) \big)}{\text{d}T}.
$$
\end{definition}

\smallskip

\begin{remark} \label{remark:deltaZ}
We recall that, from Proposition \ref{prop:ZisCinfty}, the partition function $Z$ defined by equation (\ref{eq:partitionfunc}) is $\mathcal{C}^{\infty}$. Remarking that
$$
T \delta(T) = - \frac{2}{k_B} \, \big( \text{ln} ( Z ) \big)' \left( \frac{1}{k_B \, T} \right),
$$
then the function $T \mapsto T \delta (T)$ is $\mathcal{C}^{\infty}$ on $\R_+^*$ and $D$ is thus well-defined.
\end{remark}

\smallskip
\noindent
In the following propositions we show various mathematical properties of the function $\delta(T)$.

\begin{proposition} \label{proposition:zeroenergy}
If $\mu(\varepsilon = \varepsilon^{0}) > 0$, then
$$
\delta(T) \underset{T \to 0^+}{\longrightarrow} 0 \quad \text{and} \quad D(T) \underset{T \to 0^+}{\longrightarrow} 0.
$$
\end{proposition}

\begin{proof}
First, note that $\mu(\varepsilon = \varepsilon^{0}) > 0 \implies \mu(\Bar{\varepsilon} = 0) > 0$. Since $\mu$-a.e., $\Bar{\varepsilon} \geq 0$, we then have for $k \in \N$
$$
\int_{\E}\, \left(\frac{\Bar{\varepsilon}}{k_B T} \right)^k \,  \text{exp}  \left(- \frac{\Bar{\varepsilon}(\zeta)}{k_B T} \right) \text{d}\mu(\zeta) = \mathbf{1}_{k=0} \times \mu(\Bar{\varepsilon} = 0) + \int_{\Bar{\varepsilon} > 0} \, \left(\frac{\Bar{\varepsilon}}{k_B T} \right)^k \, \text{exp} \left(- \frac{\Bar{\varepsilon}(\zeta)}{k_B T} \right) \text{d}\mu(\zeta).
$$
Now remark that for all $y > 0$, $e^{y/2} \geq 1 + \frac{y}{2} + \frac{y^2}{8}$. Thus
$$
y e^{-y} \leq 2 e^{-y / 2}  \quad \text{and} \quad y^2 e^{-y} \leq 8 e^{-y / 2}.
$$
Thus it comes
\begin{align*}
\delta(T) &= 2 \, \frac{\displaystyle\int_{\E}\, \frac{\Bar{\varepsilon}}{k_B T} \,  \text{exp}  \left(- \frac{\Bar{\varepsilon}(\zeta)}{k_B T} \right) \text{d}\mu(\zeta)}{\displaystyle \int_{\E} \text{exp} \left(- \frac{\Bar{\varepsilon}(\zeta)}{k_B T} \right) \text{d}\mu(\zeta)} \leq 2 \, \frac{ \displaystyle \int_{\Bar{\varepsilon} > 0}\, \frac{\Bar{\varepsilon}}{k_B T} \,  \text{exp}  \left(- \frac{\Bar{\varepsilon}(\zeta)}{k_B T} \right) \text{d}\mu(\zeta)}{\mu(\Bar{\varepsilon} = 0)} \\
&\leq \frac{4}{\mu(\Bar{\varepsilon} = 0)} \,  \int_{\Bar{\varepsilon} > 0}  \text{exp}  \left(- \frac{\Bar{\varepsilon}(\zeta)}{2 k_B T} \right) \text{d}\mu(\zeta).
\end{align*}
We can prove that (see proof of Proposition \ref{proposition:nbdegreesfreedom}) $0\leq D(T) \leq 2 \, \frac{\int_{\E}\, \left(\frac{\Bar{\varepsilon}}{k_B T} \right)^2 \,  \text{exp}  \left(- \frac{\Bar{\varepsilon}(\zeta)}{k_B T} \right) \text{d}\mu(\zeta)}{\int_{\E} \text{exp} \left(- \frac{\Bar{\varepsilon}(\zeta)}{k_B T} \right) \text{d}\mu(\zeta)} $, which implies
$$
0\leq D(T) \leq 2 \, \frac{\displaystyle \int_{\Bar{\varepsilon} > 0}\, \left(\frac{\Bar{\varepsilon}}{k_B T} \right)^2 \,  \text{exp}  \left(- \frac{\Bar{\varepsilon}(\zeta)}{k_B T} \right) \text{d}\mu(\zeta)}{\mu(\Bar{\varepsilon} = 0)} \leq \frac{16}{\mu(\Bar{\varepsilon} = 0)} \,  \int_{\Bar{\varepsilon} > 0}  \text{exp}  \left(- \frac{\Bar{\varepsilon}(\zeta)}{2 k_B T} \right) \text{d}\mu(\zeta).
$$
Now for all $\zeta \in \{ \Bar{\varepsilon} > 0\}$, $e^{-x \Bar{\varepsilon}(\zeta) } \underset{x \to \infty}{\longrightarrow} 0$. Since $(\zeta \mapsto e^{-x \Bar{\varepsilon}(\zeta)})_{x > 0}$ is a non-decreasing family of positive functions, we get by monotone convergence
$$
\int_{\Bar{\varepsilon} > 0}  e^{-x \Bar{\varepsilon}(\zeta)} \text{d}\mu(\zeta) \underset{x \to \infty}{\longrightarrow} 0.
$$
It follows that $\displaystyle \int_{\Bar{\varepsilon} > 0}  \text{exp}  \left(- \frac{\Bar{\varepsilon}(\zeta)}{2 k_B T} \right) \text{d}\mu(\zeta) \underset{T \to 0^+}{\longrightarrow} 0$, and thus, since $\mu(\bar{\varepsilon} = 0) > 0$,
$$
0 \leq \delta (T) \leq \frac{4}{\mu(\Bar{\varepsilon} = 0)} \,  \int_{\Bar{\varepsilon} > 0}  \text{exp}  \left(- \frac{\Bar{\varepsilon}(\zeta)}{2 k_B T} \right) \text{d}\mu(\zeta) \underset{T \to 0^+}{\longrightarrow} 0
$$
and
$$
0 \leq D (T) \leq \frac{16}{\mu(\Bar{\varepsilon} = 0)} \,  \int_{\Bar{\varepsilon} > 0}  \text{exp}  \left(- \frac{\Bar{\varepsilon}(\zeta)}{2 k_B T} \right) \text{d}\mu(\zeta) \underset{T \to 0^+}{\longrightarrow} 0,
$$
so that finally
$$
\delta(T) \underset{T \to 0^+}{\longrightarrow} 0 \quad \text{and} \quad D(T) \underset{T \to 0^+}{\longrightarrow} 0.
$$
\QED
\end{proof}

\begin{proposition} \label{proposition:constantvolume}
For all $T > 0$,
$$
\delta(T) = \frac{1}{T} \int_{0}^{T} D(T') \, \emph{d}T'.
$$
\end{proposition}

\begin{proof}
We have for $T_1,T_2 >0$, $\displaystyle \int_{T_1}^{T_2} D(T') \, \text{d}T' = T_2 \delta(T_2) - T_1 \delta(T_1)$. We thus want to prove that $T \delta(T) \underset{T \to 0^+}{\longrightarrow} 0$. We proceed by disjunction of case. First, if $\mu(\varepsilon = \varepsilon^{0}) > 0$, then from Proposition \ref{proposition:zeroenergy}, $\delta(T) \underset{T \to 0^+}{\longrightarrow} 0$, so we also have $T \delta(T) \underset{T \to 0^+}{\longrightarrow} 0$. If on the other hand $\mu(\varepsilon = \varepsilon^{0}) = 0$, then since $\varepsilon^0 = \text{inf ess}_{\mu}\{ \varepsilon \}$, we have $\mu(0 < \Bar{\varepsilon} \leq r ) = \mu(\varepsilon^{0} < \varepsilon \leq \varepsilon^{0} + r ) > 0$ for all $r > 0$. Then from Lemma \ref{lemma:equiv} in Appendix, for all $r > 0$,
\begin{align*}
\int_{\E} \Bar{\varepsilon}(\zeta) e^{-x \Bar{\varepsilon}(\zeta)} \, \text{d}\mu(\zeta) &\underset{x \to \infty}{=} \int_{\Bar{\varepsilon} \leq r } \Bar{\varepsilon}(\zeta) e^{-x \Bar{\varepsilon}(\zeta)} \, \text{d}\mu(\zeta) +  o \left( \int_{\Bar{\varepsilon} \leq r} \Bar{\varepsilon}(\zeta) e^{-x \Bar{\varepsilon}(\zeta)} \, \text{d}\mu(\zeta) \right), \\
\text{and} \quad \int_{\E}  e^{-x \Bar{\varepsilon}(\zeta)} \, \text{d}\mu(\zeta) &\underset{x \to \infty}{=} \int_{\Bar{\varepsilon} \leq r } e^{-x \Bar{\varepsilon}(\zeta)} \, \text{d}\mu(\zeta) +  o \left( \int_{\Bar{\varepsilon} \leq r} e^{-x \Bar{\varepsilon}(\zeta)} \, \text{d}\mu(\zeta) \right).
\end{align*}
Since $\mu(\{ 0 < \Bar{\varepsilon} \leq r \}) > 0$ and $\Bar{\varepsilon} \geq 0$ $\mu$-a.e., $\int_{\Bar{\varepsilon} \leq r} \Bar{\varepsilon}(\zeta) e^{-x \Bar{\varepsilon}(\zeta)} \, \text{d}\mu(\zeta)  > 0$, so that
$$
\int_{\E} \Bar{\varepsilon}(\zeta) e^{-x \Bar{\varepsilon}(\zeta)} \, \text{d}\mu(\zeta) \underset{x \to \infty}{\sim} \int_{\Bar{\varepsilon} \leq r } \Bar{\varepsilon}(\zeta) e^{-x \Bar{\varepsilon}(\zeta)} \, \text{d}\mu(\zeta).
$$
Also, $\mu(\{ \Bar{\varepsilon} \leq r \}) > 0$ implies $\int_{\Bar{\varepsilon} \leq r} e^{-x \Bar{\varepsilon}(\zeta)} \, \text{d}\mu(\zeta)  > 0$, so that
$$
\int_{\E} e^{-x \Bar{\varepsilon}(\zeta)} \, \text{d}\mu(\zeta) \underset{x \to \infty}{\sim} \int_{\Bar{\varepsilon} \leq r } e^{-x \Bar{\varepsilon}(\zeta)} \, \text{d}\mu(\zeta).
$$
Hence
$$
\frac{\int_{\E} \Bar{\varepsilon}(\zeta) e^{-x \Bar{\varepsilon}(\zeta)} \, \text{d}\mu(\zeta)}{\int_{\E} e^{-x \Bar{\varepsilon}(\zeta)} \, \text{d}\mu(\zeta)} \underset{x \to \infty}{\sim} \frac{\int_{\Bar{\varepsilon} \leq r } \Bar{\varepsilon}(\zeta) e^{-x \Bar{\varepsilon}(\zeta)} \, \text{d}\mu(\zeta)}{\int_{\Bar{\varepsilon} \leq r } e^{-x \Bar{\varepsilon}(\zeta)} \, \text{d}\mu(\zeta)}.
$$
Thus there exists $x_r > 0$ such that for all $x \geq x_r$,
$$
0 \leq \frac{\int_{\E} \Bar{\varepsilon}(\zeta) e^{-x \Bar{\varepsilon}(\zeta)} \, \text{d}\mu(\zeta)}{\int_{\E} e^{-x \Bar{\varepsilon}(\zeta)} \, \text{d}\mu(\zeta)} \leq 2\, \frac{\int_{\Bar{\varepsilon} \leq r } \Bar{\varepsilon}(\zeta) e^{-x \Bar{\varepsilon}(\zeta)} \, \text{d}\mu(\zeta)}{\int_{\Bar{\varepsilon} \leq r } e^{-x \Bar{\varepsilon}(\zeta)} \, \text{d}\mu(\zeta)}
\leq 2 \, r.
$$
Since it is true for all $r>0$, we finally get
$$
\frac{\int_{\E} \Bar{\varepsilon}(\zeta) e^{-x \Bar{\varepsilon}(\zeta)} \, \text{d}\mu(\zeta)}{\int_{\E} e^{-x \Bar{\varepsilon}(\zeta)} \, \text{d}\mu(\zeta)} \underset{x \to \infty}{\longrightarrow} 0.
$$
It follows that in any case
$$
T \delta(T) = \frac{2}{k_B} \, \frac{\int_{\E} \Bar{\varepsilon}(\zeta) e^{-\frac{\Bar{\varepsilon}(\zeta)}{k_B T}} \, \text{d}\mu(\zeta)}{\int_{\E} e^{-\frac{\Bar{\varepsilon}(\zeta)}{k_B T}} \, \text{d}\mu(\zeta)} \underset{T \to 0^+}{\longrightarrow} 0.
$$

\QED

\end{proof}

\begin{remark} \label{remark:whyisZdefinedwithepsilonbar}
This Proposition \ref{proposition:constantvolume} along with Remark \ref{remark:deltaZ} justify our choice to define the partition function by $\displaystyle \int_{\E} e^{-\beta \Bar{\varepsilon}(\zeta)} \, \text{d}\mu(\zeta)$ (with $\Bar{\varepsilon}$) instead of simply $\displaystyle \int_{\E} e^{-\beta \varepsilon(\zeta)} \, \text{d}\mu(\zeta)$ (with $\varepsilon$). Indeed, if the second definition were taken, and $\delta$ were still defined by $- \frac{2}{k_B T} \, \big( \text{ln} ( Z ) \big)' \left( \frac{1}{k_B \, T} \right)$, then $\varepsilon^0 \neq 0$ would imply $\displaystyle \delta(T) \underset{T \to 0}{\sim} \frac{2 \, \varepsilon^0}{k_B T}$, which is non-physical. On the other hand, choosing $\bar{\varepsilon}$ in the definition always yields $T \delta(T) \underset{T \to 0}{\longrightarrow} 0$, and thus the formula $\displaystyle \delta(T) = \frac{1}{T} \int_{0}^{T} D(T') \, \text{d}T'$, which is physically coherent.
\end{remark}

\smallskip

\noindent As a consequence of Proposition \ref{proposition:constantvolume}, the total energy at the equilibrium may be cast as
$$
e^{eq}[u,T] =  \varepsilon^{0} + \frac{m}{2}  |u|^2 +  \int_{0}^{T} c_V(T') \, k_B \, \text{d}T'.
$$

\noindent
A probabilistic interpretation may also be provided for the functions $\delta(T)$ and $D(T)$.
Indeed, for $T>0$, we define the Gibbs (probability) measure $\nu_{T}$ on $(\E,\A)$ by
\begin{equation} \label{eq:gibbsmeasure}
\frac{\text{d}\nu_{T}}{\text{d}\mu}(\zeta) = Z \left(\frac{1}{k_B T} \right)^{-1} \,\text{exp} \left(- \frac{\Bar{\varepsilon}(\zeta)}{k_B T} \right),
\end{equation}
where $Z$ is the partition function defined by equation \eqref{eq:partitionfunc}; then for all $T >0$, $(\E,\A,\nu_{T})$ is a probability space. Since $\Bar{\varepsilon} : \E \to \R$ is $(\A,\text{Bor}(\R))$-measurable, it is a real random variable on $(\E,\A,\nu_{T})$.

\medskip

\begin{proposition} \label{proposition:nbdegreesfreedom}
For all $T > 0$,
\begin{equation} \label{eq:degreesexpectation}
\delta(T) = 2 \,  \mathbb{E}_{\nu_{T}} \left[ \frac{\Bar{\varepsilon}}{k_B T} \right],
\end{equation}
and
\begin{equation} \label{eq:degreesvariance}
D(T) = 2 \, \emph{Var}_{\nu_{T}} \left[ \frac{\Bar{\varepsilon}}{k_B T} \right],
\end{equation}
where $\mathbb{E}_{\nu_T}$ and $\emph{Var}_{\nu_{T}}$ are respectively the expectation and the variance under the probability $\nu_T$.
\end{proposition}

\begin{proof}
The following computations are possible thanks to Assumption \ref{assumption:integrability}. The first part of the proposition comes from
$$
\frac{\int_{\E} \Bar{\varepsilon}(\zeta) \, \text{exp} \left(- \frac{\Bar{\varepsilon}(\zeta)}{k_B T} \right) \text{d}\mu(\zeta)}{\int_{\E} \text{exp} \left(- \frac{\Bar{\varepsilon}(\zeta)}{k_B T} \right) \text{d}\mu(\zeta)} = \int_{\E} \Bar{\varepsilon}(\zeta) \,  Z \left(\frac{1}{k_B T} \right)^{-1} \,\text{exp} \left(- \frac{\Bar{\varepsilon}(\zeta)}{k_B T} \right) \text{d}\mu(\zeta) = \int_{\E} \Bar{\varepsilon}(\zeta) \, \text{d}\nu_{T}(\zeta) = \mathbb{E}_{\nu_{T}} [\Bar{\varepsilon}].
$$
For equation \eqref{eq:degreesvariance}, we remark that from Proposition \ref{prop:ZisCinfty},
$$
\mathbb{E}_{\nu_{T}} [\Bar{\varepsilon}] = \frac{\int_{\E} \Bar{\varepsilon}(\zeta) \, \text{exp} \left(- \frac{\Bar{\varepsilon}(\zeta)}{k_B T} \right) \text{d}\mu(\zeta)}{\int_{\E} \text{exp} \left(- \frac{\Bar{\varepsilon}(\zeta)}{k_B T} \right) \text{d}\mu(\zeta)} = \frac{- Z' \left(\frac{1}{k_B T} \right)}{Z \left(\frac{1}{k_B T} \right)}, \quad  \mathbb{E}_{\nu_{T}} \left[\Bar{\varepsilon}^2 \right] = \frac{\int_{\E} \Bar{\varepsilon}(\zeta)^2 \, \text{exp} \left(- \frac{\Bar{\varepsilon}(\zeta)}{k_B T} \right) \text{d}\mu(\zeta)}{\int_{\E} \text{exp} \left(- \frac{\Bar{\varepsilon}(\zeta)}{k_B T} \right) \text{d}\mu(\zeta)} = \frac{Z'' \left(\frac{1}{k_B T} \right)}{Z \left(\frac{1}{k_B T} \right)}.
$$
It follows that
\begin{align*}
D(T) &= 2 \, \frac{\text{d}}{k_B \, \text{d}T} \left(   \frac{- Z' \left(\frac{1}{k_B T} \right)}{Z \left(\frac{1}{k_B T} \right)}   \right) =  \frac{2}{(k_B T)^2} \left(   \frac{ Z'' \left(\frac{1}{k_B T} \right)}{Z \left(\frac{1}{k_B T} \right)} - \left(\frac{-Z' \left(\frac{1}{k_B T} \right)}{Z \left(\frac{1}{k_B T} \right)}\right)^2  \right) \\ &= \frac{2}{(k_B T)^2} \left(  \mathbb{E}_{\nu_{T}} \left[ \Bar{\varepsilon}^2 \right]  - (\mathbb{E}_{\nu_{T}} [\Bar{\varepsilon}])^2  \right) =\frac{2}{(k_B T)^2} \, \text{Var}_{\nu_{T}} \left[\Bar{\varepsilon} \right]
= 2 \, \text{Var}_{\nu_{T}} \left[ \frac{\Bar{\varepsilon}}{k_B T} \right].
\end{align*}
\QED
\end{proof}

\smallskip

\begin{corollary} \label{corollary:Dpositif}
For all $T>0$, $D(T) \geq 0$.
\end{corollary}

\begin{proof}
It is an immediate consequence of Proposition \ref{proposition:nbdegreesfreedom}.
\QED
\end{proof}

\begin{corollary} \label{corollary:boundedenergy}
If there exists $R \in \R$ such that $\varepsilon \leq R$ $\mu$-a.e., then
$$
\delta(T) \underset{T \to \infty}{\longrightarrow} 0, \quad \text{and} \quad D(T) \underset{T \to \infty}{\longrightarrow} 0.
$$
\end{corollary}
\begin{proof}
Necessarily $R \geq \varepsilon^{0}$, and from Proposition \ref{proposition:nbdegreesfreedom}
$$
0 \leq \delta(T) = 2 \, \mathbb{E}_{\nu_{T}} \left[ \frac{\Bar{\varepsilon}}{k_B T} \right] \leq \frac{2 (R- \varepsilon^{0})}{k_B T} \underset{T \to \infty}{\longrightarrow} 0,
$$
and
$$
0 \leq D(T) = 2 \, \text{Var}_{\nu_{T}} \left[ \frac{\Bar{\varepsilon}}{k_B T} \right] \leq 2 \, \mathbb{E}_{\nu_{T}} \left[ \left(\frac{\Bar{\varepsilon}}{k_B T} \right)^2 \right] \leq \frac{2 (R- \varepsilon^{0})^2}{(k_B T)^2} \underset{T \to \infty}{\longrightarrow} 0.
$$
\QED
\end{proof}

\noindent
We define now the function $\Theta$, for all $T \in \R_+^*$, as
\begin{equation} \label{eqdef:theta}
\Theta(T) :=  \int_{0}^{T} c_V(T') \, k_B \, \text{d}T',
\end{equation}
such that
$$
e^{eq}[u,T] =  \varepsilon^{0} +  \frac{m}{2} |u|^2 + \Theta(T).
$$
\begin{proposition}
$\Theta$ is continuous on $\R_+^*$, can be extended by continuity to $\R_+$, setting $\Theta(0) = 0$, increasing on $\R_+$, with $\Theta(T) \underset{T \rightarrow \infty}{\longrightarrow} \infty$, and thus a bijection from $\R_+$ to $\R_+$.
\end{proposition}

\begin{proof}
The result is immediate, since from Corollary \ref{corollary:Dpositif} we have $\displaystyle c_V \geq \frac{3}{2}$.
\QED
\end{proof}

\smallskip

\noindent We then have
$$
T =  \Theta^{-1} \left( \frac{ m}{\rho} \,  \int_{\E}  \int_{\R^3}\left( \frac{m}{2} |v- u|^2 + \Bar{\varepsilon}(\zeta) \right) \, f(v,\zeta) \, \text{d}v \, \text{d}\mu(\zeta) \right),
$$
which explains and rigorously justifies the definition of temperature stated in equation \eqref{eqdef:temperature}.

\section{Euler limit} \label{section:euler}

Let us assume that there exists $f$ solution to the Boltzmann equation \eqref{eq:boltzmann} such that for all $t,x$, $f(t,x,\cdot,\cdot)$ is a Maxwellian, that is
$$
f(t,x,v,\zeta) = \mathcal{M}[\rho(t,x),u(t,x),T(t,x)](v,\zeta).
$$
Since
$$
\partial_t f(t,x,v,\zeta) + v \cdot \nabla_x f(t,x,v,\zeta) = \B(f,f)(t,x,v,\zeta)
$$
and
$$
\int_{\E} \int_{\R^3} \begin{pmatrix}  m \\  m v \\ \frac{m}{2}|v|^2 + \varepsilon(\zeta) \end{pmatrix} \B(f,f)(t,x,v,\zeta) \, \text{d}v \, \text{d}\mu(\zeta) = 0,
$$
we have
\begin{align*}
\partial_t  \left( \int_{\E} \int_{\R^3}  m \, f(t,x,v,\zeta) \, \text{d}v \, \text{d}\mu(\zeta) \right) &+ \text{div}_x \left( \int_{\E} \int_{\R^3}  m \, v f(t,x,v,\zeta) \, \text{d}v \, \text{d}\mu(\zeta) \right)  = 0\\
\partial_t  \left( \int_{\E} \int_{\R^3}  m \, v f(t,x,v,\zeta) \, \text{d}v \, \text{d}\mu(\zeta) \right) &+ \text{div}_x \left(  \int_{\E} \int_{\R^3}  m \, v \otimes v f(t,x,v,\zeta) \, \text{d}v \, \text{d}\mu(\zeta) \right)  = 0\\
\partial_t  \Bigg(  \int_{\E} \int_{\R^3} \left(\frac{ m}{2} |v|^2 + \varepsilon(\zeta) \right)f(t,x,v,\zeta) \, \text{d}v \, \text{d}\mu(\zeta) \Bigg) &+  \text{div}_x \Bigg( \int_{\E} \int_{\R^3} v \left(\frac{ m}{2}  |v|^2 + \varepsilon(\zeta) \right) f(t,x,v,\zeta) \, \text{d}v \, \text{d}\mu(\zeta) \Bigg)  = 0.
\end{align*}
Since $f(t,x,v,\zeta) = \mathcal{M}[\rho(t,x),u(t,x),T(t,x)](v,\zeta)$, we have from the definitions of mass, momentum, and total energy
\begin{align*}
\partial_t \rho &+  \text{div}_x \left( \rho \,  u \right) = 0\\
\partial_t  \left( \rho \, u \right) &+ \text{div}_x \left(  \int_{\E} \int_{\R^3}  m \, v \otimes v \mathcal{M}[\rho,u,T](v,\zeta) \, \text{d}v \, \text{d}\mu(\zeta) \right)  = 0\\
\partial_t  \left(  \rho \, \varepsilon^{0} +  \frac{1}{2} \rho |u|^2 + \frac{\rho}{m} \, \frac{3 + \delta(T)}{2} \, k_B T \right) &+  \text{div}_x \Bigg( \int_{\E} \int_{\R^3} v \left(\frac{ m}{2}  |v|^2 + \varepsilon(\zeta) \right) \mathcal{M}[\rho,u,T](v,\zeta) \, \text{d}v \, \text{d}\mu(\zeta) \Bigg)  = 0.
\end{align*}
We are left with two terms to compute. First, by a classical argument (see \cite{bouchut2000kinetic}, Chapter 2),
\begin{align*}
\int_{\E} \int_{\R^3}  m \, v \otimes v \, \mathcal{M}[\rho,u,T](v,\zeta) \, \text{d}v \, \text{d}\mu(\zeta) &=   \rho \, m^{1/2} \,  (2 \pi k_B T)^{-3/2} \, \int_{\R^3}  m \, v \otimes v \,  \text{exp} \left( - \frac{ m |v-u|^2}{2 k_B T} \right) \, \text{d}v \\
&= \rho \, u \otimes u + \frac{\rho}{ m} \, k_B T \, \mathbf{I_d},
\end{align*}
where $\mathbf{I_d}$ is the identity matrix. Also by a classical argument,
\begin{align*}
&\int_{\E} \int_{\R^3} v \left(\frac{ m}{2}  |v|^2 + \varepsilon(\zeta) \right) \mathcal{M}[\rho,u,T](v,\zeta) \, \text{d}v \, \text{d}\mu(\zeta) \\
&=  \rho \, m^{1/2} \, (2 \pi k_B T)^{-3/2} \, \int_{\R^3} v \frac{ m}{2}  |v|^2 \, \text{exp} \left( - \frac{ m|v-u|^2}{2 k_B T} \right) \, \text{d}v + \int_{\E} \int_{\R^3} v \, \varepsilon(\zeta) \mathcal{M}[\rho,u,T](v,\zeta) \, \text{d}v \, \text{d}\mu(\zeta) \\
&= \frac12 \rho \, |u|^2 u + \frac{\rho}{ m} \, k_B T \, u + \frac{\rho}{ m} \, \frac{3}{2} k_B T \, u + \frac{\rho}{ m} \, \left(\frac{\int_{\E} \varepsilon(\zeta) \, \text{exp} \left(- \frac{\Bar{\varepsilon}(\zeta)}{k_B T} \right) \text{d}\mu(\zeta)}{\int_{\E} \text{exp} \left(- \frac{\Bar{\varepsilon}(\zeta)}{k_B T} \right) \text{d}\mu(\zeta)}  \right) u \\
&= \frac12 \rho \, |u|^2 u + \frac{\rho}{ m} \, k_B T \, u + \frac{\rho}{ m} \, \frac{3 + \delta(T)}{2} k_B T \, u +  \frac{\rho}{ m} \, \varepsilon^{0} \, u.
\end{align*}
We set
$$
p(t,x) := \frac{\rho}{ m} \, k_B T(t,x) \quad \text{and} \quad \theta(t,x) := 
\frac{3 + \delta(T(t,x))}{2} k_B T(t,x).
$$
Remarking that $\partial_t (\rho \, \varepsilon^0) +  \text{div}_x \left( \rho \, \varepsilon^0 \, u \right) = \varepsilon^0 \big(\partial_t \rho +  \text{div}_x \left( \rho \,  u \right) \big) = 0$,  we obtain the compressible Euler set of equations, with $\rho$ the density, $u$ the velocity, $p$ the pressure and $\theta$ the specific internal energy density.
\begin{equation} \label{eq:eulergeneral}
\begin{cases}
\displaystyle \hspace{62.4pt} \partial_t \rho +  \text{div}_x \left( \rho \,  u \right) = 0\\
\displaystyle \hspace{45.6pt} \partial_t  \left( \rho \, u \right) + \text{div}_x \left( \rho \, u \otimes u \right) + \nabla_x p  = 0\\
\displaystyle \partial_t  \left( \frac12 \rho |u|^2 + \frac{\rho}{ m} \, \theta \right) +  \text{div}_x \left( \frac12 \rho |u|^2 u + \frac{\rho}{ m} \, \theta\, u + p \, u \right)  = 0.
\end{cases}
\end{equation}

\section{Equipartition Theorem} \label{section:combine}

In this section we show how is possible to combine different measure spaces of internal states into a unique kinetic model for the considered polyatomic gas.

\begin{theorem} \textbf{Equipartition Theorem.} \label{theorem:equipartition}
Let $L \in \N^*$. Consider $(\E_l,\A_l,\mu_l)_{1 \leq l \leq L}$, and for all $1\leq l \leq L $, $\varepsilon_l : (\E_l, \A_l) \to (\R,\emph{Bor}(\R))$ measurable such that Assumptions \ref{assumption:minimum} and \ref{assumption:integrability} hold for all $(\mu_l,\varepsilon_l)$. Let us define
\begin{align*}
\E = \E_1 \times \dots \times \E_L,\qquad \A = \A_1 \otimes \dots \otimes \A_L, \qquad \mu = \mu_1 \otimes \dots \otimes \mu_L
\end{align*}
and
$$
\forall \zeta = (\zeta_1, \dots, \zeta_L) \in \E, \qquad \varepsilon(\zeta) = \sum_{l=1}^L  \varepsilon_l(\zeta_l).
$$
Then Assumptions \ref{assumption:minimum} and \ref{assumption:integrability} hold for $(\mu,\varepsilon)$, and for all $T > 0$, the Gibbs measure defined in \eqref{eq:gibbsmeasure} turns out to be
$$
\nu_{T} = \nu^1_{T} \otimes \dots \otimes \nu^L_{T}
$$
and moreover one has
$$
\delta(T) = \sum_{l=1}^L \delta_{l}(T), \qquad D(T) = \sum_{l=1}^L D_l(T).
$$
\end{theorem}

\begin{proof}

Let $T > 0$. Let us first prove that $\nu_{T} = \nu^1_{T} \otimes \dots \otimes \nu^L_{T}$. We first remark that $\Bar{\varepsilon}(\zeta) = \sum_{l=1}^L  \Bar{\varepsilon}_l(\zeta_l)$. Let $A = (A_1,\dots,A_L) \in \A_1 \times \dots \times \A_L$. Then
\begin{align*}
\int_{A}  \text{exp} \left(- \frac{\Bar{\varepsilon}(\zeta)}{k_B T} \right) \text{d}\mu(\zeta) &= \int_{A_1} \dots \int_{A_L} \left(  \text{exp} \left( -  \frac{\sum_{l=1}^L \Bar{\varepsilon}_l(\zeta_l)}{k_B T} \right) \right) \, \text{d}\mu_1(\zeta_1) \dots \text{d}\mu_L(\zeta_L) \\
&= \int_{A_1} \dots \int_{A_L} \left( \prod_{l=1}^L \text{exp} \left(- \frac{\Bar{\varepsilon}_l(\zeta_l)}{k_B T} \right) \right) \, \text{d}\mu_1(\zeta_1) \dots \text{d}\mu_L(\zeta_L) \\
&= \prod_{l=1}^L \int_{A_l}   \text{exp} \left( -  \frac{ \Bar{\varepsilon}_l(\zeta_l)}{k_B T} \right)  \, \text{d}\mu_l(\zeta_l).
\end{align*}
Since this also holds for $A = \E = \E_1 \times \dots \times \E_L$, we then have
$$
\nu_T(A) = \prod_{l=1}^L \frac{ \int_{A_l}   \text{exp} \left( -  \frac{ \Bar{\varepsilon}_l(\zeta_l)}{k_B T} \right)  \, \text{d}\mu_l(\zeta_l)}{\int_{\E_l}   \text{exp} \left( -  \frac{ \Bar{\varepsilon}_l(\zeta_l)}{k_B T} \right)  \, \text{d}\mu_l(\zeta_l)} = \prod_{l=1}^L \nu^l_T(A_l),
$$
that is $\nu_T = \nu^1_T \otimes \dots \otimes \nu^L_T$.

\smallskip

\noindent Each $\varepsilon_l$ (thus also $\Bar{\varepsilon}_l$) can also be seen as a real random variable on $(\E,\A,\nu_T)$, by setting $\varepsilon_l(\zeta) \equiv \varepsilon_l(\zeta_l)$. It follows that
$$
\delta(T) = 2 \, \mathbb{E}_{\nu_{T}} \left[ \frac{\Bar{\varepsilon}}{k_B T} \right] =  2 \, \mathbb{E}_{\nu_{T}} \left[ \sum_{l=1}^L \frac{\Bar{\varepsilon}_l}{k_B T} \right] = 2 \sum_{l=1}^L \mathbb{E}_{\nu_{T}} \left[ \frac{\Bar{\varepsilon}_l}{k_B T} \right] = \sum_{l=1}^L \delta_l(T).
$$
Finally, since $\varepsilon_1, \dots, \varepsilon_L$ are independent under $\nu_T$ ($\nu_T$ is tensorized and they each depend on a different variable),
$$
D(T) = 2 \, \text{Var}_{\nu_{T}} \left[ \frac{\Bar{\varepsilon}}{k_B T} \right] =  2 \, \text{Var}_{\nu_{T}} \left[ \sum_{l=1}^L \frac{\Bar{\varepsilon}_l}{k_B T} \right] = 2 \sum_{l=1}^L \text{Var}_{\nu_{T}} \left[ \frac{\Bar{\varepsilon}_l}{k_B T} \right] =  \sum_{l=1}^L D_l(T).
$$
\QED
\end{proof}

\smallskip
We show now the results relevant to some particular options.

\begin{proposition} \label{proposition:alphabetadelta}
Let $C_1,C_2 > 0$, $\beta > 0$ and $\alpha > - 1$. Denoting by $\emph{Bor}(\R)$ the set Borelians of $\R$, we consider
$$
(\E, \A, \emph{d}\mu(\zeta)) = \left( \R, \emph{Bor}(\R), C_1 |z|^\alpha \, \emph{d}z \right), \quad \varepsilon(z) = C_2 |z|^\beta.
$$
Then for all $T > 0$,
$$
\delta(T) = D(T) = \frac{2 (\alpha + 1)}{\beta}.
$$
\end{proposition}

\begin{proof}
Set $T > 0$ and $\alpha > - 1$. Performing the change of variables $x = C_2 |z|^{\beta}$, and then $y = x/(k_B T)$, we get
\begin{align*}
\int_{\R} C_2 |z|^{\beta} \, \text{exp} \left( -\frac{C_2 |z|^{\beta}}{k_B T} \right) \, C_1 |z|^{\alpha}\, \text{d}z &= \frac{2 C_1 C_2^{\frac{\alpha - 1}{\beta}}}{\beta} \int_0^{\infty} x \, \text{exp} \left( -\frac{x}{k_B T} \right) \, x^{\frac{\alpha}{\beta}} \,  x^{\frac{1- \beta}{\beta}} \text{d}x \\
&= \frac{2 C_1 C_2^{\frac{\alpha - 1}{\beta}}}{\beta} \int_0^{\infty} x^{\frac{\alpha + 1}{\beta} } \, \text{exp} \left( -\frac{x}{k_B T} \right) \,  \text{d}x \\
 &= \frac{2 C_1 C_2^{\frac{\alpha - 1}{\beta}}}{\beta} (k_B T)^{\frac{\alpha  + 1}{\beta} + 1} \int_0^{\infty} y^{\frac{\alpha + 1}{\beta}} \, e^{-y} \,  \text{d}y \\
 &= \frac{2 C_1 C_2^{\frac{\alpha - 1}{\beta}}}{\beta} (k_B T)^{\frac{\alpha  + 1}{\beta} + 1} \left(\frac{\alpha + 1}{\beta} \right) \int_0^{\infty} y^{\frac{\alpha + 1}{\beta} - 1} \, e^{-y} \,  \text{d}y.
\end{align*}
Analogously, performing the same changes of variables
\begin{align*}
\int_{\R} \text{exp} \left( -\frac{C_2 |z|^{\beta}}{k_B T} \right) \, C_1 |z|^{\alpha}\, \text{d}z &= \frac{2 C_1 C_2^{\frac{\alpha - 1}{\beta}}}{\beta} \int_0^{\infty}  \text{exp} \left( -\frac{x}{k_B T} \right) \, x^{\frac{\alpha}{\beta}} \, x^{\frac{1- \beta}{\beta}} \text{d}x \\
&= \frac{2 C_1 C_2^{\frac{\alpha - 1}{\beta}}}{\beta} \int_0^{\infty} x^{\frac{\alpha + 1}{\beta} -1} \, \text{exp} \left( -\frac{x}{k_B T} \right) \,  \text{d}x \\
&= \frac{2 C_1 C_2^{\frac{\alpha - 1}{\beta}}}{\beta} (k_B T)^{\frac{\alpha  + 1}{\beta}} \int_0^{\infty} y^{\frac{\alpha + 1}{\beta}-1} \, e^{-y} \,  \text{d}y.
\end{align*}
We deduce that
$$
\delta(T) = \frac{2}{k_B T} \, \frac{\int_{\R} C_2 |z|^{\beta} \, \text{exp} \left( -\frac{C_2 |z|^{\beta}}{k_B T} \right) \, C_1 |z|^{\alpha}\, \text{d}z}{\int_{\R} \text{exp} \left( -\frac{C_2 |z|^{\beta}}{k_B T} \right) \, C_1 |z|^{\alpha}\, \text{d}z} = \frac{2(\alpha + 1)}{\beta}.
$$
Since $\delta$ does not depend on $T$, we have $\displaystyle D(T) = \delta(T) = \frac{2(\alpha + 1)}{\beta}$.
\QED
\end{proof}

\smallskip

\begin{corollary} \label{corollary:equipartition} 
We consider $a_1,\dots,a_d > 0$ and
$$
(\E, \A, \emph{d}\mu(\zeta)) = \left( \R^d, \emph{Bor}(\R^d), \emph{d}z \right), \qquad \varepsilon(z) = \sum_{l = 1}^d a_l \, z_l^2.
$$
Then for all $T > 0$,
$$
\delta(T) = D(T) = d.
$$
\end{corollary}

\begin{proof}
This Corollary is a consequence of Proposition \ref{proposition:alphabetadelta} with $\alpha=0$, $\beta = 2$, and Proposition \ref{theorem:equipartition} with $(\E_l,\A_l,\text{d}\mu_l) = \left( \R, \text{Bor}(\R), \text{d}z \right)$  and $\varepsilon_l(z) = a_l \, z^2$ for $1\leq l \leq d$.
\QED
\end{proof}

\section{Models within the framework} \label{section:application}

\subsection{Existing models} \label{subsection:existingmodels}
In this subsection, we show that our framework encapsulates various existing models. The parameters to choose in order to build the model are $(\E,\A,\mu)$ and $\varepsilon$.

\subsubsection*{The monoatomic gas}
The model for the monoatomic gas is the classical Boltzmann model, which considers elastic collisions only. We refer the reader to \cite{bouchut2000kinetic}. To recover this model in our framework, simply set $\E = \{0\}$, $\A = \mathcal{P}(\E)$, $\mu(\{0\}) = 1$ and $\varepsilon(0) = \varepsilon^{0} \in \R$. Then the state of a molecule is described by $(v,\zeta) \in \R^3 \times \{ 0 \}$, thus simply by $v \in \R^3$. The conservation laws become
\begin{equation*}
    \begin{cases}
   \hspace{68pt}  m v +  m v_* =   m v' +  m v'_* \\
    \displaystyle  \frac{m}{2} |v|^2 + \varepsilon^{0} + \frac{m}{2} |v_*|^2 + \varepsilon^{0} = \frac{m}{2} |v'|^2 + \varepsilon^{0}  + \frac{m}{2} |v'_*|^2 + \varepsilon^{0},
    \end{cases}
\end{equation*}
which simplify into the conservation laws of an elastic collision
\begin{equation*}
    \begin{cases}
   \hspace{20.5pt} v + v_* =  v' + v'_* \\
    \displaystyle |v|^2  + |v_*|^2 = |v'|^2 +  |v'_*|^2.
    \end{cases}
\end{equation*}
From there, the classical post-collision velocities
\begin{equation*}
\begin{cases}
v' &= \displaystyle \frac{v + v_*}{2} + \frac{|v-v_*|}{2} \, T_{\omega}\left[ \frac{v - v_*}{|v-v_*|} \right] \\
v'_* &= \displaystyle \frac{v + v_*}{2} - \frac{|v-v_*|}{2} \,T_{\omega}\left[ \frac{v - v_*}{|v-v_*|} \right]
\end{cases}
\end{equation*}
are recovered. The number of internal degrees of freedom and heat capacity at constant volume associated with this model are, as expected,
$$
\delta = 0 \quad \text{and} \quad c_V = \frac{3}{2},
$$
reproducing that fact that an atom has only the three translational degrees of freedom.

\subsubsection*{The weighted model with one continuous variable}
This model was originally proposed by Borgnakke et al. in \cite{borgnakke1975statistical} and completed with the introduction of a weight $\varphi$ of integration by Desvillettes et al. in \cite{desvillettes1997modele,desvillettes2005kinetic} in order to accurately describe polyatomic gases with the introduction of a single parameter $I$ related to the internal energy of the molecule. The state of a molecule is described by $(v,I) \in \R^3 \times \R_+$, and the associated total energy is $\displaystyle \frac{m}{2} |v|^2 + I$.

\smallskip

\noindent To recover this model in our framework, set $\E = \R_+$, $\A$ the set of Borelians of $\R_+$, $\text{d}\mu(I) = \varphi(I) \text{d}I$, and for $I \in \R_+$, $\varepsilon(I) = I$. The conservation laws of a collision write, when the collision is possible, as

\begin{equation*}
    \begin{cases}
   \hspace{71.5pt}  m v +  m v_* =   m v' +  m v'_* \\
   \displaystyle \frac{m}{2} |v|^2 + I + \frac{m}{2} |v_*|^2 + I_* = \frac{m}{2} |v'|^2 + I'  + \frac{m}{2} |v'_*|^2 + I'_*.
    \end{cases}
\end{equation*}
Since $I$ is a continuous parameter, in \cite{desvillettes2005kinetic}, Desvillettes et al. introduce $\displaystyle e = \frac{m}{4} |v-v_*|^2 + I + I_*$, $\displaystyle r = \frac{I'}{I' + I'_*}$ and $\displaystyle R = \frac{\frac{m}{4} |v'-v'_*|^2}{e}$. Using these parameters, the post-collision velocities write in their formulation
\begin{equation*}
\begin{cases}
v' &= \displaystyle  \frac{v + v_*}{2} + \sqrt{Re} \, T_{\omega}\left[ \frac{v - v_*}{|v-v_*|} \right] \\
v'_* &= \displaystyle \frac{v + v_*}{2} - \sqrt{Re} \,T_{\omega}\left[ \frac{v - v_*}{|v-v_*|} \right]
\end{cases}
\end{equation*}
Our formula for post-collision velocities is equivalent and writes, for allowed collisions,
\begin{equation*}
\begin{cases}
v' &= \displaystyle \frac{v + v_*}{2} + \sqrt{\frac14 |v-v_*|^2 + \frac{1}{m}(I + I_* - I' - I'_*)} \, T_{\omega}\left[ \frac{v - v_*}{|v-v_*|} \right] \\
v'_* &= \displaystyle \frac{v + v_*}{2} - \sqrt{\frac14 |v-v_*|^2 + \frac{1}{m}(I + I_* - I' - I'_*)} \,T_{\omega}\left[ \frac{v - v_*}{|v-v_*|} \right]
\end{cases}
\end{equation*}
The main technical difference between the two approaches is that the authors in \cite{desvillettes2005kinetic} fix all the pre--collision parameters (velocities, internal energies, $r$ and $R$) and study then the transformation
$$
\Tilde{S}_{\omega} : (v,v_*,I,I_*,r,R) \mapsto (v',v'_*,I',I'_*,r',R')
$$
defined on $\R^3 \times \R^3 \times \R_+ \times \R_+ \times [0,1] \times [0,1]$, while, in our framework, we fix pre--collision velocities and all pre-- and post--interaction internal energies and we study the transformation
$$
S_{\omega}[I,I_*,I',I'_*] : (v,v_*) \mapsto (v',v'_*)
$$
defined on $E[I,I_*,I',I'_*]$.

We are able to prove that the two methods are equivalent. Indeed, we denote by $\Phi$ the bijection
$$
\Phi : \begin{cases}
\hspace{63pt} \Tilde{\Omega} \to \Omega \\
(v,v_*,I,I_*,r,R) \mapsto (v,v_*,I,I_*,I',I'_*),
\end{cases}
$$
where $\Tilde{\Omega} = (\R^3)^2 \times (\R_+)^2  \times [0,1]^2 $ and $\Omega = \left\{ (v,v_*,I,I_*,I',I'_*) \in (\R^3)^2 \times (\R_+)^4, \; (v,v_*) \in E[I,I_*,I',I'_*] \right\}$. Its Jacobian can be computed and is equal to
$$
J[\Phi](v,v_*,I,I_*,r,R) = (1 - R) e^2,
$$
where we recall that $\displaystyle e = \frac{m}{4} |v-v_*|^2 + I + I_* = \frac{m}{4} |v'-v'_*|^2 + I' + I'_* = e'$. Then
\begin{align*}
S_{\omega}[I,I_*,I',I'_*] (v,v_*) &= (v',v'_*)\\
\Phi \circ \Tilde{S}_{\omega} \circ \Phi^{-1}(v,v_*,I,I_*,I',I'_*) &= (v',v'_*,I',I'_*,I,I_*),
\end{align*}
and we remark that (see Lemma 1. in \cite{desvillettes2005kinetic} and equation \eqref{eq:jacobian})
\begin{align*}
J\left[ \Phi \circ \Tilde{S}_{\omega} \circ \Phi^{-1} \right](v,v_*,I,I_*,I',I'_*) &= (1-R')(e')^2\, \frac{(1-R) |v'-v'_*|}{(1-R') |v-v_*|} \, \frac{1}{(1-R)e^2} \\
&= \frac{|v'-v'_*|}{|v-v_*|} = J\left[S_{\omega}[I,I_*,I',I'_*] \right](v,v_*).
\end{align*}
The authors in \cite{desvillettes2005kinetic} consider a kernel $B :\Tilde{\Omega} \times \mathbb{S}^2 \to \R_+$ with the micro-reversibility conditions
\begin{align*}
B(v,v_*,I,I_*,R,r,\omega) &= B(v_*,v,I_*,I,R,1-r,\omega)\\
B(v,v_*,I,I_*,R,r,\omega) &= B(v',v'_*,I',I'_*,R',r',\omega).
\end{align*}
We define on $\Omega \times \mathbb{S}^2$
$$
b(v,v_*,I,I_*,I',I'_*,\omega) = \frac{1}{e^2} |v-v_*|^{-1} \, \frac{B(\cdot,\omega) \circ \Phi^{-1}(v,v_*,I,I_*,I',I'_*)}{\varphi(I) \varphi(I_*) \varphi(I') \varphi(I'_*)},
$$
and set $b$ to zero elsewhere. Then, since $e=e'$,
\begin{align*}
b(v,v_*,I,I_*,I',I'_*,\omega) &= b(v_*,v,I_*,I,I'_*,I',\omega)\\
|v-v_*| \, b(v,v_*,I,I_*,I',I'_*,\omega) &= |v'-v'_*| \, b(v',v'_*,I',I'_*,I,I_*,\omega),
\end{align*}
which are the symmetry and micro-reversibility conditions in our framework. It follows that the difference of approach is purely technical, and does not lead to theoretical discrepancies. The same distribution $f$ and macroscopic quantities are recovered. Indeed, in this case, with $\zeta = I$,
\begin{align*}
    &\B(f,f)(v,\zeta) = \iiint_{\E^3}  \int_{\R^3}  \int_{\mathbb{S}^2} \, \Big(f(v',\zeta')f(v'_*,\zeta'_*) - f(v,\zeta)f(v_*,\zeta_*) \Big) \; b(\cdot)  \; \text{d}\omega  \, \text{d}v_* \,  \text{d}\mu^{\otimes 3}(\zeta_*,\zeta',\zeta'_*) \\
    &= \iiint_{(\R_+)^3} \int_{\R^3} \int_{\mathbb{S}^2}  \, \Big(f(v',I')f(v'_*,I'_*) - f(v,I)f(v_*,I_*) \Big) \; b(\cdot)  \;\text{d}\omega  \,  \text{d}v_* \,  \varphi(I_*) \varphi(I') \varphi(I'_*) \, \text{d}I_* \, \text{d}I' \, \text{d}I'_*\\
    &= \int_{\mathbb{S}^2}  \int_{\Omega[v,I]} \, \Big(f(v',I')f(v'_*,I'_*) - f(v,I)f(v_*,I_*) \Big) \; \frac{1}{e^2}|v-v_*|^{-1} \, B(\cdot,\omega) \circ \Phi^{-1}(\cdot) \, \frac{1}{\varphi(I)}  \;   \text{d}v_* \, \text{d}I_* \, \text{d}I' \, \text{d}I'_* \, \text{d}\omega\\
    &=  \int_{\mathbb{S}^2} \int_{\tilde{\Omega}[v,I]}  \, \Big(f(v',I')f(v'_*,I'_*) - f(v,I)f(v_*,I_*) \Big) \; |v-v_*|^{-1}  \, B(\cdot,\omega) \frac{J[\Phi](\cdot)}{e^2} \, \frac{1}{\varphi(I)}  \;  \text{d}v_*  \, \text{d}I_* \, \text{d}R \, \text{d}r \, \text{d}\omega  \, \\
    &= \iint_{[0,1]^2} \int_{\R_+} \int_{\mathbb{S}^2}  \, \Big(f(v',I')f(v'_*,I'_*) - f(v,I)f(v_*,I_*) \Big) \; |v-v_*|^{-1} \, (1-R) \, B(\cdot,\omega)  \, \frac{1}{\varphi(I)}  \; \text{d}\omega  \, \text{d}v_*  \, \text{d}I_* \, \text{d}R \, \text{d}r\,,
\end{align*}
where we recall that $J[\Phi](v,v_*,I,I_*,r,R) = (1-R) e^2$ is the Jacobian of $\Phi$. We denoted by $\Omega[v,I]$ the set of $(v_*,I_*,I',I'_*)$ such that $(v,v_*,I,I_*,I',I'_*) \in \Omega$ and the same for $\Tilde{\Omega}$. We recover the same macroscopic quantities as well, for example the density
$$
\rho(t,x) = \int_{\E} \int_{\R^3}  m f(t,x,v,\zeta) \, \text{d}v \, \text{d}\mu(\zeta) = \int_{\R_+} \int_{\R^3}  m f(t,x,v,I) \, \text{d}v \, \varphi(I) \, \text{d}I.
$$
\begin{remark}
The approach in \cite{borgnakke1975statistical,desvillettes2005kinetic} consists in distributing the energy of the incoming molecules simultaneously to the post-collision relative kinetic energy and to the internal states, whereas, in our framework, we fix the pre and post-collision internal states and consider afterwards the set of pre-collision velocity pairs for which this collision is possible.
\end{remark}

\noindent This model is generally used with $\varphi(I) = I^{\alpha}$, with $\alpha > -1$. From Proposition \ref{proposition:alphabetadelta}, we recover that in this case ($\beta=1$ in the proposition), the number of internal degrees of freedom and heat capacity at constant volume are
$$
\delta = 2 (\alpha + 1) \quad \text{and} \quad c_V = \frac{5}{2} + \alpha.
$$
\subsubsection*{The model with discrete energy levels}
In order to accurately describe the vibration inside molecules, a model with a discrete-energy levels description is proposed by Groppi and Spiga in \cite{groppi1999kinetic}. The authors consider a finite set of internal energy levels $(\varepsilon_n)_{n \in \llbracket 0,N \rrbracket} \in \R^{N+1}$. The state of a molecule is then $(v,n) \in \R^3 \times \llbracket 0,N \rrbracket$, and the associated total energy is $\displaystyle \frac{m}{2} |v|^2 + \varepsilon_n$. The authors define $\delta_{ij}^{kl} = \frac{4}{m}(\varepsilon_k + \varepsilon_l - \varepsilon_i - \varepsilon_j)$, $g = |v-v_*|$ and $g' = |v'-v'_*|$. A collision is allowed when $g^2 - \delta_{ij}^{kl} \geq 0$, with then $g' = \sqrt{g^2 - \delta_{ij}^{kl}}$.

\smallskip

\noindent To recover this model in our framework, we set $\E = \llbracket 0,N \rrbracket$, $\A = \mathcal{P}(\llbracket 0,N \rrbracket)$, $\mu$ the counting measure on $\llbracket 0,N \rrbracket$ and for $n \in \llbracket 0,N \rrbracket$, $\varepsilon(n) = \varepsilon_n$. Indeed, note that
$$
g^2 - \delta_{ij}^{kl} \geq 0 \iff \frac14 |v-v_*|^2 + \frac{1}{m} (\varepsilon_i +  \varepsilon_j -  \varepsilon_k -  \varepsilon_l) \geq 0 \iff \Delta (v,v_*,i,j,k,l) \geq 0,
$$
where $\Delta$ is defined in equation \eqref{eq:def:delta}. Moreover, the authors give (formula 2.6 in \cite{groppi1999kinetic})
$$
\text{d}v' \, \text{d}v'_* \, \text{d}\Omega' = \mathbf{1}_{g^2 \geq \delta_{ij}^{kl}} \, \frac{g'}{g} \, \text{d}v \, \text{d}v_* \, \text{d}\Omega,
$$
which corresponds in our framework to the Jacobian of the transformation $(v,v_*) \mapsto (v',v'_*)$, see formula \eqref{eq:jacobian}. The authors in \cite{groppi1999kinetic} give the following post-collision velocities
\begin{align*}
    v' &= \frac{v+v_*}{2} + \frac12 \sqrt{g^2 - \delta_{ij}^{kl}} \, \Omega \\
    v'_* &=  \frac{v+v_*}{2} - \frac12 \sqrt{g^2 - \delta_{ij}^{kl}} \, \Omega,
\end{align*}
where $\Omega \in \mathbb{S}^2$. Now remark that $\frac12 \sqrt{g^2 - \delta_{ij}^{kl}} = \sqrt{\Delta (v,v_*,i,j,k,l)}$ and there exists (assuming $v \neq v_*$) $\omega \in \mathbb{S}^2$ such that $\Omega = T_{\omega}\left[ \frac{v-v_*}{|v-v_*|} \right]$. We then have, like in our framework,
\begin{align*}
    v' &= \frac{v+v_*}{2} + \sqrt{\Delta(v,v_*,i,j,k,l)} \,\, T_{\omega}\left[ \frac{v-v_*}{|v-v_*|} \right] \\
    v'_* &=  \frac{v+v_*}{2} - \sqrt{\Delta(v,v_*,i,j,k,l)}\, \, T_{\omega}\left[ \frac{v-v_*}{|v-v_*|} \right].
\end{align*}
The authors in \cite{groppi1999kinetic} define the cross-section $\sigma_{ij}^{kl}$, with the symmetry and micro-reversibility conditions
\begin{align*}
\sigma_{ij}^{kl}(v,v_*,\omega) &= \sigma_{ji}^{lk}(v_*,v,\omega) \\
|v-v_*|^2 \sigma_{ij}^{kl}(v,v_*,\omega) &= |v'-v'_*|^2 \sigma^{kl}_{ij}(v',v'_*,\omega).
\end{align*}
By defining the collision kernel
$$
b(v,v_*,i,j,k,l,\omega) =  \mathbf{1}_{g^2 \geq \delta_{ij}^{kl}} \, |v-v_*| \sigma_{ij}^{kl}(v,v_*,\omega),
$$
$b$ verifies the positivity assumption (if $\sigma_{ij}^{kl}$ is properly defined) and from the properties assumed for $\sigma_{ij}^{kl}$ we recover
\begin{align*}
b(v,v_*,i,j,k,l,\omega) &= b(v_*,v,j,i,l,k,\omega)\\
|v-v_*| b(v,v_*,i,j,k,l,\omega) &= |v'-v'_*| b(v',v'_*,k,l,i,j,\omega).
\end{align*}
Finally, with $(\zeta,\zeta_*,\zeta',\zeta'_*) = (i,j,k,l)$,
\begin{align*}
\B(f,f)(v,\zeta) &= \iiint_{\E^3}  \int_{\R^3} \int_{\mathbb{S}^2}  \, \Big(f(v',\zeta')f(v'_*,\zeta'_*) - f(v,\zeta)f(v_*,\zeta_*) \Big) \; b(\cdot)  \; \text{d}\omega  \, \text{d}v_* \, \text{d}\mu^{\otimes 3}(\zeta_*,\zeta',\zeta'_*) \\
&= \sum_{j,k,l \in \llbracket 0,N \rrbracket}  \int_{\R^3} \int_{\mathbb{S}^2}  \, \Big(f(v',k)f(v'_*,l) - f(v,i)f(v_*,j) \Big) \; b(v,v_*,i,j,k,l,\omega)  \; \text{d}\omega  \, \text{d}v_*  \\
&= \sum_{j,k,l \in \llbracket 0,N \rrbracket}  \int_{\R^3} \int_{\mathbb{S}^2} \mathbf{1}_{g^2 \geq \delta_{ij}^{kl}} \, |v-v_*| \, \sigma_{ij}^{kl}(v,v_*,\omega) \, \Big(f_k(v')f_l(v'_*) - f_i(v)f_j(v_*) \Big)  \; \text{d}\omega  \, \text{d}v_*,
\end{align*}
with the notation $f_i(v) = f(v,i)$. The same macroscopic quantities are recovered as well, for instance the density
$$
\rho(t,x) = \int_{\E} \int_{\R^3}  m f(t,x,v,\zeta) \, \text{d}v \, \text{d}\mu(\zeta) = \sum_{i = 0}^N \int_{\R^3}  m  f_i(t,x,v) \, \text{d}v.
$$

\subsection{Combination of the continuous and discrete models}

The continuous and discrete models can be combined in our framework. Let us consider a finite set of internal energy levels $(\epsilon_n)_{n \in \llbracket 0,N \rrbracket} \in \R^{N+1}$ and a weight function $\varphi$ on $\R_+$. We can build the following model
\begin{equation}
(\E,\A,\text{d}\mu(\zeta)) = \left( \R_+ \times \llbracket 0, N \rrbracket, \, \text{Bor}(\R_+) \otimes \mathcal{P}(\llbracket 0, N \rrbracket), \, \varphi(I) \, \text{d}I \times 1 \right), \quad \quad \varepsilon(I,n) = I + \epsilon_n,
\end{equation}
where $1$ stands for the counting measure. The Boltzmann operator writes
\begin{align*}
&\B(f,f)(v,I,i) =  \\&\sum_{j,k,l \in \llbracket 0,N \rrbracket} \iiint_{(\R_+)^3} \int_{\R^3} \int_{\mathbb{S}^2}  \, \Big(f_k(v',I')f_l(v'_*,I'_*) - f_i(v,I)f_j(v_*,I_*) \Big) \; b(\cdot)  \; \text{d}\omega  \, \text{d}v_* \, \varphi(I_*)\, \text{d} I_* \, \varphi(I')\, \text{d} I' \, \varphi(I'_*)\, \text{d} I'_*\,.
\end{align*}

\noindent Let us set $\epsilon^0 = \underset{0 \leq n \leq N}{\text{min}} \{\epsilon_n \}$. From Theorem \ref{theorem:equipartition}, the Maxwellian writes
$$
\mathcal{M}[\rho,u,T](v,I,i) = \rho \, m^{1/2}  \, (2 \pi k_B T)^{-3/2} \, Z_c \left( \frac{1}{k_B T} \right)^{-1} \, Z_d \left( \frac{1}{k_B T} \right)^{-1} \, \text{exp} \left( - \frac{ m|v - u|^2}{2 k_B T} - \frac{I}{k_B T} - \frac{\epsilon_i - \epsilon^0}{k_B T} \right),
$$
where $Z_c$ and $Z_d$ are the partition functions respectively associated to the continuous and discrete parts, for all $\beta > 0$
$$
Z_c (\beta) = \int_{\R_+} e^{-\beta  I} \, \varphi(I) \, \text{d}I, \quad \quad Z_d(\beta) = \sum_{n = 0}^N e^{-\beta  (\epsilon_n - \epsilon^0)}.
$$
Finally, again from Theorem \ref{theorem:equipartition}, the number of internal degrees of freedom writes
$$
\delta(T) = \delta_c(T) + \delta_d (T) = \frac{2}{k_B T} \, \frac{ \int_{\R_+} I \, \text{exp} \left(-\frac{I}{k_B T}\right) \, \varphi(I) \, \text{d}I}{ \int_{\R_+} \text{exp} \left(-\frac{I}{k_B T}\right) \, \varphi(I) \, \text{d}I} + \frac{2}{k_B T} \, \frac{ \sum_{n = 0}^N  (\epsilon_n - \epsilon^0) \, \text{exp} \left(-\frac{ (\epsilon_n - \epsilon^0)}{k_B T}\right)}{ \sum_{n = 0}^N  \text{exp} \left(-\frac{ (\epsilon_n - \epsilon^0)}{k_B T}\right)}.
$$

\noindent
As seen previously, if $\varphi(I) = I^{\alpha}$ with $\alpha > -1$, then $\delta_c = 2 (\alpha + 1)$. This model with internal energy described by two different variables, a continuous and a discrete one, may be used to separate the vibrational energy of a molecule from the rotational part. As suggested in \cite{Herzberg}, the rotational energy may be modelled by means of a continuous variable, while for the vibrational energy a discrete approximation is more suitable. Note that the number of internal degrees of freedom turns out to be, as expected, the sum of the vibrational and the rotational ones.

\subsection{Building a model within the framework} \label{subsection:buildmodel}

In this subsection, we explain how to build a model for a polyatomic gas within our framework, with an example. Building an accurate model corresponds to giving an accurate description of the internal states of a given molecule. On this subject, the main field to rely on is Molecular Mechanics/Quantum Chemistry, see for example \cite{hehre2003guide}. There are various phenomena to be taken into account, typically rotation, vibration, electronic excitation and nuclei spin, that could also be correlated. For the sake of simplicity of the mathematical model, it could be enough to take only rotation and vibration into account, and assume them to be uncorrelated, but there are physical problems where this kind of description is too simplistic. The most appropriate model in a given situation typically depends on the considered regime and the desired degree of complexity and accuracy. This is the reason why we wish not to propose one model, but to give examples and insights on how to build realistic models inside our framework.

\subsubsection*{Quantum description of the internal structure}

To our knowledge, the currently best internal description of a molecule is given by quantum mechanics. Suppose to have a Hamiltonian operator $\hat{H}$. Let $\{\lambda_q\}_{q\in Q}$ be the set of its (real) eigenvalues, where $Q$ is a discrete set (e.g. $\N, \N^2, \dots$). We denote by $r_q$ the dimension of the eigenvector space associated to the eigenvalue $\lambda_q$. The values $\lambda_q$ are the possible measurable energy levels associated to the operator $\hat{H}$, and $r_q$ is the so--called degeneracy of the $q^{\text{th}}$ energy level. We then consider a subset $\Tilde{Q} \subset Q$, the set of energy levels of the \emph{bound states} (the physically admissible energy levels), and build the following model
$$
(\E,\A,\mu) = (\Tilde{Q},\mathcal{P}(\Tilde{Q}),(r_q)_q), \qquad \varepsilon(q) = \lambda_q\,,
$$
where the notation $\mu = (r_q)_q$ means that $\mu(\{ q\}) = r_q$. As we explained in subsection \ref{subsection:degeneracy}, degeneracy is indeed included in the model via the measure $\mu$.

\noindent
We remark that if the set of bound states $\Tilde{Q} = Q$ and \text{inf}$\{\lambda_q\}_{q\in Q} = 0$, then we recover the formula of statistical mechanics for the partition function (see \cite{huang1963statistical}, Chapter 9)
$$
Z (\beta) = \sum_{q \in Q} e^{-\beta \lambda_q} \, r_q = \text{Tr}(e^{- \beta \hat{H}}),
$$
where $\text{Tr}(e^{- \beta \hat{H}})$ is the trace of the operator $e^{- \beta \hat{H}}$.

\begin{remark}
 The framework we present in this paper, with the Boltzmann equation \ref{eq:boltzmann} is "at most" semi-classical, in the sense that the velocity part is described by classical mechanics. A full quantum mechanical Boltzmann equation was derived by Waldmann \cite{Wa57} and Snider~\cite{Sn60}. In our framework, we assume that the gas is isotropic (no polarization effects), and other quantum considerations would only appear in the internal structure $(\E,\A,\mu), \, \varepsilon$ of the molecule.
\end{remark}

\smallskip

\subsubsection*{Semi-classical description of the internal structure}
A possible approach is to consider only rotation and vibration, describing independently rotation with classical mechanics and vibration with quantum mechanics. The description of vibration is thus exactly the model of the previous paragraph, with a Hamiltonian describing only vibration,
$$
(\E_{vib},\A_{vib},\mu_{vib}) = (\Tilde{Q},\mathcal{P}(\Tilde{Q}),(r_q)_q), \qquad \varepsilon_{vib}(q) = \lambda_q.
$$
We assume the molecule to be a rigid rotor, that is, in the description of rotation we assume that no deformation is induced in the molecule. The rotation-related internal state we consider is thus the angular velocity of the molecule in a coordinate system attached to the molecule. In general, this angular velocity lives in $\R^3$, however when the molecule is linear, being symmetric by rotation around its own axis, the contribution of the angular momentum's coordinate along this axis can be assumed to be 0 (this approximation can be found for example in \cite{van1951coupling}, section 7). Thus, in the case of a linear molecule, the angular momentum, in a coordinate system attached to the molecule with one axis being the molecule's axis, would live in $\R^2$. Moreover, in the linear case, the moment of inertia is the same in both directions. This way of modelling leads to the following framework for the description of rotation
$$
(\E_{rot}, \A_{rot}, \text{d}\mu_{rot}(z)) = \left( \R^{d}, \text{Bor}(\R^{d}), \text{d}z \right), \qquad \varepsilon_{rot}(z) = \frac12 \sum_{i=1}^{d} \mathcal{I}_i \, z_i^2,
$$
where $\mathcal{I}_i$ is the moment of inertia along the $i$-th axis, $z$ is the angular velocity in a coordinate system attached to the molecule, $d = 2$ and $\mathcal{I}_1 = \mathcal{I}_2 = \mathcal{I}$ if the molecule is linear, $d=3$ if the molecule is non-linear. The model taking rotation and vibration into account then writes
$$
(\E, \A, \text{d} \mu(\zeta)) = \big(\R^{d} \times \Tilde{Q}, \, \text{Bor}(\R^{d}) \otimes \mathcal{P}(\Tilde{Q}), \, \text{d}z \, r_q \big), \quad \quad \varepsilon(z,q) = \frac12 \sum_{i=1}^{d} \mathcal{I}_i \, z_i^2 + \lambda_q,
$$
From Theorem \ref{theorem:equipartition}, for all $T > 0$,
$$
\text{d} \nu_T(z,q) = \text{d} \nu^{rot}_T(z) \, \text{d} \nu^{vib}_T (q) = \prod_{i=1}^d  \frac{\text{exp} \left(- \frac{\mathcal{I}_i \, z_i^2}{ 2k_B T} \right)}{\left(\frac{2 \pi k_B T}{ \mathcal{I}_i} \right)^{1/2} } \, \text{d}z\,\, \frac{\text{exp} \left(- \frac{ (\lambda_q - \lambda^0)}{k_B T} \right)}{\sum_{q' \in \Tilde{Q}}  \text{exp} \left(-\frac{ (\lambda_{q'} - \lambda^0)}{k_B T}\right) \, r_{q'} } \, r_q.
$$
Moreover, setting $\displaystyle \lambda^0 = \underset{q \in \Tilde{Q}}{\text{inf}}  \{ \lambda_q \}$ ($\in \R$ by assumption) and using Theorem \ref{theorem:equipartition} and Proposition \ref{proposition:alphabetadelta}, the number of internal degrees of freedom turns out to be
$$
\delta(T) = \delta_{rot}(T) + \delta_{vib}(T) = d + 2 \, \frac{ \sum_{q \in \Tilde{Q}} \frac{(\lambda_q - \lambda^0)}{k_B T}  \, \text{exp} \left(-\frac{ (\lambda_q - \lambda^0)}{k_B T}\right) \, r_q}{ \sum_{q \in \Tilde{Q}}  \text{exp} \left(-\frac{ (\lambda_q - \lambda^0)}{k_B T}\right) \, r_q}.
$$

%
%
%

\subsubsection*{Example: the  $\prescript{1}{}{\emph{H}}\prescript{19}{}{\emph{F}}$ gas}

As an example, we propose four models within our framework for the Hydrogen Fluoride ($\prescript{1}{}{\text{H}}\prescript{19}{}{\text{F}}$) gas. We focus only on rotation and vibration. We propose two semi-classical models and two quantum models. For the formulas of eigenvalues presented hereafter, we refer to the Chapter {\it Spectroscopy Constants of Diatomic Molecules} in \cite{crchandbook}.

\bigskip

\noindent \textsf{1. Harmonic semi-classical model.} We describe the internal states of $\prescript{1}{}{\text{H}}\prescript{19}{}{\text{F}}$ with a semi-classical approach. Since this molecule is diatomic, it is linear, so that $d = 2$ and $\mathcal{I}_1 = \mathcal{I}_2$ in the description of rotation, and has only one mode of vibration. In this simplified model, the Hamiltonian is assumed to include the harmonic potential (model of the quantum harmonic oscillator). This leads to the model
$$
(\E,\A,\text{d}\mu(\zeta)) = \left(\R^2 \times \N,\, \text{Bor}(\R^2) \otimes \mathcal{P}(\N),\, \text{d}z \times 1 \right), \qquad \varepsilon(z,n) =  \underset{\text{rotation}}{\underbrace{\frac12 \mathcal{I} \, |z|^2}} + \underset{\text{vibration}}{\underbrace{hc \, \nu_e \left( n + \frac12 \right)}},
$$
where $\mathcal{I}$ is the moment of inertia associated with $\prescript{1}{}{\text{H}}\prescript{19}{}{\text{F}}$, $h$ is the Planck constant (expressed in SI units), $c$ the speed of light (expressed in $\text{cm}.s^{-1}$) and $\nu_e$ the wavenumber (expressed in $\text{cm}^{-1}$). Defining $T_{vib} = h c \, \nu_e /k_B$, the number of internal degrees of freedom writes
$$
\delta_1(T) = \delta^{rot}_1(T) + \delta^{vib}_1(T) = 2 + 2 \, \frac{T_{vib} / T}{\text{exp}(T_{vib} / T) - 1}\,.
$$

\bigskip

\noindent \textsf{2. Anharmonic semi-classical model.} Instead of using the harmonic potential, the Morse potential is a useful approximation of the actual internuclear potential, since it allows anharmonicity and an explicit computation of the eigenvalues of the Hamiltonian. The family of eigenvalues is, up to a constant (see Morse \cite{morse1929diatomic}),
$$
\left( hc \, \nu_e \left( n + \frac12 \right) - hc \, \nu_e x_e \left( n + \frac12 \right)^2 \right)_{n \in \N}
$$
where $x_e$ represents anharmonicity. Not all eigenvalues correspond to bound states. We thus restrict the family to the set $\llbracket 0, N_{max} \rrbracket$, where $\displaystyle N_{max} = \left\lfloor \frac{1}{2 x_e} \right\rfloor - 1$. With this potential, the model becomes
\begin{align*}
(\E,\A,\text{d}\mu(\zeta)) &= \left(\R^2 \times \llbracket 0, N_{max} \rrbracket,\, \text{Bor}(\R^2) \otimes \mathcal{P}(\llbracket 0, N_{max} \rrbracket),\, \text{d}z \times 1 \right), \\
\varepsilon(z,n) &=  \frac12 \mathcal{I} \, |z|^2 +  hc \, \nu_e \left( n + \frac12 \right) -  \underset{\text{anharmonicity}}{\underbrace{hc \, \nu_e x_e \left( n + \frac12 \right)^2}}.
\end{align*}
The number of internal degrees of freedom writes

$$
\delta_2(T) = \delta^{rot}_2(T) + \delta^{vib}_2(T) = 2 + \delta^{vib}_2(T).
$$

\bigskip

\noindent \textsf{3. Simplified quantum model.} Here we use a quantum description for both rotation and vibration. To simplify, we describe independently rotation and vibration, using the rigid-rotor assumption for rotation and the harmonic potential for vibration. Since $\prescript{1}{}{\text{H}}\prescript{19}{}{\text{F}}$ is a diatomic molecule, there is only one mode of vibration. This leads to the model
$$
(\E,\A,\mu) = \left(\N \times \N,\, \mathcal{P}(\N) \otimes \mathcal{P}(\N),\, (2J + 1)_{(J,n)\in \N^2} \right), \qquad \varepsilon(J,n) =   \underset{\text{rotation}}{\underbrace{B J (J+1)}} + \underset{\text{vibration}}{\underbrace{hc \, \nu_e \left( n + \frac12 \right)}},
$$
where $B$ is the rotational constant associated with  $\prescript{1}{}{\text{H}}\prescript{19}{}{\text{F}}$. Defining $T_{vib} = h c \, \nu_e /k_B$, the number of internal degrees of freedom writes
$$
\delta_3(T) = \delta^{rot}_3(T) + \delta^{vib}_3(T) = \delta^{rot}_3(T) + 2 \, \frac{T_{vib} / T}{\text{exp}(T_{vib} / T) - 1}.
$$

\bigskip

\noindent \textsf{4. Improved quantum model.} For a better description, we can also take into account the correlation of rotation and vibration, considering non-rigid rotation, and using the Morse potential for vibration to allow for anharmonicity. This would lead to the model
\begin{align*}
(\E,\A,\mu) &= \left(\Tilde{Q},\, \mathcal{P}(\Tilde{Q}),\, (2J + 1)_{(J,n)\in \Tilde{Q}} \right), \\
\varepsilon(J,n) &= \underset{\text{rot-vib coupling}}{\underbrace{\left(B - \alpha \left( n + \frac12 \right) \right)}} J (J+1) - \underset{\text{centrifugal distortion}}{\underbrace{D[J (J+1)]^2}} + hc \, \nu_e \left( n + \frac12 \right) - \underset{\text{anharmonicity}}{\underbrace{hc \, \nu_e x_e \left( n + \frac12 \right)^2}},
\end{align*}
where $\Tilde{Q}$ is the set of bound states, defined by
$$
\Tilde{Q} = \left\{  (J,n) \in \N^2, \; \varepsilon(J,n) \geq \varepsilon(J-1,n) \; \text{ and } \; \varepsilon(J,n-1) \geq \varepsilon(J,n) \right\} \subset \llbracket 0,J_{max} \rrbracket \times \llbracket 0,N_{max} \rrbracket.
$$
We denote by $\delta_4$ the number of internal degrees of freedom in this improved case, which cannot be written as the sum of rotational and vibrational parts due to the coupling.

\medskip

\noindent
From \cite{crchandbook} we know that reasonable data for $\prescript{1}{}{\text{H}}\prescript{19}{}{\text{F}}$ are $\nu_e = 4138.39 \, \text{cm}^{-1}$, $\nu_e x_e = 89.94 \, \text{cm}^{-1}$, $B / h c = 20.95 \, \text{cm}^{-1}$, $\alpha / hc = 0.793 \, \text{cm}^{-1}$ and $D / hc = 0.00215 \, \text{cm}^{-1}$. The value of $\mathcal{I}$ does not matter in the computation of the number of internal degrees of freedom. Just for the sake of comparison of our proposed models, we plot on Fig. \ref{fig:deltas} the various numbers of internal degrees of freedom $\delta_1$, $\delta_2$, $\delta_3$ and $\delta_4$ corresponding to these data, as functions of the temperature (in log-scale), expressed in Kelvin (K). The temperature ranges from 10 to 10 000 K. We see that in this example vibration is negligible for $T \lesssim 1000K$, and becomes important around $2000\sim3000K$. The direct computation of $\delta$ from the choice of the model can be useful to quickly check the validity of an approximation in a given regime.

\begin{figure}[ht]
    \centering
    \includegraphics[width=.7\linewidth]{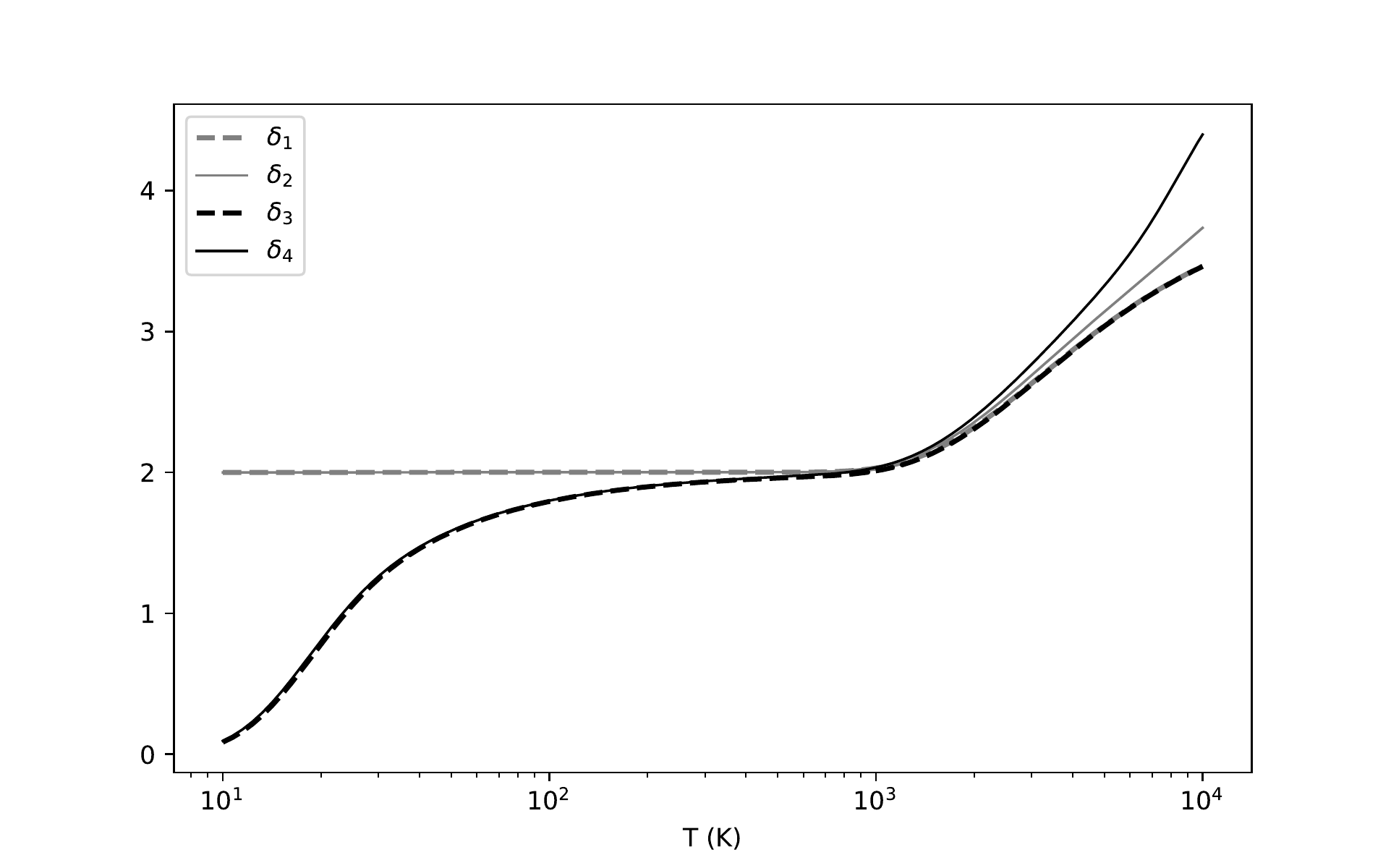}
    \caption{Number of internal degrees of freedom in the various models as functions of temperature, with log-scale $x$-axis.}
    \label{fig:deltas}
\end{figure}

\noindent
We plot on Fig. \ref{fig:deltashighT} the same numbers of internal degrees of freedom as on Fig. \ref{fig:deltas}, for temperature ranging from $10^3$ to $10^6 K$. The decrease of $\delta_2$ and $\delta_4$ after $T \sim 10^4K$ corresponds to the result of Corollary \ref{corollary:boundedenergy}. While the number of internal degrees of freedom is expected to increase with temperature, this plot is an illustration of possible limits of validity of the framework. Indeed, high-temperatures considerations such as bond-breaking are not taken into account in our setting.

\begin{figure}[ht]
    \centering
    \includegraphics[width=.7\linewidth]{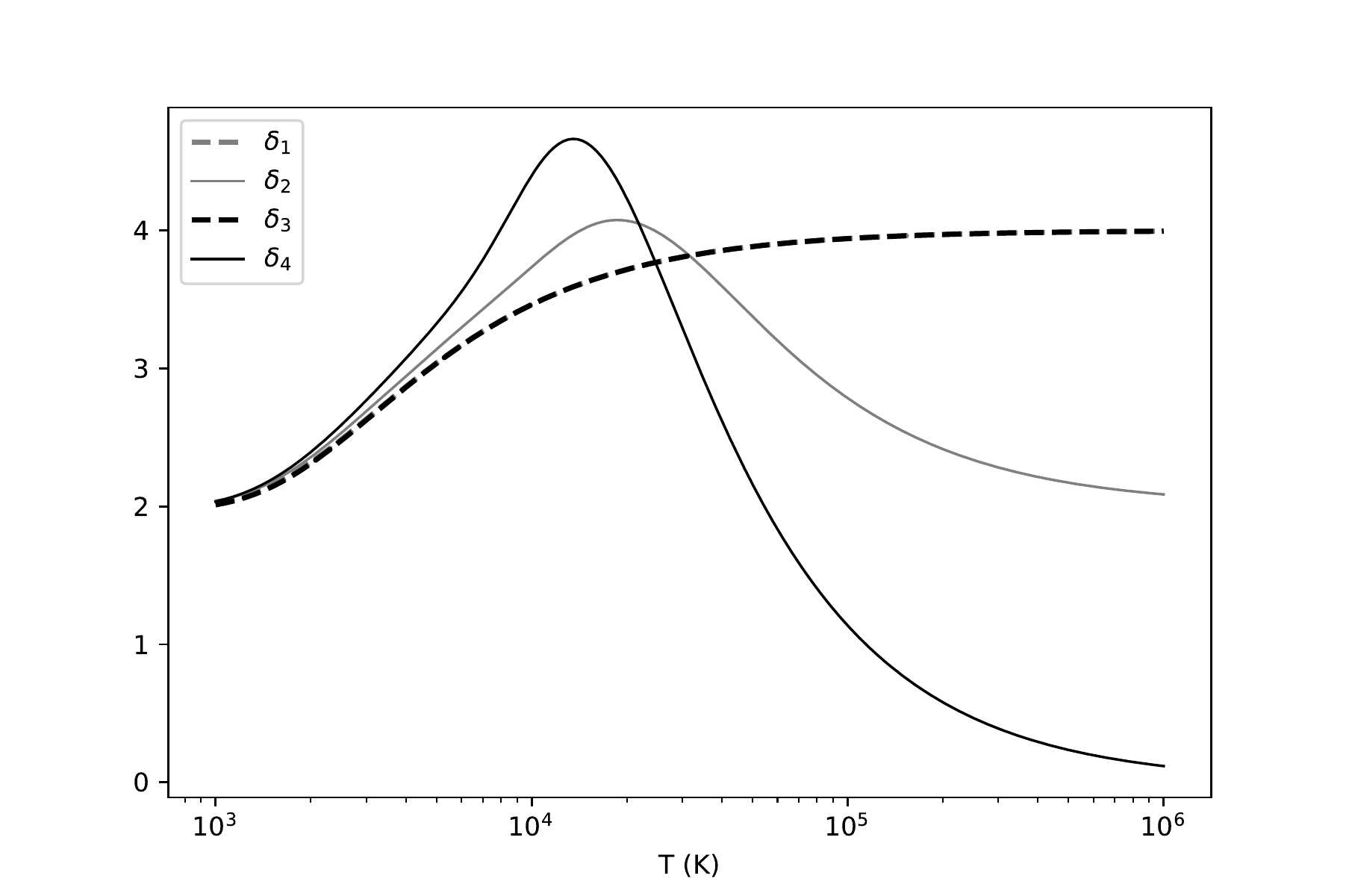}
    \caption{Number of internal degrees of freedom in the various models at high temperatures.}
    \label{fig:deltashighT}
\end{figure}

\section{Reduction to one real variable and comparison with Borgnakke-Larsen} \label{section8}
In the previous section, we saw that our general framework allows to build a whole range of models that can be of great complexity. This possibility can be a strength in regards of the precision it allows in the description of the internal structure. However, too much complexity can be a weakness in regards of numerical aspects. Notably, the strength of the model proposed by Borgnakke and Larsen \cite{borgnakke1975statistical} is that the internal structure is described by a single real parameter $I \in \R_+$, which is highly desirable for numerical simulations. In this section, we show that any model fitting our framework can, under a condition on the collision kernel, be reduced to a one-real-parameter model. This means that all the complexity can be concentrated in a suitable measure on $\R_+$. If this measure has a density, then it can be written as $\varphi(I) \, \text{d}I$ and we recover the model with one continuous variable proposed by Desvillettes et al. \cite{desvillettes1997modele,desvillettes2005kinetic}, which is suited for numerical applications. While this weight $\varphi$ in the original approach in \cite{desvillettes1997modele,desvillettes2005kinetic}  is a parameter chosen a posteriori in order to recover desired macroscopic properties (e.g. the number of internal degrees of freedom), in our case it is computed from the physical structure of the molecule, through the reduction process that we present in this section.

\begin{definition}
We define $\mu^{\Bar{\varepsilon}}$ as the image measure of $\mu$ by $\Bar{\varepsilon}$ on $\R_+$, for all $B \in \text{Bor}(\R_+)$,
$$
\mu^{\Bar{\varepsilon}} \left( B \right) := \mu \left( \Bar{\varepsilon}^{\, -1} (B) \right).
$$
\end{definition}

\begin{theorem} \label{theorem:reduction}
If the collision kernel $b$ depends on the internal states only through their associated energy (see Remark \ref{remark:bdependsonenergy} below), then the general model
$$
\big( \E,\A,\mu \big), \qquad \varepsilon
$$
can be reduced to the model with one non-negative variable
$$
\big(\R_+,\emph{Bor}(\R_+),\mu^{\Bar{\varepsilon}} \big), \qquad I \mapsto I + \varepsilon^0.
$$

\end{theorem}

\begin{proof}
First, remark that the set $E[\zeta,\zeta_*,\zeta',\zeta'_*]$ and the transformation $S_{\omega}[\zeta,\zeta_*,\zeta',\zeta'_*]$ depend on the internal states $\zeta,\zeta_*,\zeta',\zeta'_*$ only through their energy, cf. equations \eqref{eq:setofcoll} and \eqref{eq:transformoutvelocities}. It follows that if $b$ and $f$ depend on the internal states only through their energy, the same holds for $\B(f,f)$. Thus, in the deep structure of our framework, all dependence on $\zeta$ is in fact a dependence on $\varepsilon(\zeta)$ when it is the case for $b$. From the change of variable formula, for any measurable function $\phi : \R \to \R$ such that one of the following integrals makes sense,
$$
\int_{\E} \phi \big(\varepsilon  (\zeta) \big) \, \text{d}\mu(\zeta) = \int_{\E} \big(\phi  (\cdot + \varepsilon^0) \circ \Bar{\varepsilon} \big) (\zeta) \, \text{d}\mu(\zeta) = \int_{\R_+} \phi (I + \varepsilon^0) \, \text{d}\mu^{\Bar{\varepsilon}}(I).
$$
Applying this change of variable to all integrals in this paper ends the proof.

\QED
\end{proof}

\begin{remark} \label{remark:bdependsonenergy}
We say that $b$ depends on the internal states only through their energy if
$$
b(v,v_*,\zeta,\zeta_*,\zeta',\zeta'_*,\omega) \equiv \tilde{b}(v,v_*,\varepsilon(\zeta),\varepsilon(\zeta_*),\varepsilon(\zeta'),\varepsilon(\zeta'_*),\omega).
$$
This trivially holds true for all existing kernels constructed within the existing models presented in Subsection \ref{subsection:existingmodels}, since for the discrete model there is a bijection between the set of indexes and the set of energy levels, and for the continuous model with weight, the internal state \emph{is} the energy: in this case $\zeta = I = \varepsilon(\zeta)$. To the knowledge of the authors, all existing kernel for polyatomic gases considered up to now have this property.
\end{remark}

\medskip

\noindent Theorem \ref{theorem:reduction} establishes a link between two points of view. The first, the one of the present paper, is \emph{state-based}, considering a space of internal (physical) states $\E$. The second, the one of the reduced model, which is an extension of the model with continuous energy $I$ and weight $\varphi$ proposed by Desvillettes et al. \cite{desvillettes1997modele,desvillettes2005kinetic} (when $\mu^{\Bar{\varepsilon}}$ is absolutely continuous with respect to the Lebesgue measure, there exists $\varphi \in L^1(\R_+)$ such that $\text{d} \mu^{\Bar{\varepsilon}}(I) = \varphi(I) \, \text{d} I$), is \emph{energy-based}, since the variable of interest is directly the internal energy of the molecule. Theorem \ref{theorem:reduction} shows that to any states-based model corresponds an energy-based one (that is, a model with one continuous variable $I$ and measure $\mu^{\Bar{\varepsilon}}$/ weight $\varphi$). Since the measure $\mu^{\Bar{\varepsilon}}$/ weight $\varphi$ corresponds to the image measure of $\mu$ by $\Bar{\varepsilon}$, it can easily be computed \emph{from the state-based model of the molecule}, which is a major difference of approach of that can be found in \cite{desvillettes1997modele,desvillettes2005kinetic} where $\varphi$ is computed with \emph{a posteriori} knowledge of the number of degrees of freedom of the gas, moreover usually assuming it to be temperature-independant. In the following propositions we show some examples of this reduction for some physically meaningful models presented in the previous section, the link with existing models known to be accurate at describing rotation, and the computation of a weight $\varphi$ describing rotation \emph{and} vibration, which is, as expected, \emph{not} of the form $I^{\alpha}$.

\noindent

\begin{proposition} \label{prop:reducedimension}
Let $d \in \N^*$. The model
$$
\big( \R^d,\emph{Bor}(\R^d),\emph{d}z \big), \qquad \varepsilon(z) = \frac12 \mathcal{I} \, |z|^2
$$
can be reduced to
$$
\left(\R_+,\emph{Bor}(\R_+), C_d(\mathcal{I}) \, I^{\frac{d}{2} - 1} \, \emph{d}I \right), \qquad I \mapsto I,
$$
with $C_d(\mathcal{I}) = 2^{\frac{d}{2}-1} \, \mathcal{I}^{-\frac{d}{2}} \, \emph{Leb}_{d-1}(\mathbb{S}^{d-1})$, where $\emph{Leb}_{d-1}$ is the Lebesgue measure of dimension $d-1$.
\end{proposition}

\begin{proof}
We compute $\mu^{\Bar{\varepsilon}}$. Let $B \in \text{Bor}(\R_+)$.
$$
\mu^{\Bar{\varepsilon}}(B) = \int_{\R^d} \mathbf{1}_{\frac12 \mathcal{I}|z|^2 \in B} \, \text{d}z = \text{Leb}_{d-1}(\mathbb{S}^{d-1}) \int_{\R_+}  \mathbf{1}_{\frac12 \mathcal{I} r^2 \in B} \, r^{d - 1} \text{d}r = \int_{\R_+}  \mathbf{1}_{I \in B} \, C_d(\mathcal{I}) \, I^{\frac{d}{2} - 1} \text{d}I,
$$
with $C_d(\mathcal{I}) = 2^{\frac{d}{2}-1} \, \mathcal{I}^{-\frac{d}{2}} \, \text{Leb}_{d-1}(\mathbb{S}^{d-1})$. It follows that
$$
\text{d} \mu^{\Bar{\varepsilon}}(I) = C_d(\mathcal{I}) \, I^{\frac{d}{2}-1} \, \text{d}I.
$$
\QED
\end{proof}

\begin{remark}
We saw in subsection \ref{subsection:buildmodel} that the value of the moment of inertia of the molecule $\mathcal{I}$ has no importance at equilibrium, however Proposition \ref{prop:reducedimension} shows that it should be taken into account outside equilibrium.

\noindent
We also remark that the reduced version of classical model for the rotation of a linear molecule presented in subsection \ref{subsection:buildmodel} is, with $C_2(\mathcal{I}) =  \frac{2 \pi}{\mathcal{I}}$,
$$
\left(\R_+,\text{Bor}(\R_+), C_2(\mathcal{I}) \, \text{d}I \right), \quad I \mapsto I.
$$
This model, being the the model with continuous variable $I$ and weight $\varphi(I) = C_2(\mathcal{I})$, that is the one of Borgnakke-Larsen \cite{borgnakke1975statistical}, is indeed well-known to be accurate in describing rotation for diatomic molecules. On the other hand, the reduction of the classical model for the rotation of a non-linear molecule presented in subsection \ref{subsection:buildmodel} can be harder in general, because the computation of the measure relies on the computation of the surface of a triaxial ellipsoid. Nevertheless in this case, assuming $\mathcal{I}_i = \mathcal{I}_j = \mathcal{I}$ and with $C_3(\mathcal{I}) = \frac{4 \sqrt{2} \pi}{\mathcal{I}^{3/2}}$, the model reduces to
$$
\left(\R_+,\text{Bor}(\R_+), C_3 (\mathcal{I})\, \sqrt{I} \, \text{d}I \right), \quad I \mapsto I,
$$
which is the model with continuous variable $I$ and weight $\varphi(I) = C_3 (\mathcal{I})\, \sqrt{I}$, well-known for the description of rotation for non-linear molecules.
\end{remark}

\begin{proposition} \label{prop:reductiondiscrete}
Let $Q$ be a discrete set, $(\epsilon_q)_{q \in Q} \in \R^Q$ and $(r_q)_{q \in Q} \in (\R^*_+)^Q$. We denote by $\epsilon^0 = \emph{inf}\{ \epsilon_q\}_{q \in Q}$ and we assume $\epsilon^0 \in \R$ and that for all $\beta > 0$ one has $\sum_{q \in Q} e^{-\beta (\epsilon_q - \epsilon^0)} \, r_q < \infty$. Then the model
$$
\big( Q, \mathcal{P}(Q), (r_q)_{q \in Q} \big), \qquad \varepsilon(q) = \epsilon_q
$$
can be reduced to
$$
\left(\R_+,\emph{Bor}(\R_+), \sum_{q \in Q} \, r_q \, \chi_{\epsilon_q - \epsilon^0} \right), \qquad I \mapsto I + \epsilon^0,
$$
where $\chi_{\epsilon_q - \epsilon^0}$ is the Dirac mass at $\epsilon_q - \epsilon^0$.
\end{proposition}

\begin{proof}
For $B \in \text{Bor}(\R_+)$, we have $\displaystyle \mu^{\Bar{\varepsilon}}(B) = \mu(\Bar{\varepsilon}^{\, -1} (B)) = \sum_{q \in Q} r_q \, \mathbf{1}_{\Bar{\varepsilon}(q) \in B}= \sum_{q \in Q} r_q \, \mathbf{1}_{\epsilon_q - \epsilon^0 \in B} = \sum_{q \in Q} \, r_q \, \chi_{\epsilon_q - \epsilon^0}(B)$.

\QED
\end{proof}

\begin{proposition} \label{prop:reductionofcombi}
\textbf{Reduction of a combination}. Let $(\E_1,\A_1,\mu_1), \, \varepsilon_1$ and $(\E_2,\A_2,\mu_2), \, \varepsilon_2$ be two general models. Then the combined model
$$
(\E,\A,\mu) = (\E_1 \times \E_2,\A_1 \otimes \A_2,\mu_1 \otimes \mu_2), \qquad \varepsilon(\zeta_1,\zeta_2) = \varepsilon_1 (\zeta_1) + \varepsilon_2 (\zeta_2)
$$
can be reduced to
$$
\left(\R_+,\emph{Bor}(\R_+), \mu^{\Bar{\varepsilon}_1}_1 \ast \mu^{\Bar{\varepsilon}_2}_2 \right), \qquad I \mapsto I + \varepsilon^0_1 + \varepsilon^0_2,
$$
where $\ast$ stands for the convolution of measures and we recall that $\mu^{\Bar{\varepsilon}_i}_i$, $i=1,2$ is the image measure of $\mu_i$ by $\bar{\varepsilon}_i$.
\end{proposition}

\begin{proof}
First of all, note that $\Bar{\varepsilon}(\zeta_1,\zeta_2) = \Bar{\varepsilon}_1 (\zeta_1) + \Bar{\varepsilon}_2 (\zeta_2)$. Thus, denoting by $\mu^{\Bar{\varepsilon}}$ the image measure of $\mu$ by $\Bar{\varepsilon}$, we have for $B \in \text{Bor}(\R_+)$
$$
\mu^{\Bar{\varepsilon}}(B) = \int_{\E_1} \int_{\E_2} \mathbf{1}_{\Bar{\varepsilon}_1(\zeta_1) + \Bar{\varepsilon}_2(\zeta_2) \in B} \; \text{d}\mu_1(\zeta_1) \, \text{d}\mu_2(\zeta_2) = \int_{\R_+} \int_{\R_+} \mathbf{1}_{I_1 + I_2 \in B} \; \text{d}\mu_1^{\Bar{\varepsilon}_1}(I_1) \, \text{d}\mu_2^{\Bar{\varepsilon}_2}(I_2) = \mu_1^{\Bar{\varepsilon}_1} \ast \mu_2^{\Bar{\varepsilon}_2} (B).
$$
\end{proof}

\begin{corollary} \label{cor:reduction1H19F}
The Harmonic semi-classical model for the $\prescript{1}{}{\emph{H}}\prescript{19}{}{\emph{F}}$ gas presented in Subsection \ref{subsection:buildmodel},
$$
(\E,\A,\emph{d}\mu(\zeta)) = \left(\R^2 \times \N,\, \emph{Bor}(\R^2) \otimes \mathcal{P}(\N),\, \emph{d}z \times 1 \right), \qquad \varepsilon(z,n) =  \frac12 \mathcal{I} \, |z|^2 + hc \, \nu_e \left( n + \frac12 \right),
$$
can be reduced to
$$
\left(\R_+,\emph{Bor}(\R_+), \varphi_{HF}(I) \, \emph{d} I \right), \qquad I \mapsto I + \frac12 hc \, \nu_e,
$$
with
\begin{equation}
\varphi_{HF}(I) = C_2(\mathcal{I}) \left\lceil \frac{I}{hc \, \nu_e} \right\rceil.
\end{equation}
\end{corollary}

\begin{proof}
This is a consequence of Proposition \ref{prop:reducedimension} with $d = 2$, Proposition \ref{prop:reductiondiscrete} with $Q = \N$, $r_n = 1$ and $\epsilon_n = hc \, \nu_e \left( n + \frac12 \right)$, and Proposition \ref{prop:reductionofcombi} with $(\E_1,\A_1,\text{d}\mu_1(\zeta_1)) = (\R^2,\, \text{Bor}(\R^2),\, \text{d}z)$, $\varepsilon_1(z) = \frac12 \mathcal{I} \, |z|^2$ and $(\E_2,\A_2,\text{d}\mu_2(\zeta_2)) = (\N,\,  \mathcal{P}(\N),\,  1)$, $\varepsilon_2(n) = hc \, \nu_e \left( n + \frac12 \right)$.
\QED
\end{proof}

\noindent This Corollary \ref{cor:reduction1H19F} illustrates the use of the reduction process: the model with continuous variable $I$ and weight $\varphi_{HF}$ corresponds to the model (classical for rotation, Harmonic quantum for vibration) for the $\prescript{1}{}{\text{H}}\prescript{19}{}{\text{F}}$ gas. The reader shall note that $\varphi_{HF}$ is \emph{not of the form} $I^{\alpha}$. As a matter of illustration, we plot on Fig. \ref{fig:varphis} the function $\varphi_{HF}$.

\begin{figure}[ht]
    \centering
    \includegraphics[width=.7\linewidth]{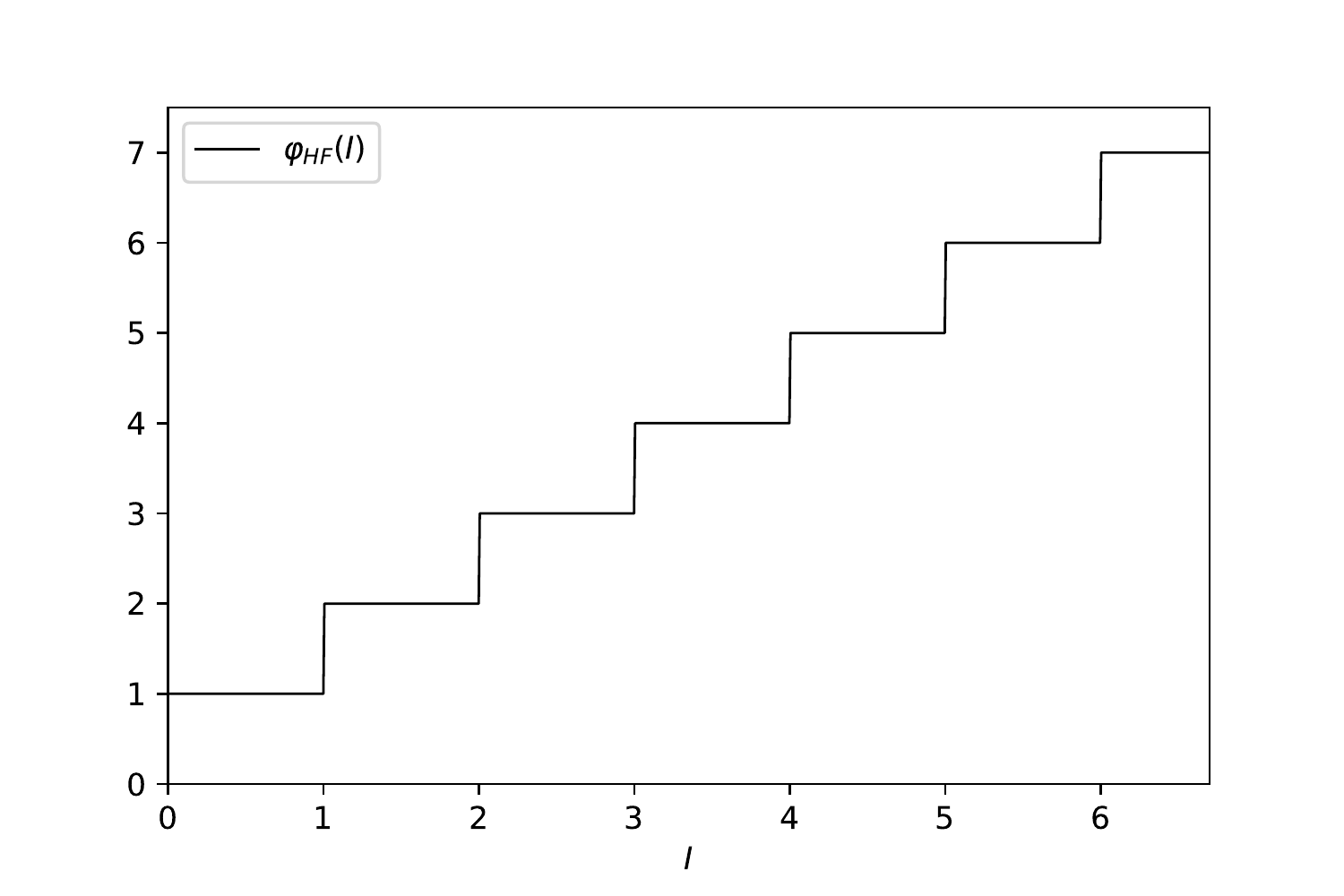}
    \caption{Plot of $\varphi_{HF}$, re-scaled with $C_2(\mathcal{I}) \equiv 1$ and $hc \, \nu_e \equiv 1$.}
    \label{fig:varphis}
\end{figure}

\begin{corollary}
\textbf{Reduction of the combination of the continuous and discrete models.} More generally, let $\varphi_0 \in L^1_{loc}(\R_+,\emph{d}I)$ be non-negative and non identically zero,  $Q$ a discrete set and $(r_q)_{q \in Q} \in (\R^*_+)^Q$. We denote by $\epsilon^0 = \emph{inf}\{ \epsilon_q\}_{q \in Q}$ and we assume $\epsilon^0 \in \R$ and that for all $\beta > 0$ one has $\sum_{q \in Q} e^{-\beta (\epsilon_q - \epsilon^0)} \, r_q < \infty$. Then the model
$$
\big( \R_+ \times Q, \, \emph{Bor}(\R_+) \otimes \mathcal{P}(Q), \, \varphi_0(I) \, \emph{d}I \, r_q \big), \qquad \varepsilon(I, q) = I + \epsilon_q
$$
can be reduced to
$$
\left(\R_+, \, \emph{Bor}(\R_+), \, \varphi(I) \, \emph{d} I \right), \qquad I \mapsto I + \epsilon^0,
$$
with
\begin{equation} \label{eq:formulavarphi}
\varphi(I) = \sum_{q \in Q}\mathbf{1}_{\epsilon_q \leq I + \epsilon^0} \; \; r_q \; \varphi_0\left( I + \epsilon^0  - \epsilon_q  \right).
\end{equation}

\end{corollary}

\begin{proof}
Again, this result directly follows from Propositions \ref{prop:reductiondiscrete} and \ref{prop:reductionofcombi}. \QED
\end{proof}

\begin{remark}
Still from Proposition \ref{prop:reductionofcombi}, the model with two continuous variables with weights 
$$
\big( \R_+ \times \R_+, \, \text{Bor}(\R_+) \otimes \text{Bor}(\R_+), \, \varphi_1(I_1) \, \text{d}I_1 \, \varphi_2 (I_2) \, \text{d}I_2  \big), \qquad \varepsilon(I_1,I_2) = I_1 + I_2,
$$
can be reduced to the continuous model with weight $\varphi_1 \ast \varphi_2$, where $\displaystyle \varphi_1 \ast \varphi_2(I) = \int_0^I \varphi_1(I') \, \varphi_2(I-I') \, \text{d}I'$.
\end{remark}

\begin{remark}
If we denote by $\delta_{\mu^{\Bar{\varepsilon}}}$ the number of internal degrees of freedom associated with the general model $(\E,\A,\mu)$, $\varepsilon$, we get $\displaystyle  \delta_{\mu_1^{\Bar{\varepsilon}_1} \ast \mu_2^{\Bar{\varepsilon}_2}} = \delta_{\mu_1^{\Bar{\varepsilon}_1}} + \delta_{\mu_2^{\Bar{\varepsilon}_2}}$ from the Equipartition Theorem \ref{theorem:equipartition}. Thus $\mu^{\Bar{\varepsilon}} \mapsto \delta_{\mu^{\Bar{\varepsilon}}}$ is a morphism.
\end{remark}

\medskip

\subsubsection*{A comment on numerical aspects}

\noindent Not only the process of model reduction links the \emph{state}-based and \emph{energy}-based points of view, it is also useful for numerical considerations. For instance, the space of internal states of the semi-classical model for $\prescript{1}{}{\text{H}}\prescript{19}{}{\text{F}}$ is $\R^2 \times \N$, whereas it is $\R_+$ for any reduced model, which is far more efficient to be used for numerical simulations. In order to construct an approximate model suited and efficient for simulation, one wishes to approximate the measure $\mu^{\Bar{\varepsilon}}$ by a discrete measure, a finite sum of Diracs, of the form $\displaystyle \sum_{k=1}^N r_k \, \chi_{\Bar{\epsilon}_k}$. Then the numerically-suited approximation of the general model
$$
(\E,\A,\mu),\qquad \varepsilon
$$
is thus the discrete model
$$
(\llbracket 1, N \rrbracket, \mathcal{P}(\llbracket 1, N \rrbracket), (r_k)_{1 \leq k \leq N}), \qquad (\Bar{\epsilon}_k + \varepsilon^0)_{1\leq k \leq N},
$$
At fixed $N$, we suggest to choose the $(r_k,\bar{\epsilon}_k)_{1 \leq k \leq N}$ that minimizes a distance (for instance the 1 or 2-Wassertein distance, explicit in 1D) between the Gibbs probability measures associated with the measures $\mu^{\Bar{\varepsilon}}$ and $\displaystyle \sum_{k=1}^N r_k \, \chi_{\Bar{\epsilon}_k}$ at a temperature coherent with the considered problem. In \cite{magin2012coarse}, Magin et al. propose to simplify a discrete model for dinitrogen ($N_2$) composed of 9390 rotation-vibration energy-levels by creating energy bins, and replace the first model with an approximate discrete one composed of only 500 energy levels, which corresponds to approximating a sum a 9390 Diracs by a sum of 500 Diracs. This binning method can be extended to the general case. Like in section 2.3. of \cite{magin2012coarse}, consider a family of disjoint compact intervals of $\R_+$, $(R_k)_{1 \leq k \leq N}$, and define the average energy of the $k$-th bin by
\begin{equation} \label{eq:binenergy}
\epsilon_k := \varepsilon^0 + \frac{1}{r_k} \, \int_{R_k} I \, \text{d} \mu^{\bar{\varepsilon}}(I),
\end{equation}
where $r_k := \mu^{\bar{\varepsilon}}(R_k)$ is the degeneracy associated with this energy level, with $R_k$ such that $r_k$ is not equal to zero. The choice of $(R_k)_{1 \leq k \leq N}$ can be arbitrary, or made by minimizing a distance, like we suggested earlier.

\section{Conclusions} \label{concl}

In this paper we have built up a general framework for the kinetic modelling of non-relativistic mono and polyatomic gases. It is based on a set of allowed internal states $\mathcal{E}$ endowed with a suitable measure $\mu$. Each particle is characterized by a proper state $\zeta \in \mathcal{E}$ with associated energy $\varepsilon(\zeta)$. Owing only to conservations of momentum and total energy, we are able to define the collision rule and the corresponding Boltzmann operator leaving $\mathcal{E}$, $\mu$ and $\varepsilon(\zeta)$ generic (not explicit). The Boltzmann H-Theorem has been proved in this general setting, and Maxwellian equilibria (depending also on internal energy) have been explicitly recovered. Also the number of internal degrees of freedom has been investigated, and the fluid--dynamic Euler equations have been derived.
We have shown that usual models such as the monoatomic gas, the continuous internal energy description with weights \cite{desvillettes2005kinetic}, and the discrete energy levels description \cite{groppi1999kinetic,Giovangigli} fit into our framework. Moreover, several different models may be built within the present general framework, as semi-classical models and quantum description (as main example, we have proposed four models for the hydrogen fluoride).

The main advantage of this general setting is to be able to construct a model from direct physical considerations, since we consider \emph{internal states} and no longer directly \emph{energy}. For instance, instead of considering from the beginning the rotational energy of a molecule, we may start from the angular velocity and construct the energy by the laws of classical mechanics (owing to inertia tensor). Moreover, thanks to the Equipartition Theorem, we are allowed to combine any pair of different internal states spaces $\mathcal{E}_1 \times \mathcal{E}_2$ and the corresponding measures as $\mu_1 \otimes \mu_2$ thus, as an important consequence, we have the possibility to keep separate the vibrational and the rotational parts of the internal energy of polyatomic molecules. Indeed, setting the space of internal states as $\mathcal{E} = \mathbb{R}_+ \times \mathbb{N}$ and $\varepsilon(I,n) = I + e_n$, we are able to describe the rotational energy by means of a continuous variable, and the vibrational one by a discrete energy variable, as suggested in \cite{Herzberg}.
Analogously, also the options of keeping both kinds of internal energies continuous (as in Extended Thermodynamics) or discrete are admissible.

Finally we have also shown that, generically, all models that fit our framework can be reduced to a one-real-parameter model, similar to the continuous internal energy model with weights, at the price of suitably changing, in a rigorous way, the integration measure. In this reduction procedure the considerations on states turn out to be summed up into the energy (the reduced model considers the internal energy directly). From theory to simulation, given a molecule to study, we suggest to first construct the \emph{state-based} (general) model of the molecule from its physical description, then compute the associated energy law $\mu^{\Bar{\varepsilon}}$ by performing the reduction process and finally define the numerically-suited approximate discrete model as detailed in Section \ref{section8}. Also for the investigation of other interesting mathematical properties of the Boltzmann operator in this general frame, as the validity of the Fredholm alternative for the linearized operator, and the corresponding rigorous Chapman-Enskog asymptotic expansion up to Navier-Stokes equations, the original internal states formulation could be more intuitive to use.
We finally remark that even the combination of the continuous and the discrete energy model finally turns out to be, after the reduction, a continuous model with a weight explicitly computed in the paper, and this is a quite surprising result.

Of course it would be desirable to include in the kinetic model even more physical features of polyatomic particles. For instance, the quantum mechanical Boltzmann equation derived by Waldmann \cite{Wa57} and Snider~\cite{Sn60} is able to describe also the polarizations resulting from the effect of external fields on polyatomic gases \cite{MBKK90}. Such model admits two vectorial collision invariants, corresponding to momentum and angular momentum, bearing in mind that polyatomic particles are generally non spherical, and also the corresponding macroscopic equations include a proper angular momentum conservation equation \cite{MS64,FK72,MBKK90}. It is well known \cite{MS64} that in absence of polarization effects, namely for isotropic gases, the quantum mechanical theory yields the same formal results as classical or semi-classical approaches considered in this paper, but a general kinetic framework able to include also possible polarization (and therefore additional collision invariants) could be an interesting further step in kinetic investigation of polyatomic gases.

We aim also at extending our general way of modelling to mixtures of polyatomic gases, possibly undergoing chemical reactions. Of course collision rules, and consequently some technical parts of the proofs, would be much more complicated due to the presence of mass ratios and of the amount of energy produced or consumed by chemical reactions. The investigation of a suitable general framework for gas mixtures will be the subject of a future work.

\section*{Appendix}

\begin{lemma} \label{lemma:equiv}
Let $r>0$, assume that $\mu(\{ 0 < \Bar{\varepsilon} \leq r \}) > 0$. Then
\begin{align*}
\int_{\E} \Bar{\varepsilon}(\zeta) e^{-x \Bar{\varepsilon}(\zeta)} \, \emph{d}\mu(\zeta) &\underset{x \to \infty}{=} \int_{\Bar{\varepsilon} \leq r } \Bar{\varepsilon}(\zeta) e^{-x \Bar{\varepsilon}(\zeta)} \, \emph{d}\mu(\zeta) +  o \left( \int_{\Bar{\varepsilon} \leq r} \Bar{\varepsilon}(\zeta) e^{-x \Bar{\varepsilon}(\zeta)} \, \emph{d}\mu(\zeta) \right) \\
\text{and} \quad \int_{\E}  e^{-x \Bar{\varepsilon}(\zeta)} \, \emph{d}\mu(\zeta) &\underset{x \to \infty}{=} \int_{\Bar{\varepsilon} \leq r } e^{-x \Bar{\varepsilon}(\zeta)} \, \emph{d}\mu(\zeta) +  o \left( \int_{\Bar{\varepsilon} \leq r} e^{-x \Bar{\varepsilon}(\zeta)} \, \emph{d}\mu(\zeta) \right).
\end{align*}
\end{lemma}

\bigskip

\begin{proof}
First, we have that
\begin{align*}
\int_{\E} \Bar{\varepsilon}(\zeta) e^{-x \Bar{\varepsilon}(\zeta)} \, \text{d}\mu(\zeta) &= \int_{\Bar{\varepsilon} \leq r } \Bar{\varepsilon}(\zeta) e^{-x \Bar{\varepsilon}(\zeta)} \, \text{d}\mu(\zeta) +  \int_{\Bar{\varepsilon} > r } \Bar{\varepsilon}(\zeta) e^{-x \Bar{\varepsilon}(\zeta)} \, \text{d}\mu(\zeta), \\
\text{and} \quad \int_{\E}  e^{-x \Bar{\varepsilon}(\zeta)} \, \text{d}\mu(\zeta) &= \int_{\Bar{\varepsilon} \leq r } e^{-x \Bar{\varepsilon}(\zeta)} \, \text{d}\mu(\zeta) +  \int_{\Bar{\varepsilon} > r } e^{-x \Bar{\varepsilon}(\zeta)} \, \text{d}\mu(\zeta).
\end{align*}
Thus we must prove
\begin{align*}
\int_{\Bar{\varepsilon} > r } \Bar{\varepsilon}(\zeta) e^{-x \Bar{\varepsilon}(\zeta)} \, \text{d}\mu(\zeta) &\underset{x \to \infty}{=}  o \left( \int_{\Bar{\varepsilon} \leq r} \Bar{\varepsilon}(\zeta) e^{-x \Bar{\varepsilon}(\zeta)} \, \text{d}\mu(\zeta) \right), \\
\int_{\Bar{\varepsilon} > r } e^{-x \Bar{\varepsilon}(\zeta)} \, \text{d}\mu(\zeta)&\underset{x \to \infty}{=}  o \left( \int_{\Bar{\varepsilon} \leq r} e^{-x \Bar{\varepsilon}(\zeta)} \, \text{d}\mu(\zeta) \right).
\end{align*}
First, note that
$$
\int_{\Bar{\varepsilon} \leq r} e^{-x \Bar{\varepsilon}(\zeta)} \, \text{d}\mu(\zeta) \geq \mu(\{\Bar{\varepsilon} \leq r\}) \, e^{-x r}.
$$
Since $\mu(\{\Bar{\varepsilon} \leq r\}) > 0$, we deduce that $\int_{\Bar{\varepsilon} \leq r} e^{-x \Bar{\varepsilon}(\zeta)} \, \text{d}\mu(\zeta) > 0$ and
$$
0 \leq \frac{\int_{\Bar{\varepsilon} > r } e^{-x \Bar{\varepsilon}(\zeta)} \, \text{d}\mu(\zeta)}{\int_{\Bar{\varepsilon} \leq r} e^{-x \Bar{\varepsilon}(\zeta)} \, \text{d}\mu(\zeta)} \leq \frac{e^{x r}}{\mu(\{\Bar{\varepsilon} \leq r\})} \, \int_{\Bar{\varepsilon} > r } e^{-x \Bar{\varepsilon}(\zeta)} \, \text{d}\mu(\zeta) = \frac{1}{\mu(\{\Bar{\varepsilon} \leq r\})} \, \int_{\Bar{\varepsilon} > r } e^{-x (\Bar{\varepsilon}(\zeta)-r)} \, \text{d}\mu(\zeta).
$$
Now note that, on the set $\{\Bar{\varepsilon} > r\}$,
$$
e^{-x (\Bar{\varepsilon}(\zeta)-r)} \underset{x \to \infty}{\longrightarrow} 0.
$$
Also, $(\zeta \mapsto e^{-x (\Bar{\varepsilon}(\zeta)-r)})_{x > 0}$ is a non-increasing family of non-negative functions on the set $\{\Bar{\varepsilon} > r\}$. By monotone convergence theorem,
$$
\int_{\Bar{\varepsilon} > r} e^{-x (\Bar{\varepsilon}(\zeta)-r)} \, \text{d}\mu(\zeta) \underset{x \to \infty}{\longrightarrow} 0.
$$
We deduce that
$$
\frac{\int_{\Bar{\varepsilon} > r } e^{-x \Bar{\varepsilon}(\zeta)} \, \text{d}\mu(\zeta)}{\int_{\Bar{\varepsilon} \leq r} e^{-x \Bar{\varepsilon}(\zeta)} \, \text{d}\mu(\zeta)}\underset{x \to \infty}{\longrightarrow} 0,
$$
that is our goal
$$
\int_{\Bar{\varepsilon} > r } e^{-x \Bar{\varepsilon}(\zeta)} \, \text{d}\mu(\zeta) \underset{x \to \infty}{=} o \left( \int_{\Bar{\varepsilon} \leq r} e^{-x \Bar{\varepsilon}(\zeta)} \, \text{d}\mu(\zeta) \right).
$$
Concerning the other goal, it obviously holds
$$
\int_{\Bar{\varepsilon} \leq r} \Bar{\varepsilon}(\zeta) e^{-x \Bar{\varepsilon}(\zeta)} \, \text{d}\mu(\zeta) \geq e^{-x r} \, \int_{\Bar{\varepsilon} \leq r} \Bar{\varepsilon}(\zeta) \, \text{d}\mu(\zeta),
$$
and, since $\mu(\{0 < \Bar{\varepsilon} \leq r\}) > 0$, we deduce that $\int_{\Bar{\varepsilon} \leq r} \Bar{\varepsilon}(\zeta) \, \text{d}\mu(\zeta) > 0$. Moreover,
$$
\int_{\Bar{\varepsilon} \leq r} \Bar{\varepsilon}(\zeta) \, \text{d}\mu(\zeta) \leq r \mu(\{ \Bar{\varepsilon} \leq r \}) < \infty
$$
(last inequality has been shown in Section \ref{section:presentationmodel}). It follows that
$$
0 \leq \frac{\int_{\Bar{\varepsilon} > r } \Bar{\varepsilon}(\zeta) e^{-x \Bar{\varepsilon}(\zeta)}  \text{d}\mu(\zeta)}{\int_{\Bar{\varepsilon} \leq r} \Bar{\varepsilon}(\zeta) e^{-x \Bar{\varepsilon}(\zeta)}  \text{d}\mu(\zeta)} \leq \frac{e^{x r}}{\int_{\Bar{\varepsilon} \leq r} \Bar{\varepsilon}(\zeta) \, \text{d}\mu(\zeta)} \, \int_{\Bar{\varepsilon} > r } \Bar{\varepsilon}(\zeta) e^{-x \Bar{\varepsilon}(\zeta)} \, \text{d}\mu(\zeta) = \frac{\int_{\Bar{\varepsilon} > r } \Bar{\varepsilon}(\zeta) e^{-x (\Bar{\varepsilon}(\zeta)-r)} \, \text{d}\mu(\zeta)}{\int_{\Bar{\varepsilon} \leq r} \Bar{\varepsilon}(\zeta) \, \text{d}\mu(\zeta)}.
$$
Now note that, on the set $\{\Bar{\varepsilon} > r\}$,
$$
\Bar{\varepsilon}(\zeta) e^{-x (\Bar{\varepsilon}(\zeta)-r)} \underset{x \to \infty}{\longrightarrow} 0.
$$
Also, $(\zeta \mapsto \Bar{\varepsilon}(\zeta) e^{-x (\Bar{\varepsilon}(\zeta)-r)})_{x > 0}$ is a non-increasing family of non-negative functions on the set $\{\Bar{\varepsilon} > r\}$. By monotone convergence theorem,
$$
\int_{\Bar{\varepsilon} > r} \Bar{\varepsilon}(\zeta) e^{-x (\Bar{\varepsilon}(\zeta)-r)} \, \text{d}\mu(\zeta) \underset{x \to \infty}{\longrightarrow} 0.
$$
We deduce that
$$
\frac{\int_{\Bar{\varepsilon} > r } \Bar{\varepsilon}(\zeta) e^{-x \Bar{\varepsilon}(\zeta)} \, \text{d}\mu(\zeta)}{\int_{\Bar{\varepsilon} \leq r} \Bar{\varepsilon}(\zeta) e^{-x \Bar{\varepsilon}(\zeta)} \, \text{d}\mu(\zeta)}\underset{x \to \infty}{\longrightarrow} 0,
$$
that is
$$
\int_{\Bar{\varepsilon} > r } \Bar{\varepsilon}(\zeta) e^{-x \Bar{\varepsilon}(\zeta)} \, \text{d}\mu(\zeta) \underset{x \to \infty}{=} o \left( \int_{\Bar{\varepsilon} \leq r} \Bar{\varepsilon}(\zeta) e^{-x \Bar{\varepsilon}(\zeta)} \, \text{d}\mu(\zeta) \right).
$$
\QED
\end{proof}

\end{document}